\documentclass[prx,a4paper,aps,twocolumn,superscriptaddress,longbibliography]{revtex4-1}

\usepackage[utf8]{inputenc}
\usepackage[T1]{fontenc}
\usepackage{amsmath, amsthm, amssymb, mathtools, dsfont}
\usepackage{times}
\usepackage{graphicx}
\usepackage{dcolumn}
\usepackage{bm,bbm}
\usepackage{color}
\usepackage[usenames,dvipsnames]{xcolor}
\usepackage{esint}
\usepackage{paralist}

\usepackage{geometry}
\definecolor{myurlcolor}{rgb}{0,0,0.7}
\definecolor{myrefcolor}{rgb}{0.8,0,0}
\usepackage[unicode=true,pdfusetitle, bookmarks=true,bookmarksnumbered=false,bookmarksopen=false, breaklinks=false,pdfborder={0 0 0},backref=false,colorlinks=true, linkcolor=myrefcolor,citecolor=myurlcolor,urlcolor=myurlcolor]{hyperref}

\DeclareGraphicsExtensions{.png,.pdf,.eps}
\graphicspath{ {./figs/} }
\usepackage{epstopdf}

\makeatletter

 \theoremstyle{plain}
 \newtheorem{fact}{Fact}
 \theoremstyle{plain}
 \newtheorem{lem}{Lemma}
 \theoremstyle{plain}
 \newtheorem{thm}{Theorem}
 \theoremstyle{plain}
 \newtheorem{exa}{Example}
 \theoremstyle{plain}
 \newtheorem{corr}{Corollary}
 \theoremstyle{plain}
 
 \theoremstyle{remark}
 \newtheorem*{rem*}{Remark}
	\theoremstyle{plain}
	 \newtheorem{rem}{Remark}

\renewcommand{\cdot}{\,} 
\newcommand{\e}{\mathrm{e}}
\renewcommand{\d}{\mathrm{d}}
\newcommand{\param}{\varphi}
\newcommand{\1}{\mathds{1}}
\newcommand{\ot}{\otimes}

\newcommand{\ket}[1]{| #1 \rangle}
\newcommand{\bra}[1]{\langle #1 |}
\newcommand{\braket}[2]{\langle #1|#2\rangle}
\newcommand{\ketbra}[2]{\ket{#1}\!\bra{#2}}
\newcommand{\proj}[1]{\ketbra{#1}{#1}}
\renewcommand{\t}[1]{\mathrm{#1}}
\newcommand{\ii}{\mathrm{i}}
\DeclareMathOperator{\tr}{tr}
\DeclareMathOperator{\dB}{\mathnormal{d}_B\!}

\DeclareMathOperator{\QFI}{\mathnormal{F}\!}
\DeclareMathOperator{\FI}{\mathnormal{F}_\t{cl}}
\DeclareMathOperator{\lspan}{span}
\renewcommand{\H}{\mathcal{H}}
\DeclareMathOperator*{\Hloc}{\mathcal{H}_{loc}}
\newcommand{\dloc}{d}
\newcommand{\Hsym}{\mathcal{S}}
\newcommand{\R}{\mathbb{R}}
\newcommand{\SU}{\mathrm{SU}}
\DeclareMathOperator*{\Expect}{\mathbb{E}}

\newcommand{\landauO}{\mathcal{O}}

\newcommand{\eref}[1]{\eqref{#1}}
\newcommand{\eqnref}[1]{Eq.~\eqref{#1}}
\newcommand{\eqnsref}[2]{Eqs.~\eqref{#1} and \eqref{#2}}
\newcommand{\figref}[1]{Fig.~\ref{#1}}

\newcommand{\secref}[1]{Section~\ref{#1}}
\newcommand{\appref}[1]{Appendix~\ref{#1}}
\newcommand{\lemref}[1]{Lemma~\ref{#1}}
\newcommand{\thmref}[1]{Theorem~\ref{#1}}
\newcommand{\corrref}[1]{Corollary~\ref{#1}}
\newcommand{\exaref}[1]{Example~\ref{#1}}
\newcommand{\remref}[1]{Remark~\ref{#1}}
\newcommand{\factref}[1]{Fact~\ref{#1}}
\newcommand{\refcite}[1]{Ref.~\cite{#1}}

\global\long\global\long\global\long\def\bra#1{\mbox{\ensuremath{\langle#1|}}}
\global\long\global\long\global\long\def\ket#1{\mbox{\ensuremath{|#1\rangle}}}

\global\long\global\long\global\long\def\kb#1#2{\mbox{\ensuremath{\ensuremath{\ensuremath{|#1\rangle\!\langle#2|}}}}}

\makeatother

\begin{document}

\title{Random bosonic states for robust quantum metrology}

\author{M. Oszmaniec}
\affiliation{ICFO-Institut de Ciencies Fotoniques, The Barcelona Institute of Science and Technology, 08860 Castelldefels (Barcelona), Spain}
\email{michal.oszmaniec@icfo.es} 

\author{R. Augusiak}
\affiliation{ICFO-Institut de Ciencies Fotoniques, The Barcelona Institute of Science and Technology, 08860 Castelldefels (Barcelona), Spain}
\affiliation{Center for Theoretical Physics, Polish Academy of Sciences, Aleja Lotnik\'ow 32/46, 02-668 Warsaw, Poland}

\author{C. Gogolin}
\affiliation{ICFO-Institut de Ciencies Fotoniques, The Barcelona Institute of Science and Technology, 08860 Castelldefels (Barcelona), Spain}
\affiliation{Max-Planck-Institut f{\"u}r Quantenoptik, Hans-Kopfermann-Stra{\ss}e 1, 85748 Garching, Germany}

\author{J. Ko\l{}ody\'{n}ski}
\affiliation{ICFO-Institut de Ciencies Fotoniques, The Barcelona Institute of Science and Technology, 08860 Castelldefels (Barcelona), Spain}

\author{A. Ac\'in}
\affiliation{ICFO-Institut de Ciencies Fotoniques, The Barcelona Institute of Science and Technology, 08860 Castelldefels (Barcelona), Spain}
\affiliation{ICREA-Instituci\'o Catalana de Recerca i Estudis Avan\c cats, Lluis Companys 23, 08010 Barcelona, Spain}

\author{M. Lewenstein}
\affiliation{ICFO-Institut de Ciencies Fotoniques, The Barcelona Institute of Science and Technology, 08860 Castelldefels (Barcelona), Spain}
\affiliation{ICREA-Instituci\'o Catalana de Recerca i Estudis Avan\c cats, Lluis Companys 23, 08010 Barcelona, Spain}

\begin{abstract}
We study how useful random states are for quantum metrology, i.e., whether they surpass the classical limits imposed on precision in the canonical phase sensing scenario.
First, we prove that random pure states drawn from the Hilbert space of distinguishable particles typically do not lead to super-classical scaling of precision even when allowing for local unitary optimization.
Conversely, we show that random pure states from the symmetric subspace typically achieve the optimal Heisenberg scaling without the need for local unitary optimization.
Surprisingly, the Heisenberg scaling is observed for random isospectral states of arbitrarily low purity and preserved under loss of a fixed number of particles. Moreover, we prove that for pure states a standard photon-counting interferometric measurement suffices to typically achieve resolutions following the Heisenberg scaling for all values of the phase at the same time. Finally, we demonstrate that metrologically useful states can be prepared with short random optical circuits generated from three types of beam-splitters and a single non-linear (Kerr-like) transformation.
\end{abstract}

\maketitle

\section{Introduction}
Quantum metrology opens the possibility to exploit quantum features to measure unknown physical quantities with accuracy surpassing the constraints dictated by classical physics \cite{Giovannetti2011,Toth2014,Demkowicz2015,Dowling2015}.
Classically, by employing $N$ probes to independently sense a parameter, the mean squared error of estimation scales at best as $1/N$.
This resolution is also known as the \emph{Standard Quantum Limit} (SQL) \cite{Giovannetti2004}. Quantum mechanics, however, allows one to engineer entangled states of $N$ particles which, when used as probes, can lead to resolutions beyond the SQL. Crucially, in the canonical phase sensing
scenario a precision scaling like $1/N^{2}$, known as the \emph{Heisenberg Limit} (HL) \cite{Giovannetti2006}, may be attained. In practice, the destructive impact of noise must also be taken into account \cite{Huelga1997,Escher2011,Demkowicz2012}, but quantum-enhanced resolutions have been successfully observed in optical interferometry \cite{Mitchell2004,Nagata2007} (including gravitational-wave detection \cite{LIGO2011,*LIGO2013}), ultracold ion spectroscopy \cite{Leibfried2004,Roos2006}, atomic magnetometry \cite{Wasilewski2010,Sewell2012}, and in entanglement-assisted atomic clocks \cite{Appel2009,Louchet2010}.

A fundamental question is to understand which quantum states offer an advantage for quantum metrology. 
Quantum-enhanced parameter sensitivity may only be observed with systems exhibiting inter-particle entanglement \cite{Pezze2009}, thus, 
such enhanced sensitivity can be used to detect multipartite entanglement \cite{Riedel2010,Krischek2011,Strobel2014,Lucke2014} and 
lower-bound the number of particles being entangled \cite{Sorensen2001,Hyllus2012,*Toth2012}.
However, the precise connection between entanglement and a quantum metrological advantage is so far not fully understood.

It is known that states achieving the optimal sensitivity must have
entanglement between all their particles \cite{Hyllus2012,*Toth2012}, like for example the
Greenbegrer-Horne-Zeilinger (GHZ) state (equivalent to the N00N state
in optical interferometry), yet there also exist classes of such
states which are useless from the metrological perspective
\cite{Hyllus2010a}. The optimal states, however, belong to the symmetric (bosonic)
subspace, from which many states have been recognized to
offer a significant advantage in quantum metrology \cite{Hyllus2010a,Apellaniz2015,Benatti2013}. On the other hand, a very weak form of entanglement---so-called undistillable entanglement---may lead to Heisenberg scaling \cite{Czekaj2015}, while any super-classical scaling arbitrarily close to the HL ($1/N^{2-\epsilon}$ with $\epsilon>0$) can be achieved with states whose geometric measure of entanglement vanishes in the limit $N \to \infty$ \cite{Augusiak2016}.

\begin{figure}[!t]
\centering
\includegraphics[width=\columnwidth]{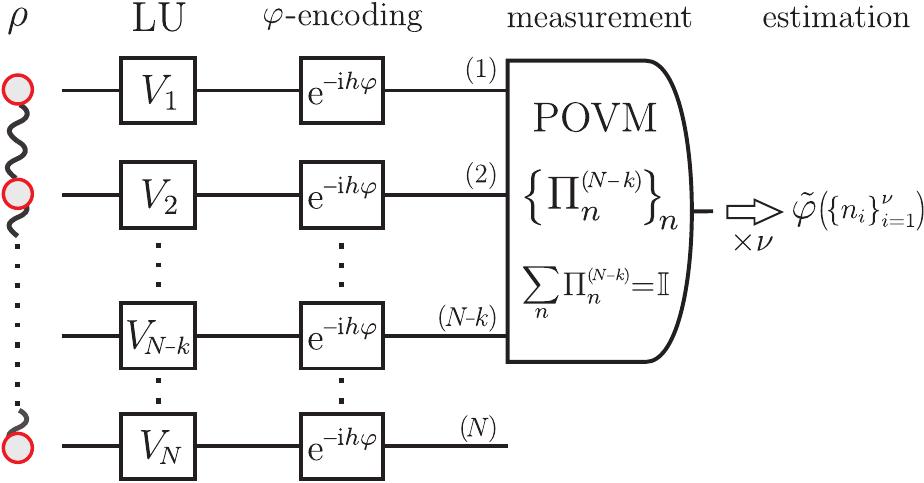}
\caption{(color online)
LU-optimised lossy quantum metrology protocol:~%
A given state $\rho$ is adjusted with local unitary (LU) operations $V_1,\dots,V_N$ to probe as 
precisely as possible small fluctuations of the parameter $\param$, which is independently 
and unitarily encoded into each of the $N$ constituent particles.
Finally, only $N-k$ particles are measured, reflecting the possibility of losing $k$ of them.
The most general measurement is then described by a positive-operator-valued measure (POVM), i.e., a collection of positive semi-definite operators $\{\Pi_n^{N-k}\}_n$ which act on the remaining $N-k$ particles and sum up to the identity.
The process is repeated $\nu$ times, in order to construct the most sensitive estimate of the parameter, $\tilde\param$, based on the measurement outcomes $\{n_i\}_{i=1}^\nu$.}
\label{fig:protocol}
\end{figure}

Here, we go beyond merely presenting examples of states leading to quantum-enhanced precision.
Instead, we conduct a systematic study by analyzing typical properties of the \emph{quantum} and \emph{classical Fisher information} on various ensembles of quantum states.
First, we show that states of distinguishable particles typically are \emph{not} useful for metrology, despite having a large amount of 
entanglement as measured by the entanglement entropy \cite{Page1993,Foong1994,Sen1996,Lewenstein2012} and various other measures \cite{Hayden2006,Gross2009,Oszmaniec2014,moPHD}. On the contrary, we show most pure random states from the 
\emph{symmetric} (bosonic) subspace of any local dimension achieve resolutions at the HL.
Moreover, we prove that the usefulness of random symmetric states is robust against loss of a finite number of particles, and holds also for mixed states with fixed spectra (as long as the distance from the maximally mixed state in the symmetric subspace is sufficiently large).
This is in stark contrast to the case of GHZ states, which completely lose their (otherwise ideal) phase sensitivity upon loss of just a single particle.
Third, we show that, even for a fixed measurement, random pure bosonic states typically allow to sense the phase at the HL.
Concretely, this holds in the natural quantum optics setting of photon-number detection in output modes of a balanced beam-splitter \cite{Bachor2004} and independently of the value of the parameter.
Finally, we demonstrate that states generated using random circuits with gates from a universal gate-set on the symmetric subspace consisting only of beam-splitters and a single non-linear Kerr-like transformation also typically  achieve Heisenberg scaling---again even for a fixed measurement.
As all our findings also equally apply to standard atomic interferometry \cite{Schumm2005,Sebby2007,Zhang2012,Stroescu2015}.
Our work shows that metrological usefulness is a more generic feature than previously thought and opens up new possibilities for quantum-enhanced metrology based on random states.

Lastly, let us note that, as metrological usefulness of quantum states is tantamount 
to the notion of state \emph{macroscopicity} \cite{Froewis2012a}, our results directly apply in this context \cite{Froewis2012,Tichy2015}.
Moreover, since states attaining HL can be used to approximately clone $N$ quantum gates into as many as $N^2$ gates as $N \to \infty$, one can immediately use our findings to also infer that typical symmetric states provide a resource allowing for optimal \emph{asymptotic replication of unitary gates} \cite{Chiribella2013,Duer2015}.

Our results are based on leveraging recent insights concerning the \emph{continuity of quantum Fisher information} \cite{Augusiak2016}, \emph{measure concentration techniques} \cite{Hayden2006,Ledoux2005,Anderson2010,Collins2015}, lately proven results about the \emph{spectral gap in the special unitary group} \cite{Bourgain2011}, as well as the theory of \emph{approximate $t$-designs} \cite{Low09,Low2010a,Brandao2012}.

{Our work sheds new light on the role of symmetric states in 
quantum metrology \cite{Hyllus2010a,Benatti2013,Apellaniz2015,Froewis2014,Tichy2015}. 
In particular, it clarifies the usefulness of symmetric states from the typicality perspective \cite{Hyllus2010a}, 
but also analytically confirms the findings about their typical properties previously suggested by numerical compuations \cite{Froewis2014,Tichy2015}.}

The remainder of the paper is organized as follows.
In \secref{sec:setting} we introduce the setting of quantum parameter estimation including the classical and quantum Fisher information and their operational interpretation.
In \secref{sec:notat} we familiarize the reader with the technique of measure concentration and introduce some additional notation.
In \secref{sec:disting} we present our results on the lack of usefulness of random states of distinguishable particles for quantum metrology.
Subsequently, in \secref{sec:symuseful} we show that states from the symmetric (bosonic) subspace typically attain the HL. Following that, in \secref{sec:robust}, we analyze the robustness of metrological usefulness of such states under noise and loss of particles.
In \secref{sec:phys} we turn our attention to the classical Fisher information in a concrete measurement setup.
We prove that with random pure symmetric states Heisenberg scaling can be typically achieved with a physically accessible measurement setup---essentially a Mach-Zehnder interferometer with particle number detectors.
Finally, in \secref{sec:numerics}, we demonstrate that symmetric states whose metrological properties effectively mimic those of random symmetric states can be prepared by short random circuits generated from a universal gate-set in the symmetric subspace.
We conclude our work in \secref{sec:conclusions}.

\section{Quantum metrology}\label{sec:setting}
We consider the canonical \emph{phase sensing} scenario of quantum metrology 
\cite{Giovannetti2006}, in which one is given $N$ devices (black boxes) that 
encode a phase-like parameter $\param$. The task is to determine the optimal strategy 
of preparing quantum probes, so that after interaction with the devices the probes can be 
measured in a way that reveals highest sensitivity to fluctuations of the parameter $\param$.
Crucially, in quantum mechanics one has the freedom to apply the devices ``in parallel'' to a 
(possibly entangled) state $\rho$ of $N$ particles---see \figref{fig:protocol}.
The whole process is assumed to be repeated many $\nu\gg1$ times, so that sufficient 
statical data is always guaranteed.
Note that parameter sensing is the most ``optimistic'' metrology setting \cite{Lehmann1998}, 
in the sens that any type of general phase estimation problem that also accounts for the 
non-perfect prior knowledge about the value of $\param$ or finite size of the 
measurement data is bound to be more difficult \cite{Gill1995}.

Let $p_{n|\param}$ be the probability that (in any given round) the measurement resulted 
in the outcome $n$ given that the initial state was $\rho$ and the parameter was $\param$.
Then, the mean squared error $\Delta^2\tilde\param$ of any unbiased and consistent estimator $\tilde\param$ of $\param$ is 
lower-bounded by the \emph{Cram\'{e}r-Rao Bound} (CRB) \cite{Lehmann1998}: 
\begin{equation} \label{eq:CRB}
\Delta^2\tilde\param
\geq
\frac{1}{\nu\FI\left(\left\{p_{n|\param}\right\}\right)},
\end{equation}
where
\begin{equation} \label{eq:FI}
\FI\left(\left\{p_{n|\param}\right\}\right)\coloneqq \sum_n
\frac{1}{p_{n|\param}}\left(\frac{\d p_{n|\param}}{\d\param}\right)^2
\end{equation}
is the \emph{classical Fisher information} (FI).
Importantly, in the phase sensing scenario, the CRB \eref{eq:CRB} is guaranteed to be 
tight in the limit of large $\nu$ \cite{Lehmann1998}.
The classical Fisher information, hence, quantifies in an operationally 
meaningful way what is the ultimate resolution to the fluctuations of $\param$ with a 
given measurement and state.

The quantum device is taken to act on a single particle (photon, atom, etc.) with 
Hilbert space $\Hloc$ of finite dimension $\dloc \coloneqq |\Hloc|$ via a Hamiltonian $h$, which without loss of generality we assume to be traceless
(any nonzero trace of $h$ can always be incorporated into an irrelevant global phase-factor).
The device unitarily encodes the unknown parameter $\param$ by performing the transformation $\e^{-\ii h \param}$.
The multi-particle state $\rho \in \H_N \coloneqq \left(\Hloc\right)^{\otimes N}$ moves along the trajectory
\begin{equation}
\label{eq:par_encod}
 \param \mapsto (\e^{-\ii h \param})^{\otimes N}\,\rho\,(\e^{\ii h \param})^{\otimes N} ,
\end{equation}
corresponding to the unitary evolution with the global Hamiltonian 
\begin{equation}
\label{eq:ham_encod}
H \coloneqq H_N \coloneqq \sum_{j=1}^N h^{(j)} .
\end{equation}
The measurement is defined by a positive-operator-valued measure (POVM), $\{\Pi_n^N\}_n$, acting 
on the whole system (or, while accounting for particle losses, only on the remaining to $N-k$ particles, see \figref{fig:protocol}) and satisfying $\sum_n \Pi_n^{N}=\mathbb{I}$.
It yields outcome $n$ with probability 
\begin{equation}
p_{n|\param} = \tr\left(\Pi_n^N\,\e^{-\ii H \param}\,\rho\,\e^{\ii H \param}\right).
\end{equation}

In the seminal work by \citet{Braunstein1994}, it was shown how to quantify the maximal usefulness 
of a state $\rho$ in the above scenario by maximizing the classical Fisher information \eref{eq:FI} over all possible POVMs.
The resulting quantity is called \emph{quantum Fisher Information} (QFI).
It depends solely on the quantum state $\rho$ being considered and the Hamiltonian $H$ responsible for the parameter encoding and we denote it by $\QFI(\rho,H)$.
A QFI scaling faster than linear with $N$ (for the fixed local Hamiltonian $h$) ultimately leads to super-classical resolutions by virtue of the CRB \eref{eq:CRB}. Since the Hamiltonian \eref{eq:ham_encod} is local and parameter-independent, the ultimate HL is unambiguously given by $\FI \propto N^2$ with super-classical scaling being possible solely due to the entanglement properties of $\rho$ and not due to a non-local or non-linear parameter dependence \cite{Boixo2007,Rivas2010,Zwierz2010}. Importantly, thanks to the unitary character of the parameter encoding \eref{eq:par_encod}, 
resolutions that scale beyond SQL can indeed be attained in metrology \cite{Escher2011,Demkowicz2012}.

Although in the phase sensing scenario the optimal measurement can be designed for a particular value of $\param$,
it is often enough to know, prior to the estimation, the parameter to lie within a sufficiently narrow window 
of its values, as then there exist a sequence of measurements which eventually---in the limit of many protocol repetitions 
($\nu\to\infty$ in \figref{fig:protocol}) with measurements adaptively adjusted---yields a classical Fisher information 
that still achieves the QFI \cite{Fujiwara2006}. Crucially, the optimal scaling of the QFI achievable in the phase sensing 
scenario is proportional to $N^2$, what proves the Heisenberg scaling to indeed be the ultimate one.


For the sake of having a concise terminology, we call a family of states for increasing $N$ \emph{useful} for quantum sensing if there exists a Hamiltonian
$h$ in \eqnref{eq:par_encod} such that the corresponding QFI scales faster than $N$ (i.e., $\QFI(\rho, H) \notin \landauO(N)$) in the limit of large $N$.
In contrast, we say that the family of states is \emph{not useful} for quantum sensing (and hence also for all less  ``optimistic'' metrological scenarios), if its QFI scales asymptotically at most like $N$ (i.e., $\QFI(\rho, H) \in \landauO(N)$).
We adopt the above nomenclature for the sake of brevity and concreteness. However, let us stress that states reaching beyond SQL, despite not preserving super-classical precision scaling, may also yield dramatic precision enhancement (e.g., squeezed states in gravitational detectors
\cite{LIGO2011}), which, in fact, guarantees then their rich entanglement
structure \cite{Riedel2010,Krischek2011,Strobel2014,Lucke2014,Sorensen2001,Toth2012}.
Nevertheless, it is the super-classical precision scaling that
manifests the necessary entanglement properties to be maintained with
the system size. In particular, its protection at the level of QFI has recently
allowed to design novel noise-robust metrology protocols \cite{Kessler2014,Brask2015}.

In the remainder of this section we give a mathematical definition of the QFI and discuss some of its properties.
The QFI has an elegant geometric interpretation \cite{Braunstein1994} as the ``square of the speed'' along the trajectory \eref{eq:par_encod} measured with respect to the Bures distance $\dB\left(\rho,\sigma\right)\coloneqq \sqrt{2 \, [1-\mathcal{F}\left(\rho,\sigma\right)]}$, where $\mathcal{F}\left(\rho,\sigma\right) \coloneqq \tr\sqrt{\rho^{1/2}\,\sigma\,\rho^{1/2}}$ is the fidelity.
This allows one to define the QFI geometrically \cite{Bengtsson2006}:
\begin{equation} \label{eq:QFI_geom}
 \QFI\left(\rho,H\right) \coloneqq \left[ \lim_{\param\to 0} 2\,\dB\left(\rho, \e^{-\ii H \param}\,\rho\,\e^{\ii H \param}\right)/|\param| \right]^2.
\end{equation}
This geometric interpretation of QFI is key for the derivation of the following results.
Using the spectral decomposition of a quantum state $\rho = \sum_{j=1}^{d^N}\,p_j \proj{e_j}$, with  $p_j \geq 0$ 
denoting its eigenvalues, we can write the QFI more explicitly as \cite{Toth2014,Demkowicz2015}:
\begin{equation} \label{eq:QFI_formula}
 \QFI\left(\rho,H\right)=2\sum_{j,k:\,p_{j}+p_{k}\neq0}\frac{\left(p_{j}-p_{k}\right)^{2}}{p_{j}+p_{k}}\left|\braket{e_{j}}{H\left|e_{k}\right.}\right|^{2},
\end{equation}
which for pure states, $\rho=\psi$, simplifies further to the variance $\left.\Delta^2H\right|_{\psi}$ of $H$, that is,
\begin{equation} \label{eq:QFI_formula_pure}
\QFI\left(\psi,H\right)=4\left[\tr\left(\psi H^2\right)-\tr\left(\psi H\right)^2\right] \eqqcolon \left.4\Delta^2H\right|_{\psi} .
\end{equation}
Let us also recall that the QFI is a convex function on the space of quantum states.
This reflects that mixing states can never increase their parameter sensitivity above their average sensitivity.
This, together with the fact that the QFI is also additive on product states \cite{Toth2014,Demkowicz2015}, directly proves that only entangled states can lead to resolutions beyond the SQL.

\section{Concentration of measure phenomenon} \label{sec:notat}
Before presenting our main results, we briefly discuss a key ingredient for their proofs---the \emph{concentration of measure} phenomenon \cite{Ledoux2005,Anderson2010,Collins2015}.
For a more detailed discussion we point the reader to \appref{sec:concentration}.
For any finite dimensional Hilbert space we denote by $\mu\left(\H\right)$ the Haar measure on the special unitary group $\SU\left(\H\right)$.
The Haar measure can be thought of as the uniform distribution over unitary transformations.
We denote by $\Pr_{U \sim \mu\left(\H\right)}(A(U))$ the probability that a statement $A$ holds for unitaries $U$ drawn from the measure $\mu\left(\H\right)$ and by
\begin{equation}\label{Average}
\Expect_{U \sim \mu\left(\H\right)} f(U) \coloneqq \int_{\SU\left(\H\right)}d\mu\left(U\right) f(U)
\end{equation}
the expectation value of a function $f:\SU\left(\H\right) \to \R$.
Our findings concern the typical value of such functions.
For example, $f(U)$ could be the QFI of some family of so-called isospectral states, i.e., states of the form $U\,\rho\,U^\dagger$, with $\rho$ some fixed state on $\H$ and $U \sim \mu\left(\H\right)$ a unitary drawn from the Haar measure on $\SU(\H)$ \cite{Oszmaniec2014}.
Note that as $\QFI(U\,\rho\,U^\dagger,H) = \QFI(\rho,U^\dagger\,H\,U)$ (this follows directly 
from \eqnref{eq:QFI_formula}), all our results can also be interpreted as being about random Hamiltonians instead of random states.

To show that for almost all $U$ the value of such a function is close to the typical value and that this typical value is close to the average, we repeatedly employ the following concentration of measure inequality \cite{Anderson2010}:
\begin{equation}\label{eq:concentrsu1}
 \Pr_{U \sim \mu\left(\H\right)}\left( \left|f\left(U\right)-\Expect_{U \sim \mu\left(\mathcal{\mathcal{H}}\right)}f\right|\geq\epsilon\right)
 \leq 
2\, \mathrm{exp}\left(-\frac{|\H|\epsilon^{2}}{4L^{2}}\right)\ 
\end{equation}
It holds for every $\epsilon\geq0$ and every function $f:\mathrm{SU}\left(\H\right) \to \R$ that is \emph{Lipschitz continuous} (with respect to the geodesic distance on $\mathrm{SU}\left(\H\right)$) and thus possesses its corresponding Lipschitz constant $L$.
Recall that the Lipschitz constant gives the bound on how fast the value of a function can change under a change of its argument.
For a formal definition of $L$ see \eqnref{eq:lipschitz constant formula} in \appref{sec:concentration}, where we explicitly prove bounds on Lipschitz constants of all the functions relevant for our considerations.
Let us here only note that as both the FI \eref{eq:FI} and the QFI \eref{eq:QFI_geom} are non-trivial (in particular non-linear) functions 
of quantum states, we need to resort to advanced techniques of differential geometry.

Before we move on to our results, we introduce a minimal amount of additional notation:
Given two functions $f,g$ we write $f(N)\in\Theta(g(N))$ if both functions have the same behaviour in the limit of large $N$ (up to a positive multiplicative constant) and write $f(N)\in\landauO(g(N))$ if there exists a constant $C$ such that $f(N) \leq C\,g(N)$ in the limit of large $N$.
Slightly abusing notation, we sometimes also use the symbols $\Theta(f(N))$ and $\landauO(f(N))$ to denote an arbitrary function with the same asymptotic behavior as $f$.
For any operator $X$ we denote its operator norm by
\begin{equation}
\|X\|\coloneqq \sup_{|\psi\rangle\neq0}\frac{\left\|X\ket{\psi}\right\|}{\|\ket{\psi}\|},
\end{equation}
where $\|\ket{\psi}\|=\sqrt{\langle\psi|\psi\rangle}$ stands for the standard vector norm. 
Then, the trace and the Hilbert-Schmidt norms of $X$ are defined 
as $\|X\|_1\coloneqq \tr\sqrt{X^\dagger X}$ and 
$\left\|X\right\|_\t{HS} \coloneqq \sqrt{\tr\left(X^\dagger X\right)}$, respectively.
These generally 
obey the relation
$\|X\|\!\le\!\|X\|_\t{HS}\!\le\!\|X\|_1$.

\section{Futility of general random states} \label{sec:disting}
First, we show that Haar-random isospectral states of distinguishable particles are typically not useful for quantum metrology even if they are pure and hence typically highly entangled \cite{Page1993,Foong1994, Sen1996,Hayden2006,Hayden2006,Gross2009,Oszmaniec2014,moPHD}.
This remains true even if one allows for local unitary (LU) optimization  before the parameter is encoded (see \figref{fig:protocol}). 
Note that in the special case $d=2$ the LU-optimization of the input state is equivalent to an optimization over 
all unitary parameter encodings.
The maximal achievable QFI with LU-optimization is given by
\begin{equation} \label{eq:QFI_LU}
\QFI^\mathrm{LU}\left(\rho,H\right) \coloneqq \sup_{V \in \mathrm{LU}} \QFI\left(V\,\rho\,V^\dagger,H\right),
\end{equation}
where $\mathrm{LU}$ denotes the local unitary group, i.e., a group consisting of unitaries
of the form $V = V_1 \otimes V_2 \otimes \dots \otimes V_N$ with $V_j$ acting on the $j$-th particle of the system (see \figref{fig:protocol}).
Despite the fact that for other states this sometimes boosts their QFI \cite{Hyllus2010a}, we have the following no-go theorem for random states from the full Hilbert space $\H_N$ of distinguishable particles:
\begin{thm}[Most random states of distinguishable particles are not useful for metrology even after LU-optimization]
\label{thm:moststatesareuseless}
Fix a local dimension $d$, single-particle Hamiltonian $h$, and a pure state $\psi_N$ on $\H_N$.
Let $\QFI^{\mathrm{LU}}(U) \coloneqq \QFI^{\mathrm{LU}}(U\,\psi_N\,U^\dagger, H)$, then
 \begin{equation}\label{eq:distineq}
  \Pr_{U \sim \mu\left(\H_N\right)}\left( \QFI^{\mathrm{LU}}(U) \notin \Theta\left(N\right) \right) \leq \exp\left(-\Theta\left(\frac{d^N}{N^2}\right)\right)
 \end{equation}
\end{thm}
\begin{proof}[Proof sketch.]
 
 From \eqnref{eq:QFI_formula_pure} we have that $\QFI(\psi,H)\leq 4\,\tr(\psi\,H)$, which implies
 \begin{equation}
  \begin{split}
   &\QFI^{\mathrm{LU}}\left(U\,\psi_N\,U^\dagger,H\right) \leq 4\sup_{V\in\mathrm{LU}} \tr\left(U\,\psi_N\,U^\dagger V\,H^2\,V^\dagger \right) \\
   &= 4\sup_{V_1,\dots,V_N} \sum_{j,k}\tr\left(U\,\psi_N\,U^\dagger V_j\,h^{(j)}\,V_j^\dagger \, V_k\,h^{(k)}\,V_k^\dagger\right).
  \end{split}
 \end{equation}
 The terms with $j=k$ give a contribution of at most $4\,N\,\| h \|^2$.
 In the remaining terms, however, the operator $V_j\,h^{(j)}\,V_j^\dagger \, V_k\,h^{(k)}\,V_k^\dagger$ is traceless, so that
 \begin{equation}
  \begin{split}
   &\sup_{V_1,\dots,V_N}\sum_{j \neq k} \tr\left(U\,\psi_N\,U^\dagger V_j\,h^{(j)}\,V_j^\dagger \, V_k\,h^{(k)}\,V_k^\dagger\right)\\
   &\quad\leq \| h\|^2 \sum_{j\neq k} \left\|\tr_{\neg j,k}(U\,\psi_N\,U^\dagger) - \1/d^2 \right\|_1 \ .
  \end{split}
 \end{equation}
 But, the average of $\|\tr_{\neg j,k}(U\,\psi_N\,U^\dagger) - \1/d^2 \|_1$ can be upper bounded exponentially \cite{Popescu2006} as
 \begin{equation}
   \Expect_{U \sim \mu\left(\H_N\right)} \left\|\tr_{\neg j,k}(U\,\psi_N\,U^\dagger) - \1/d^2 \right\|_1 \leq \frac{d^2}{\sqrt{d^N}} ,
 \end{equation}
 so that
 \begin{equation}\label{eq:ineq1}
  \Expect_{U \sim \mu\left(\H_N\right)} \QFI^{\mathrm{LU}}(U) \leq 4N\|h\|^2 \left[1+\frac{(N-1)\,d^2}{\sqrt{d^N}}\right].
 \end{equation}
 Conversely, a direct computation of the average non LU-optimized QFI yields a lower bound of order $N\|h\|^2$.
 Applying a concentration inequality of the type given in \eqnref{eq:concentrsu1} yields the result (see \appref{sec:finproofs} for details).
\end{proof}

Due to the convexity of QFI the typical behavior of the QFI on any unitary-invariant ensembles \cite{Bengtsson2006} of mixed density matrices in $\H_N$ can only be worse than that of pure states predicted by the above theorem. Furthermore, a bound similar to \eqnref{eq:ineq1} can also be derived for Hamiltonians $H$ with few body terms, like for example such with finite or short range interactions on regular lattices or those considered in \refcite{Boixo2007}.
Lastly, let us remark that, as we consider the most optimistic phase sensing protocol, 
\thmref{thm:moststatesareuseless} also disproves possibility of super-classical scalings 
of precision when considering random states in any general phase estimation 
protocol \cite{Gill1995}, e.g., the Bayesian single-shot scenarios \cite{Berry2000,Bagan2001b,Bartlett2007}.

The above proof relies on the fact that most random states on $\H_N$ have very mixed two-particle marginals. 
Thus, high entanglement entropy is enough to make random pure states on $\H_N$ useless for quantum metrology.
Complementing this, in \cite{Augusiak2016} it has been proven that non-vanishing geometric measure of entanglement 
$E_g\left(\psi\right)\in\Theta\left(1\right)$ is at the same time necessary for Heisenberg scaling 
(recall that the geometric measure of entanglement for a pure state $\psi$ is defined 
as $E_g(\psi)\coloneqq 1- \mathrm{max}_{\sigma\in\mathrm{SEP} } \tr(\psi \sigma)$, where $\mathrm{SEP}$
denotes the set of separable states in $\mathcal{D}\left(\H_N\right)$).
Interestingly, pure random states of $N$ particles do typically satisfy $E_g\left(\psi\right) \approx 1$ \cite{Gross2009}.
This, {together with \thmref{thm:moststatesareuseless} shows that, contrary to a recent
conjecture of \refcite{Rezakhani2015}, high geometric measure of entanglement is not sufficient for states to exhibit 
super-classical precision scaling in quantum metrology. Let us however note that this is consistent with numerical findings of \refcite{Tichy2015}}. 
That the presence or absence of super-classical scaling of the QFI arises solely from 
the two-particle marginals has recently also been noted in \refcite{Baumgratz2016}.

\section{Usefulness of random symmetric states} \label{sec:symuseful}
We now turn to the study of random states from the symmetric (bosonic) subspace of $N$ qudits, $\Hsym_N \coloneqq \lspan\{\ket{\psi}^{\otimes N}: \ket{\psi} \in \Hloc\}$, which is of dimension $|\Hsym_N| = \binom{N+d-1}{N}\in\Theta\left(N^{d-1}\right)$.
This subspace of states contains metrologically useful states such as the GHZ state or the Dicke states \cite{Apellaniz2015} and naturally appears in experimental setups employing photons and bosonic atoms \cite{Demkowicz2015}.  For the special case $d=2$ it was proven in \cite{Hyllus2010a} that with LU optimization almost all pure symmetric states exhibit $\QFI^{\mathrm{LU}}>4\,N\,\|h\|^2$.

In what follows we consider random isospectral symmetric states, i.e., states of the form $U\,\sigma_N\,U^\dagger$ with $\sigma_N$ being a fixed state on $\Hsym_N$ and $U \sim \mu\left(\Hsym_N\right)$.
By $\sigma_{\mathrm{mix}}=\1_{\Hsym_N}/|\Hsym_N|$ we denote the maximally mixed state in $\Hsym_N$.
In particular, we prove that typically such symmetric states achieve a Heisenberg-like scaling, provided that the spectrum of $\sigma_N$ differs sufficiently from the spectrum of  $\sigma_{\mathrm{mix}}$.
Interestingly, this holds even without LU-optimization:
\begin{thm}[Most random isospectral symmetric states are useful for quantum sensing]
\label{thm:symmetricstatesareuseful}
Fix a single-particle Hamiltonian $h$, local dimension $d$ and a state $\sigma_N$ from the symmetric subspace $\Hsym_N$ with eigenvalues $\{p_j\}_j$.
Let $\QFI(U) \coloneqq \QFI(U\,\sigma_N\,U^\dagger, H)$, then
 \begin{equation}\label{isospectralbound}
  \begin{split}
   &\Pr_{U \sim \mu\left(\Hsym_N\right)} \left( \QFI\left(U\right) < \dB\left(\sigma_N,\sigma_{\mathrm{mix}}\right)^2 \Theta\left(N^2\right) \right) \\
   &\qquad\leq \mathrm{exp}\left( - \dB\left(\sigma_N,\sigma_{\mathrm{mix}}\right)^3 \cdot \Theta\left(N^{d-1}\right) \right).
  \end{split}
 \end{equation}
\end{thm}
\begin{proof}[Proof sketch.]
 We use the standard integration techniques (see \appref{sec:averages} for details) on the unitary group to show that
 \begin{equation} \label{eq:result integration}
  \Expect_{U \sim \mu\left(\mathcal{\Hsym}_N\right)} \QFI(U) = \tr(h^2)\; \mathcal{G}(N,d)\;\Lambda\!\left(\{p_j\}_j\right) \, ,
 \end{equation}
 where
 \begin{equation}
  \mathcal{G}(N,d) \coloneqq \frac{4\, N\,(N+d) }{d\,(d+1)} \frac{|\Hsym_N|}{|\Hsym_N|+1},
 \end{equation}
 and $\Lambda\left(\{p_j\}_j\right)$ is a function of the eigenvalues $\{p_j\}_j$ which for pure states
 attains $\Lambda=1$. 
Therefore, for the case of pure states we have
 \begin{equation}\label{eq:av_pure}
   \Expect_{U \sim \mu\left(\mathcal{\Hsym}_N\right)} \QFI(U\psi_N U^\dagger,H ) 
= 
\tr(h^2)
\;\mathcal{G}(N,d)
 \end{equation}
 where $\psi_N$ is an arbitrary pure state on $\Hsym_N$.
 From this it clearly follows that the average QFI of random pure symmetric states exhibits Heisenberg scaling in the limit $N\rightarrow\infty$.
 Moreover, it turns out that $\Lambda\left(\{p_j\}_j\right)$ satisfies the inequality
 \begin{equation}\label{eq:funnyineq}
  \Lambda(\{p_j\}_j) \geq \frac{|\Hsym_N|}{|\Hsym_N|-1} \frac{\dB\left(\sigma_N,\sigma_{\mathrm{mix}}\right)^2}{2} \ .
 \end{equation}
 The inequality \eqref{isospectralbound} follows now from concentration of measure inequalities on $\SU\left(\Hsym_N\right)$
 (see \appref{sec:finproofs} for a detailed proof).
\end{proof}
As the Bures distance to the maximally mixed state $\dB\left(\sigma_N,\sigma_{\mathrm{mix}}\right)$ enters the theorem in a non-trivial way, we illustrate the power of the result by showing that even states that asymptotically move arbitrary close to $\sigma_{\mathrm{mix}}$ still typically achieve a super-classical scaling:
\begin{corr}[Sufficient condition for usefulness of random isospectral symmetric states]
\label{corr:sym_useful}
Let $U\,\sigma_N\,U^\dagger$ be an ensemble of isospectral states from the symmetric subspace $\Hsym_N$ with eigenvalues $\{p_j\}_j$.
\thmref{thm:symmetricstatesareuseful} implies that random states drawn from such a ensemble are typically useful for sensing as for any $\alpha<\min\left\{1/2,(d-1)/3\right\}$ they yield a precision scaling $1/N^{2(1-\alpha)}$ provided that $\dB\left(\sigma_N,\sigma_{\mathrm{mix}}\right)\geq 1/N^\alpha$.
\end{corr}
Let us remark that Theorem \eqref{thm:symmetricstatesareuseful} constitutes a fairly complete description of the typical properties of QFI on various ensembles of  isospectral density matrices. Typical properties of entanglement and its generalizations on sets of isospectral density matrices were analyzed in \refcite{Oszmaniec2014}.

\section{Robustness to imperfections} \label{sec:robust}
Next we underline the practical importance of the above results by showing that the usefulness of random symmetric states is robust against dephasing noise and particle loss.

\subsection{Depolarising noise}
We first investigate how mixed $\sigma_N$ may become while still providing a quantum advantage for metrology.
To this aim, we consider a concrete ensemble of depolarized states:

\begin{exa}[Depolarized random symmetric states]
\label{thm:depolensamble}
Fix a local dimension $d$, single particle Hamiltonian $h$, and $p\in\left[0,1\right]$.
Let $\psi_N$ be a pure state on $\Hsym_N$ and set
\begin{equation}\label{depol}
 \sigma_N(p) = (1-p)\,\psi_N + p\, \sigma_{\mathrm{mix}} \ .
\end{equation}
Let $F_p\left(U\right) \coloneqq \QFI(U\,\sigma_N(p)\,U^\dagger, H)$, then for every $\epsilon>0$
\begin{equation}\label{depolarconc}
  \Pr_{U \sim \mu\left(\Hsym_N\right)}\left( |\QFI_p(U) - \Expect F_p| \geq \epsilon F_p \right) 
  \leq \exp\left(-\epsilon^2\cdot\Theta\left(N^{d-1}\right)\right),
%
\end{equation}
where $\Expect F_p \coloneqq \Expect_{U \sim \mu\left(\mathcal{\Hsym}_N\right)} F_p$ is given by \eqnref{eq:result integration} with
\begin{equation}
\label{eq:lambda_dep}
 \Lambda = \Lambda_p\coloneqq\frac{(1-p)^2} {1-p+2\,p/|\Hsym_N|} .
\end{equation}
\end{exa}
The above example shows that for all values of $p<1$, the Heisenberg scaling of the QFI is typically retained.
The QFI then still concentrates around its average, which is of order $N^2$.
Moreover, we observe that for large $N$, the average value of QFI of random symmetric depolarized states decreases essentially linearly with $p$ as $\left|\Lambda_p-(1-p) \right|\leq 2/|\Hsym_N|$.
Finally, \eqnref{depolarconc} is a large deviation inequality for QFI  on the ensamble of depolarized pure symmetric states, with the mean $\Expect F_p$. The average $\Expect F_p$ for the special case $p=0$ is  given by \eqref{eq:av_pure}.

\subsection{Finite loss of particles}
\label{sec:part_losses}
Next we investigate whether the Heisenberg scaling of random symmetric states
is robust under the loss of particles.
We model the particle loss by the partial trace over $k \leq N$ 
particles, i.e., $\sigma_N\mapsto\tr_k\left(\sigma_N\right)$ for a given state $\sigma_N$.
Note that due to the permutation symmetry of state $\sigma_N$ considered, it does not matter which particles are lost.
In particular, such mechanism is equivalent to the situation in which one is capable of measuring only (as if distinguishable) $N-k$ of the particles.
We are therefore interested in typical properties of $F\left(\tr_k(U\,\sigma_N\,U^\dagger),H_{N-k}\right)$, where $H_{N-k}=\sum_{j=1}^{N-k} h^{(j)}$ and $U\,\sigma_N\,U^\dagger$ is a random isospectral state on $\Hsym_N$.

For comparison, let us recall that the GHZ-state, which is known to be optimal in quantum sensing \cite{Giovannetti2004}, becomes completely useless upon the loss of just a single particle as the remaining reduced state is separable.
In contrast, sufficiently pure random bosonic states typically remain useful for sensing even when a constant number of particles has been lost.
Even the Heisenberg scaling $\sim1/N^2$ of the QFI is preserved:
\begin{thm}[Random isospectral symmetric states are typically useful under finite particle loss] \label{thm:symmetriclosses}
 Fix a local dimension $d$, single particle Hamiltonian $h$, and a state $\sigma_N$ on $\Hsym_N$ with eigenvalues $\{p_j\}_j$.
 Let $\QFI_k\left(U\right) \coloneqq \QFI(\tr_k\left(U\,\sigma_N\,U^\dagger\right), H_{N-k})$, then
 \begin{equation}\label{losspar}
  \begin{split}
  &\Pr_{U \sim \mu\left(\Hsym_N\right)}\left( \QFI_k\left(U\right) < \|\sigma_N -\sigma_{\mathrm{mix}}\|_{\mathrm{HS}}^2 \, \Theta\left(\frac{N^2}{k^d}\right) \right) \\
  &\qquad\leq \exp\left(-\|\sigma_N -\sigma_{\mathrm{mix}}\|_{\mathrm{HS}}^4 \,\Theta\left(\frac{N^{d-1}}{k^{2d}}\right)\right) \ .
  \end{split}
 \end{equation}
\end{thm}

The main difficulty in the proof of the above theorem is that the random matrix ensemble induced by the partial trace of random isospectral states form $\Hsym_N$ is not well-studied.
Hence we cannot use standard techniques to compute the average of $\QFI_{k}\left(U\right)$.
We circumvent this problem by lower-bounding the
QFI by the asymmetry measure \cite{Marvian2014} and then the HS norm
\begin{equation}\label{eq:lowerBOUND}
F\left(\rho,H\right)\geq ||\rho,H||_1^2 \geq \left\|\left[\rho,H \right]\right\|_{\mathrm{HS}}^2 
\end{equation}
and then computing the average of the right hand side instead.
We provide all the technical details of the proof in the \appref{sec:finproofs}; let us here only consider the two mode case as an example:
\begin{exa}
For $d=2$ and after fixing $\tr (h^2)=1/2$, the following inequality holds
\begin{equation}\label{eq:qubitsloss}
 \begin{split}
&\Expect_{U \sim \mu\left(\Hsym_N\right)} \QFI\left(\tr_k\left(U\,\sigma_N\,U^\dagger\right), H_{N-k}\right)\\
&\qquad \geq
\frac{1}{3}\frac{(N-k)(N+1)}{(k+1)(k+2)}\frac{(N+1)\tr \rho_N^2-1}{N},
\end{split}
\end{equation}
(see \lemref{boundavloss} in \appref{sec:averages}) which for the pure states further simplifies to
\begin{equation}
\Expect_{U \sim \mu\left(\Hsym_N\right)}\QFI\left(\tr_k\left(U\,\psi_N\,U^\dagger\right), H_{N-k}\right)\geq
\frac{1}{3}\frac{(N-k)(N+1)}{(k+1)(k+2)} \ .
\end{equation}
It is interesting to note that without particle losses, i.e., $k=0$,
this formula gives
\begin{equation}
\Expect_{U \sim \mu\left(\Hsym_N\right)} \QFI\left(U\,\psi_N\,U^\dagger, H_N\right)\geq
\frac{1}{6}N(N+1)\ ,
\end{equation}
a result that differs from the exact expression \eqref{eq:av_pure} only by a factor of $1/2$.
\end{exa}

The above result is most relevant for atomic interferometry experiments \cite{Schumm2005,Sebby2007,Zhang2012,Stroescu2015}, in which unit detection efficiencies can be achieved and it is hence reasonably possible to limit the loss of particles to a small number.
In contrast, current optical implementations are limited by inefficiencies of photon detectors that are adequately modeled with a fictions beam-splitter model \cite{Demkowicz2015}.
It is known in noisy quantum metrology that generic uncorrelated noises (in particular the noise described by the beam-splitter model) constrain the ultimate precision to a constant factor beyond the SQL \cite{Escher2011,Demkowicz2012}.
The beam-splitter effectively fixes the loss-rate per particle allowing \emph{all} the particles to be lost with some probability.
In fact modeling losses with a beam-splitter is equivalent to tracing-out $k$ particles with probability $\binom{N}{k}(1-\eta)^k\eta^{N-k}$, 
where $\eta$ is the fictitious beam-splitter transitivity (see \appref{sec:partloss} for the proof).
Thus, the number $k$ of particles lost fluctuates according to a binomial distribution.
As a result, the lower bound on the average QFI utilized in \secref{sec:part_losses} must also be averaged over 
the fluctuations of $k$ and the super-classical scaling is lost.
Our results on the robustness of random bosonic states against \emph{finite} particle losses are hence fully consistent with the no-go theorems of \cite{Escher2011,Demkowicz2012}.
Moreover, the fact that random states are much more robust against particle loss than for example N00N states, which loose their metrological usefulness upon the loss of a single particle, raises the hope that for finite $N$ they might still perform comparably well even under uncorrelated noise.

\section{Attaining the Heisenberg Limit with a simple measurement} \label{sec:phys}
We have demonstrated that random bosonic states lead in a robust manner to super-classical scaling of the QFI.
This proves that, in principle, they must allow to locally sense the phase around any value 
with resolutions beyond the SQL. However, as previously explained in \secref{sec:setting}, the phase sensing 
scenario allows the measurement to be optimised for the particular parameter value considered. Moreover, such 
measurement may also strongly depend on the state utilised in the protocol, so one may question 
whether it could be potentially implemented in a realistic experiment, as theoretically it then must be
adjusted depending on the state drawn at random. Thus, it is a priori not clear if metrological usefulness of 
random symmetric states can be actually exploited in practice. Here we show that this is indeed the case. For random 
symmetric states of two-mode bosons, a standard measurement in optical and atomic interferometry suffices
to attain the Heisenberg scaling of precision when sensing the phase around any value.

In particular, we consider the detection of the distribution of the $N$ bosons
between two modes (interferometer arms) after a balanced beam-splitter transformation \cite{Demkowicz2015}.
As depicted in \figref{fig:interferometer}, this corresponds in optics to the photon-number detection at two output ports of a Mach-Zehnder (MZ) interferometer \cite{Bachor2004}.
Yet, such a setup also applies to experiments with atoms in double-well potentials, in which the beam-splitter 
transformation can be implemented via trap-engineering and atomic interactions \cite{Schumm2005,Sebby2007}, 
while number-resolving detection has recently been achieved via cavity-coupling \cite{Zhang2012} and fluoresence \cite{Stroescu2015}.

One may directly relate the general protocol of distinguishable qubits (see \figref{fig:protocol}) to the optical setup of photons in two modes (see \figref{fig:interferometer}) after acknowledging that the Dicke basis of general pure symmetric qubit states is nothing but their two-mode picture, in which a qubit in a state $\ket{0}$ ($\ket{1}$) describes a photon travelling in arm $a$ ($b$) of the interferometer.
In particular, a general pure bosonic state of $N$ qubits may then be written as a superposition of Dicke states $\{\ket{n,N-n}\}_{n=0}^N$, where each $\ket{n,N-n}$ 
represents the situation in which $n$ and $N-n$ photons are travelling in arm $a$ and $b$, respectively.
In \figref{fig:interferometer}, the estimated phase $\param$ is acquired in between the interferometer arms, i.e., by the transformation $\exp(-\ii {\hat J}_z \param)$ \footnote{${\hat J}_x=\tfrac{1}{2}(\hat a^\dagger\hat b+\hat b^\dagger\hat a)$, ${\hat J}_y=\tfrac{1}{2}(\hat b^\dagger\hat a-\hat b^\dagger\hat a)$, ${\hat J}_z=\tfrac{1}{2}(\hat a^\dagger\hat a-\hat b^\dagger\hat b)$ are the standard two-mode angular momentum operators defined via the Jordan-Schwinger map \cite{Demkowicz2015}}.
In the qubit picture this corresponds to a single-particle unitary $\exp(-\ii h\param)$ with $h=\sigma_{z}/2$.
Moreover, the unitary balanced beam-splitter transformation of \figref{fig:interferometer}, commonly defined in the modal picture as $\hat{B}\coloneqq\exp(-\ii \pi{\hat J}_x/2)$, is then equivalent to a local rotation $\exp(-\ii\pi\sigma_{x}/4)$ of each particle in the qubit picture applied after the $\param$-encoding.

\begin{figure}[!t]
 \centering
 \includegraphics[width=\columnwidth]{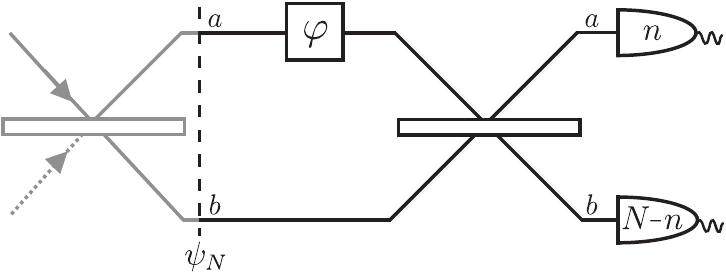}
 \caption{Mach-Zehnder interferometer:~%
   We consider a bosonic pure state $\psi_N$ of $N$ photons
   in modes $a$ and $b$ inside the interferometer. The estimated
   phase $\param$ is acquired due to path difference in the arms.
   After a balanced beam-splitter the number of photons $n$ and $N-n$ are
   measured in arms $a$ and $b$ respectively.
 }
\label{fig:interferometer}
\end{figure}

Hence, the measurement of \figref{fig:interferometer} with outcomes labeled by $n$ (the number of photons detected in mode $a$) corresponds to a POVM $\{\Pi_n^N\}_{n=0}^{N}$ with elements $\Pi_n^N=\hat B^\dagger D_n^N\hat B$, where $D_n^N$
are the projections onto Dicke states, $\ket{D_n^N}\coloneqq\ket{n,N-n}$.
Given a general pure state $\psi_N$ inside the interferometer (see \figref{fig:interferometer}), 
the state after acquiring the estimated phase reads $\psi_N(\param)\coloneqq \e^{-\ii {\hat J}_z \param} \psi_N \e^{\ii {\hat J}_z \param}$.
Then, the probability of outcome $n$ given that the unknown parameter was $\param$ is just
\begin{equation}
p_{n|\param}(\psi_N) = \tr\!\left[\Pi^{N}_n \,\psi_N(\param)\right].
\label{eq:p_n}
\end{equation}
Before we proceed, let us note that due to the identity $\hat B\,\e^{-\ii {\hat J}_z \param}\hat B^\dagger =\e^{\ii {\hat J}_y \param}$ 
it is possible to effectively map the above measurement scheme to the situation in which the 
initial state is already propagated thorough a beam-splitter:~$\tilde \psi_N= \hat{B}\psi_N\hat{B}^\dagger$, yet 
the parameter is encoded via a Hamiltonian in the $y$ direction (via $\tilde h=-\sigma_y/2$).
As a result, the measurement (POVM) elements then simplify to just projections onto the Dicke states $D^N_n$ (see \appref{sec:averagefifortwomodes} for details).

Having the explicit form of the measurement-outcome probability, we can compute the corresponding (classical) FI 
\begin{equation}
\FI(\left\{p_{n|\param}(\psi_N)\right\})
=
\sum_{n=0}^N
\frac{\tr\left(\ii\,[\Pi^{N}_n,\hat J_z] \,\psi_N(\param)\right)^2}
{\tr\left[\Pi^{N}_n \,\psi_N(\param)\right]} \ .
\label{eq:FI_p_n}
\end{equation}
The unitary rotation $\e^{-\ii {\hat J}_z \param}$ entering the definition of 
$\psi_N(\param)$ is responsible for the strong dependence (see also the numeric results 
in \secref{sec:numerics}) of the FI on the value of $\param$.
However, let us note that when averaging (with respect to Haar measure) the FI \eref{eq:FI_p_n} 
over all bosonic states $\psi_N\in\Hsym_N$, any unitary transformation of the state becomes irrelevant.
In particular, owing to the parameter being unitary encoded, $\psi_N(\param)$ may then be simply 
replaced by $\psi_N$, so that the average of \eqnref{eq:FI_p_n} manifestly seizes to depend on $\param$.

This observation alone is not sufficient to deduce that the concentration behavior of FI \eref{eq:FI_p_n} is independent of the value of the parameter $\param$. However, with the following theorem we show not only this but actually a significantly stronger statement.
We prove that the FI given in \eqnref{eq:FI_p_n} evaluated on random symmetric states not only typically attains the Heisenberg scaling for a certain value of $\param$, but typically does so for all values of $\param$ \emph{at the same time}.
\begin{thm}[Pure random symmetric qubit states typically attain the HL for all values of $\param$ in the setup of \figref{fig:interferometer}]
\label{thm:inter_meas_fixed_phi}
Let $\psi_N$ be a fixed pure state on $\Hsym_N$ for $d\!=\!2$ modes and $p_{n|\param}(U\psi_NU^\dagger)$ the probability to obtain outcome $n$ given that the value of the unknown phase parameter is $\param$ and the interferometer state is $U\psi_NU^\dagger$ (see also \eqref{eq:p_n}).
Let $\FI(U,\param) \coloneqq \FI(\left\{p_{n|\param}(U\psi_N U^\dagger)\right\})$ be the corresponding FI defined
in \eqnref{eq:FI_p_n}, then
\begin{equation}\label{eq:classfishtyp}
  \Pr_{U \sim \mu\left(\Hsym_N\right)}\left( \exists_{\param\in[0,2\pi]} \FI(U,\param) \leq \Theta(N^2) \right) \leq \mathrm{exp}\left(-\Theta(N) \right) .
\end{equation}
\end{thm}
In other words, the probability that for a random state there exists a value of the 
parameter $\param$ for which $\FI(U,\param)$ does not achieve Heisenberg scaling is 
exponentially small in $N$. Hence, when dealing with typical two-mode bosonic 
states one does not have to resort to LU-optimization 
(similarly, as in \thmref{thm:symmetricstatesareuseful}) 
in order to reveal their 
metrological usefulness \emph{even for a fixed measurement}.
Note that such an optimization would make the problem $\param$-independent. As the interferometric scheme  of \figref{fig:interferometer} 
is restricted to the symmetric subspace $\Hsym_N$, one is allowed to perform only LU operations
of the form $V^{\otimes N}$. However, setting $V=\exp(- \ii\, \theta\, \sigma_z /2)$ 
one may then always shift $\param \to \param +\theta$ to any desired value.

We provide a detailed proof of \thmref{thm:inter_meas_fixed_phi} in \appref{sec:finproofs}, where we also present its more precise and technical version.
One of its constituents---the analysis of the average value of the FI \eref{eq:FI_p_n}---%
may be found in \appref{sec:averagefifortwomodes}, where by rigorously showing that
\begin{equation}
c_-\, N^2\leq \Expect_{U \sim \mu\left(\mathcal{\Hsym}_N\right)} \FI(U,\param) \leq c_+\, N^2+N, 
\end{equation}
with constants $0<c_-<c_+<\infty$, we prove that the average FI indeed asymptotically 
follows the HL-like scaling. Although our derivation allows us only to bound the actual
asymptotic constant factor, we conjecture that $\Expect_{U \sim \mu\left(\mathcal{\Hsym}_N\right)} \FI(U,\param) \to N^2/6$.
In particular, we realise that such behaviour is recovered after replacing the denominators of 
all terms in the sum of \eqnref{eq:FI_p_n} by their average values and also verify our conjecture 
numerically in \secref{sec:numerics} below.

The fact that random symmetric states typically lead to Heisenberg scaling \emph{for all values} of $\param$ in the simple setup of \figref{fig:interferometer} has important consequences.
If this were not the case, it could be possible that for typical symmetric state $\psi_N$  there are values of $\param$ for which the sensitivity was low. Our stronger result however is directly useful for realizable setups:
In real interferometry experiments one typically starts the phase estimation protocol by calibrating the device \cite{Mitchell2004,Nagata2007}.
This is done by taking control of $\param$ and reconstructing the $p_{n|\param}$ \emph{empirically} form measurements with different known values of $\param$.
A tomography of the state $\psi_N$ is then not necessary and an efficient estimator (e.g., max-likelihood \cite{Lehmann1998}) can always 
be constructed that saturates the CRB \eref{eq:CRB} after sufficiently many protocol repetitions.
Crucially, this implies that one might randomly generate (for instance, following the protocol we present in \secref{sec:numerics}) many copies of some fixed random symmetric state $\psi_N$ and, even \emph{without being aware of its exact form}, one can still typically construct an estimator that attains the Heisenberg scaling,
while sensing small fluctuations of the parameter around any given value of $\param$.

\section{Efficient generation of random symmetric states.} \label{sec:numerics}
We have shown that random symmetric states have very promising properties for 
quantum sensing scenarios, but have so far not addressed the question of how to efficiently generate such states.
In this section, we demonstrate that the random symmetric states
can be simulated with help of short random circuits whose outputs
indeed yield, on average, the Heisenberg scaling not only of the QFI but
also of the FI for the measurement scheme depicted in \figref{fig:interferometer}.

Concretely, we consider random circuits over a set of gates that is universal on the special unitary group of the symmetric subspace and consists of four different gates: three beam splitters and a cross-Kerr non-linearity.
A set of gates is said to be universal on a certain unitary group if by taking products of its elements one can obtain 
\emph{arbitrary good} approximations (in trace norm) to any unitary operation in this group. We emphasize that the universality of in the symmetric subspace $\Hsym_N$ is not connected to the notion of universal quantum computation. This follows from the fact that the dimension of $\Hsym_N$  scales polynomially (in the case of $d=2$ modes linearly) in the number of particles $N$ and consequently this space is not sufficient for universal quantum computation.
We first construct a universal set of unitary gates on $\Hsym_N$ for $d=2$, which is inspired 
by operations commonly available when dealing with bosonic (optical and atomic) systems. 
We present our results in the language of two-mode interferometry (see \figref{fig:interferometer}).

We start with the following set of gates:
\begin{equation} \label{eq:linouticuni}
\begin{split}
&V_{1} \coloneqq \frac{1}{\sqrt{5}}\begin{pmatrix}1 & 2\,\ii\\
2\,\ii & 1
\end{pmatrix}, \quad V_2 \coloneqq \frac{1}{\sqrt{5}}\begin{pmatrix}1 & 2\\
-2 & 1
\end{pmatrix}, \\
&V_3 \coloneqq \frac{1}{\sqrt{5}}\begin{pmatrix}1+2\,\ii & 0\\
0 & 1-2\,\ii
\end{pmatrix};
\end{split}
\end{equation}
known to be a ``fast'' universal gate-set for linear optics \cite{Lubotzky1986,Harrow2002,Bouland2014,Carolan2015,Sawicki2015}.
The above matrices reflect how gates act on a single particle (which can be either in mode $a$ or in mode $b$).
The action on $\Hsym_N$ is then given by $\hat{V}_j=V_j^{\otimes N}$ for $j\in\{1,2,3\}$.
We now supplement the above collection by a two-mode gate corresponding to a cross-Kerr nonlinearity (with effective action time $t=\pi/3$) \cite{Imoto1985}.
Concretely, we take 
\begin{equation}
\hat V_{\mathrm{XK}} \coloneqq \exp\left( \frac{-\ii\,\pi\,\hat n_a \hat n_b}{3} \right),
\end{equation}
where $\hat n_{a/b}$ are the particle number operators of modes $a/b$ marked in \figref{fig:interferometer}. For the general method method for checking if a given gate promotes linear optics to universality in $\Hsym_N$ see \cite{moUniv}.

In atom optics, large cross-Kerr nonlinearities and phase shifts (XKPS) can be achieved in ultracold two-component Bose gases in the so-called two-mode approximation \cite{Albiez2005,Mele2011,Julia2012} (see also \cite{Stringari2003,Lewenstein2012}).
In optics, reaching large XKPS is more challenging \cite{Kok2010,Furusawa2011,Wang2012}, but there has recently been a spectacular progress in this area, both on weak \cite{Nemoto2004,He2011} and strong nonlinearities \cite{Chang2014,Lodahl2015,Javadi2015}.
From the theoretical perspective, using the methods of geometric control theory \cite{Zimboras2014,Zimboras2015} and ideas from representation theory of Lie algebras \cite{moUniv} it is possible to prove that the gates $\{\hat V_1,\hat V_2,\hat V_3,\hat V_{\mathrm{XK}} \}$ are universal for $\SU\left(\Hsym_N\right)$. The gate $\hat V_{\mathrm{XK}}$ is not the only gate yielding universality when supplemented with gates universal for linear optics. The comprehensive characterization of (non-linear) gates having this property will be presented in \cite{moUniv}.

A random circuit of depth $K$ over this gate set is now obtained by picking at random (according to a uniform distribution) $K$ gates from the set $\{\hat V_1,\hat V^\dagger_1,\hat V_2,\hat V^\dagger_2,\hat V_3,\hat V^\dagger_3,\hat V_{\mathrm{XK}},{\hat V_{\mathrm{XK}}^\dagger} \}$.
We call states generated by applying such a circuit to some fixed symmetric state random circuit states.

\begin{figure}[!t]
\centering
\includegraphics[width=\columnwidth]{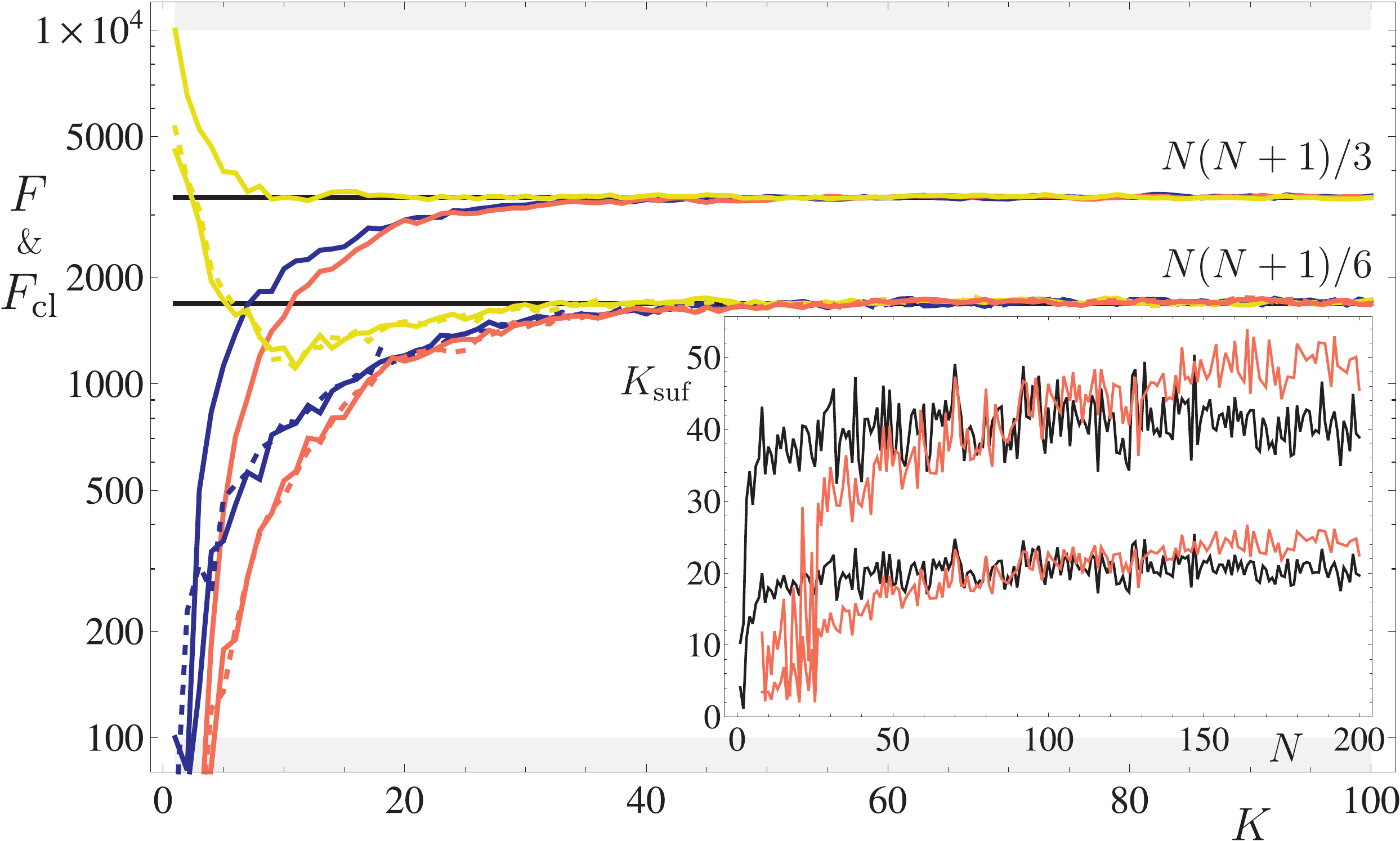}
\caption{(color online)
Convergence of QFI and FI of random circuits states for increasing depth $K$:~%
Main plot shows $\QFI$ and $\FI$ of $N=100$ two-mode photons (indistinguishable qubits) for the measurement from \figref{fig:interferometer} averaged over 150 realizations (sufficient to make the finite sample size irrelevant) of random circuits for different depths $K$.
The starting states before the random circuit are of the form $\ket{\psi_N}=\sum_{n=0}^N\sqrt{x_n}\ket{n,N-n}$ 
with:~$x_0=1$ polarised (red), $x_n=\binom{N}{n}/2^N$ balanced (blue), or $x_{0} = x_{N} =1/2$ N00N (yellow).
Fast convergence to the values $N\,(N+1)/3$ and $N\,(N+1)/6$ (black horizontal lines) is evident in all cases.
The shaded regions mark ``worse than SQL'' and ``better than HL'' precisions.
The $\FI$ curves are plotted for $\param=\pi/2$ (solid) and $\param=\pi/3$ (dotted).
The inset depicts the sufficient circuit depth $K_{\t{suf}}$ as a function of $N$, such that the 
corresponding sample-averaged QFI (black curves) or FI (red curves)
is at most $1\%$ (top curves) or $10\%$ (bottom curves) from its typical value.
As $K_{\t{suf}}$ grows at most mildly with $N$, for realistically achievable photon 
numbers \cite{Carolan2015}, $K \approx 20$ may be considered sufficient.}
\label{fig:saturation}
\end{figure}

Our intuition that the above scheme should generate unitaries distributed approximately according to $\mu\left(\Hsym_N\right)$ comes from the theory of the so-called $\epsilon$-approximate unitary $t$-designs \cite{Low2010a,Brandao2012,Mozrzymas2013}.
There are several essentially equivalent ways to define unitary designs \cite{Low2010a}.
One of them, introduced in \refcite{Low09}, implies that an $\epsilon$-approximate unitary $t$-design $\mu_{\epsilon,t}(\H)$ is a distribution over the unitary group $\SU(\H)$ acting on a Hilbert space $\H$ (of dimension $|\H|$) that efficiently approximates the Haar measure $\mu(\H)$ in a way such that for all balanced monomials $f:\SU(\H) \to \R$ of degree $t$, it holds that
\begin{equation}
 \left|\Expect_{U \sim \mu(\H)} f(U) - \Expect_{U \sim \mu_{\epsilon,t}(\H)}f(U)\right| \leq \epsilon/|\H|^t \,.
\end{equation}
Moreover, if such a function satisfies a concentration inequality of the form given in \eqnref{eq:concentrsu1} with respect to the Haar measure, then an (albeit weaker) concentration also holds with respect the design \cite{Low09}.
All these statements carry over in a similar form to balanced polynomials.
Their corresponding difference may always be bounded by the weighted sum of the differences of their constituting monomials.
The iso-spectral QFI, $\QFI(U) = \QFI(U\,\sigma_N\,U^\dagger, H)$, introduced 
in \thmref{thm:symmetricstatesareuseful} as a function of unitary rotations is precisely such a balanced polynomial of order two, which can be seen directly from \eqnref{eq:QFI_formula}.

For distinguishable qudits there are efficient methods to generate approximate designs by using random circuits over local universal gate-sets on $\H_N$ \cite{Low2010a,Brandao2012}.
These constructions unfortunately do not immediately carry over to the symmetric subspace $\Hsym_N$ of $N$ qubits.
However, one can use the fact proven in \refcite{Brandao2012} (based on results of \refcite{Bourgain2011}) that in any Hilbert space $\mathcal{H}$ sufficiently long random circuits over a set of universal gates form an $\epsilon$-approximate unitary $t$-design.
More precisely, this holds whenever the gates employed in the circuit are non-trivial and have algebraic entries.
The set of gates universal in $\Hsym_N$ given above satisfies this condition.
For this to hold, it would actually be sufficient to replace all the three gates $\hat V_1, \hat V_2, \hat V_3$ by a single non-trivial beam splitter \cite{Bouland2014} and the gate $\hat V_{\mathrm{XK}}$ by essentialy non-trivial non-linear gate.
The latter would not even have to be a reproducibly implementable non-linear operation, but its strength could even be allowed to vary from invocation to invocation.
However, it is mathematically very difficult to analytically bound the depth $K$ of such circuits that is sufficient to achieve a given $\epsilon$ for a given $t$.

\begin{figure}[!t]
\centering
\includegraphics[width=\columnwidth]{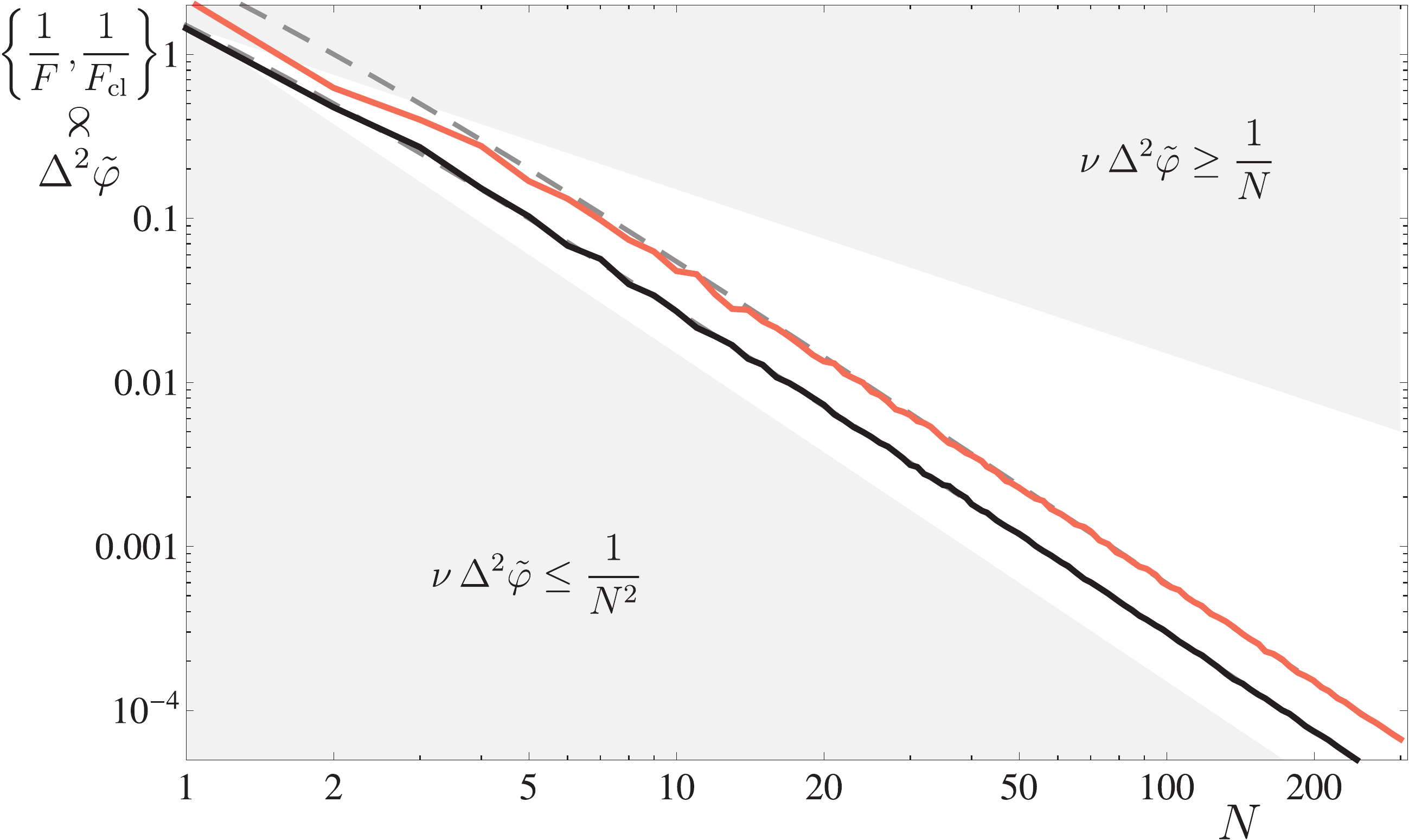}
\caption{(color online)
Mean squared error attained by random bosonic states generated by sufficiently deep 
random circuits:~We depict the ultimate limit of the resolution $\nu\,\Delta^2\tilde\param$ 
attainable with random circuit states generated by applying a deep random circuit ($K=60$) 
onto two-mode balanced ($x_n=\binom{N}{n}\tfrac{1}{2^N}$ in \figref{fig:saturation}) 
state for both the interferometric measurement of \figref{fig:interferometer} (red)
and the theoretically optimal one yielding the QFI (black). The corresponding sample-averaged
FI and QFI quickly concentrate around the typical values $N^2/3$ and $N^2/6$ (dashed lines) 
respectively. The shaded regions mark the ``worse than SQL'' and ``better than HL'' precisions.}
\label{fig:scaling}
\end{figure}

For this reason we resort to a numerical analysis to verify how rapidly with increasing $K$ the average QFI and the FI of random circuit states converge to the respective averages for random symmetric states.
We consider the scenario of \figref{fig:interferometer}.
In this two-mode case $d=2$ it holds that $|\Hsym_N|=N+1$ and from \eqnref{eq:FI_p_n} we obtain concentration of the QFI around the value $\Expect\QFI \coloneqq \Expect_{U \sim \mu\left(\H\right)}\QFI(U) = N(N+1)/3$ 
\footnote{%
This confirms the scaling obtained numerically in \refcite{Tichy2015} and is consistent with $2N(N+2)/3$ found in \refcite{Froewis2014}, in which 
the states were taken to moreover be mode-symmetric.%
}. 
For the FI, we expect to find $\Expect\FI \coloneqq \Expect_{U \sim \mu\left(\H\right)}\FI(U,\param) = N\,(N+1)/6$ 
(for details see the discussion after \thmref{thm:inter_meas_fixed_phi} and \appref{sec:averagefifortwomodes}).
In \figref{fig:saturation}, we show explicitly that indeed random circuit states generated according to our recipe allow to reach these values already for moderate $K$. 

In \figref{fig:scaling}, we verify further the behavior of the ultimate limits on the attainable precision via the relevant CRB (see \eqnref{eq:CRB}) dictated by the attained values for $1/\QFI$ and $1/\FI$.
We observe that the ultimate bounds predicted by the average $\Expect\QFI$ and $\Expect\FI$ for random symmetric states are indeed saturated quickly (here, $K=60$) and, crucially, both reach the predicted Heisenberg scaling.
This demonstrates that it is possible to generate states that share the favorable metrological properties of Haar-random symmetric states via the \emph{physical} processes of applying randomly selected optical gates.

\section{Conclusions}\label{sec:conclusions}
In this work we present a systematic study of the usefulness of random states for quantum metrology.
We show that random states, sampled according to the Haar measure from the full space of states of distinguishable particles, are typically not useful for quantum enhanced metrology.
In stark contrast, we prove that states from the symmetric subspace have many very promising properties for quantum metrology:~%
They typically achieve Heisenberg scaling of the quantum Fisher information and this scaling is robust against particle loss 
and equally holds for very mixed isospectral states. Moreover, we show that the high quantum Fisher information of 
such random states can actually be exploited with a single fixed measurement that 
is implementable with a beam splitter and particle-number detectors.
Finally, we also demonstrate that states generated with short random circuits can be 
used as a resource to achieve a classical Fisher information with the same scaling as the Heisenberg limit.
Our results on random symmetric states open up new possibilities for quantum enhanced metrology. 

Our work, being a study initiating a new research direction, naturally leads raises a number of interesting questions:
From the physical perspective it would be important to investigate the impact of more realistic noise types, 
such as:~local (and correlated) dephasing, depolarisation \cite{Froewis2014} and particle loss on the classical Fisher information in the interferometric scenario considered in \secref{sec:phys}, as well as the quantum Fisher information in general for finite $N$.
Further, it would be interesting to see whether bosonic random states are also useful for multi-parameter sensing problems with non-commuting generators \cite{Baumgratz2016}.
An important part of the quantum metrology research is devoted to infinite dimensional optical 
systems with the \emph{mean} number of particles---corresponding to the power of a light beam---being fixed \cite{Demkowicz2015} e.g.,
in squeezing-enhanced interforometry with strong laser
beams of constant power \cite{Lang2013}. 
Here one could ask whether states prepared via random Gaussian transformations \cite{Monras2010,Pinel2012,*Pinel2013} or random circuits of gates universal for linear optics are typically useful for metrology. 
Another relevant problem beyond our analysis is the speed of 
convergence to the approximate designs when considering performance of the states prepared with 
random bosonic circuits discussed in \secref{sec:numerics}. 
Further it is interesting to further study properties of the ensembles of random states generated from the Haar-random pure states and possibly particle loss, of both bosons and fermions.
Lastly, a natural question to be asked is whether the typical metrological usefulness of random 
bosonic states remains valid if one considers general phase estimation scenarios, e.g., 
single-shot protocols with no prior knowledge assumed about the parameter value \cite{Berry2000}, 
for which Bayesian inference methods must be employed to quantify the attainable 
precision \cite{Bagan2001b,Bartlett2007}.

\begin{acknowledgments}
We would like to thank Martin Kliesch for inspiring discussions. We acknowledge support from the European Research Council
(ERC AdG OSYRIS and CoG QITBOX), Axa Chair in Quantum Information Science, 
John Templeton Foundation, EU (IP SIQS and QUIC), Spanish National Plan FOQUS 
(FIS2013-46768), MINECO (Severo Ochoa Grant No.~SEV-2015-0522),
Fundaci\'{o} Privada Cellex, and Generalitat de Catalunya (Grant No.~SGR 874 and 875). We are greatful to Centro de Ciencias de Benasque Pedro Pascual fo the hospitality during the "Quantum information" conference in 2015.
M.\ O.~acknowledges the START scholarship granted by Foundation for Polish Science 
and NCN Grant No.~DEC-2013/09/N/ST1/02772, J.~K. and C.\ G. acknowledge funding from the European Union’s Horizon 2020 research 
and innovation programme under the Marie Sk\l{}odowska-Curie Q-METAPP and NIMBqUS Grants No.~655161 and 700140.
C.~G. further acknowledges support by MPQ-ICFO and ICFOnest+ (FP7-PEOPLE-2013-COFUND).
R. A. acknowledges funding from the European Union's Horizon 2020 research and
innovation programme under the Marie Sk\l{}odowska-Curie N-MuQuaS Grant No. 705109.
\end{acknowledgments}

\bibliography{concmetro}

%
\onecolumngrid
\newpage
\appendix
\part*{Appendices}

Here we give the details that are needed to obtain the main results given in the main text.
In the \secref{sec:concentration} we discuss concentration of measure on the special unitary group, and give bounds on the Lipschitz constants of the relevant functions on this group.
In \secref{sec:lowerbound} we prove a lower bound on the QFI that are useful when studying particle losses.
In \secref{sec:averages} we give bounds for averages of FI and QFI on the relevant ensembles of density matrices that we consider---isospectral density matrices of distinguishable particles, random symmetric (bosonic) states of identical particles, and random bosonic states that underwent particle loss (\secref{sec:average_qfi_for_bosons_with_particle_losses}).
In \secref{sec:finproofs} we use the previously derived technical results to prove Theorems \ref{thm:moststatesareuseless}, \ref{thm:symmetricstatesareuseful}, \ref{thm:symmetriclosses}, and \ref{thm:inter_meas_fixed_phi} in the main text.
In \secref{sec:partloss} we prove the equivalence of the beam-splitter model of particle losses and the operation of taking partial traces over particles contained in two-mode bosonic systems.

\begin{table}[h]
\begin{centering}
\begin{tabular}{|c|c|}
\hline
\textbf{Symbol} & \textbf{Explanation}\tabularnewline
\hline
$\Hloc$ & local Hilbert space\tabularnewline
\hline
$d = |\Hloc|$ & dimension of the local Hilbert spac\tabularnewline
\hline
$\H_N = \H^{\otimes N}$ & Hilbert space of $N$ distinguishable particles \tabularnewline
\hline
$D=\left|\H_N\right|$ & dimension of the space of $N$ distinguishable particles\tabularnewline
\hline
$\Hsym_N = \lspan\{\ket{\psi}^{\otimes N}: \ket{\psi} \in \H\}$ & Hilbert space of $N$ bosons\tabularnewline
\hline
$\H$ & general Hilbert space\tabularnewline
\hline
 $|\H|$ & dimension of the general Hilbert spac\tabularnewline
\hline
$h$ & local Hamiltonian encoding the phase $\param$\tabularnewline
\hline
$H = H_N = \sum_{j=1}^N h^{(j)}$ & Hamiltonian acting on $N$ particles\tabularnewline
\hline
$\mathcal{D}\left(\mathcal{H}\right)$ & set of states on the Hilbert space $\mathcal{H}$\tabularnewline
\hline
$\1$ & identity operator on the relevant Hilbert space\tabularnewline
\hline
$\rho,\sigma,\ldots$ & symbols denoting (in general) mixed states\tabularnewline
\hline
$\psi,\phi,\ldots$ & symbols denoting pure states\tabularnewline
\hline
$\FI(\left\{p_{n|\param}\right\})$ & Classical Fisher information associated to the family of probability distributions $\left\{p_{n|\param}\right\}$ \tabularnewline
\hline
$\QFI\left(\rho,H\right)$ & Fisher information computed for the state $\rho$ with respect to the
Hamiltonian $H$\tabularnewline
\hline
$\dB\left(\rho,\sigma\right)$ & Bures distance between states $\rho$ and $\sigma$\tabularnewline
\hline
$\mathcal{F}\left(\rho,\sigma\right)$ & Uhlmann fidelity between states $\rho$ and $\sigma$\tabularnewline
\hline
$\mathbb{P}^{\mathrm{sym}}$ & orthogonal projector onto $\mathrm{Sym}^{2}\left(\mathcal{H}\right)$\tabularnewline
\hline
$\mathbb{P}^{\mathrm{asym}}$ & orthogonal projector onto $\bigwedge^{2}\left(\H\right)$\tabularnewline
\hline
$\mathrm{End}\left(\mathcal{H}\right)$ & set of linear operators on $\H$\tabularnewline
\hline
$\mathrm{Herm}\left(\mathcal{H}\right)$ & set of Hermitian operators on $\H$\tabularnewline
\hline
$\SU\left(\mathcal{H}\right)$ & special unitary group on $\H$\tabularnewline
\hline
$\mu(\H)$ & Haar measure on $\SU\left(\H\right)$\tabularnewline
\hline
$\Expect_{U \sim \mu\left(\mathcal{\mathcal{H}}\right)}$ & expectation value (average) with respect to $\mu(\H)$\tabularnewline
\hline
$\Omega$ & set of isospectral density matrices in $\mathcal{H}$ (for the specified
ordered spectrum) \tabularnewline
\hline
\hline
\textbf{Acronym} & \textbf{Explanation}\tabularnewline
\hline
QFI & Quantum Fisher Information \tabularnewline
\hline
FI & classical Fisher Information \tabularnewline
\hline
GHZ & Greenberger–Horne–Zeilinger state \tabularnewline
\hline
SQL & Standard Quantum Limit \tabularnewline
\hline
HL & Heisenberg Limit \tabularnewline
\hline
POVM & Positive Operator-Valued Measure \tabularnewline
\hline
\end{tabular}

\par\end{centering}

\caption{Notation used throughout the paper, unless indicated differently.}
\end{table}

\section{Concentration of measure on special unitary group and Lipschitz constants for QFI and FI} \label{sec:concentration}
In this section we first present basic concentration of measure inequalities on the special unitary group $\SU\left(\mathcal{H}\right)$. Then, we give bounds on the Lipschitz constants of various functions based on QFI and FI that appear naturally, while studying different statistical ensembles of states on $\mathcal{H}$---isospectral density matrices, partially traced isospectral density matirices etc.

\subsection{Concentration of measure on unitary group}
We will make an extensive use of the concentration of measure phenomenon on the special unitary group $\SU\left(\H\right)$. It will be convenient to use a metric tensor $g_{\mathrm{HS}}$
induced on $\SU\left(\H\right)$ from the embedding of $\SU\left(\H\right)$ in the set of all linear operators, $\mathrm{End}\left(\H\right)$, equipped with the Hilbert-Schmidt
inner product $\left\langle A,B\right\rangle =\tr\left(A^{\dagger}B\right)$.
Let us write the formula for $g_{\mathrm{HS}}$ explicitly. The special
unitary group is a Lie group and thus for every $U\in\SU\left(\H\right)$ we
have an isomorphism $T_{U}\SU\left(\H\right)\approx\mathfrak{su}\left(\H\right)$,
where $T_U$ is the tangent space to $\SU$ at $U$ and $\mathfrak{su}\left(\H\right)$ is the Lie algebra of the group consisting of Hermitian traceless operators on $\H$.
The linear isomorphism is given by the following mapping
\begin{equation}
\mathfrak{su}\left(\H\right)\ni X\longmapsto\hat{X}=\left.\frac{d}{d\param}\right|_{\param=0}\exp\left(-\ii\param X\right)U\in T_{U}\SU\left(\H\right)\,.
\label{eq:identification special unitary}
\end{equation}
Using the identification \eqnref{eq:identification special unitary}
and treating the operator $\hat{X}=\left.\frac{d}{d\param}\right|_{\param=0}\exp\left(-\ii\param X\right)U=-\ii XU$
as an element of $\mathrm{End}\left(\H\right)$ we get
\begin{equation}
g_{\mathrm{HS}}\left(\hat{X},\hat{Y}\right)=\left\langle \ii XU,\,\ii YU\right\rangle =\tr\left(\left[\ii XU\right]^{\dagger}\ii YU\right)=\tr\left(XY\right)\,,\label{eq:explicit formula metric unitary}
\end{equation}
where we have used $X^{\dagger}=X$, the identity $UU^{\dagger}=\1$,
and the cyclic property of the trace.  The gradient of a smooth function $f:\mathrm{SU}\left(\H\right)\rightarrow\mathbb{R}$ at point $U\in\mathrm{SU}\left(\H\right)$ is defined by the condition
\begin{equation}\label{gradientdef}
g_{\mathrm{HS}}\left(\hat{\left.\nabla f\right|_{U}},\hat{X}\right)\ =
 \left.\frac{d}{d\phi}\right|_{\phi=0}f\left(\exp\left(-\ii \phi X\right)U\right) \ , 
\end{equation}
that has to be satisfied for all $X\in\mathfrak{su}\left(\H\right)$.

\begin{fact}[Concentration of measure on $\SU\left(\mathcal{H}\right)$]\label{concform}
\cite{Anderson2010}
Consider a special unitary group $\SU\left(\H\right)$
equipped with the Haar measure $\mu$ and the metric $g_{\mathrm{HS}}$.
Let
\begin{equation}
f:\SU\left(\H\right)\longmapsto\mathbb{R}
\end{equation}
be a smooth function on $\SU\left(\H\right)$ with
the mean $\mathbb{E}_{\mu}f$, and let
\begin{equation}
L=\sqrt{\max_{U\in\mathcal{\SU\left(\H\right)}}g_{\mathrm{HS}}\left(\nabla f,\nabla f\right)}\label{eq:lipschitz constant formula}
\end{equation}
be the Lipschitz constant of $f$. Then, for every $\epsilon\geq0$,
the following concentration inequalities hold
\begin{align}
\Pr_{U \sim \mu\left(\Hsym_N\right)}\left( f\left(U\right)-\Expect_{U \sim \mu\left(\mathcal{\mathcal{H}}\right)}f\geq\epsilon \right)&\leq\mathrm{exp}\left(-\frac{D\epsilon^{2}}{4L^{2}}\right) \ , \\
\Pr_{U \sim \mu\left(\Hsym_N\right)}\left(f\left(U\right)-\Expect_{U \sim \mu\left(\mathcal{\mathcal{H}}\right)}f\leq-\epsilon\ \right)&\leq\mathrm{exp}\left(-\frac{D\epsilon^{2}}{4L^{2}}\right)\,, \label{eq:su concentration}
\end{align}
where $D=|\H|$ is the dimension of $\H$.
\end{fact}

\begin{fact}\label{GenLip}
Concentration inequalities \eqnref{eq:su concentration} hold also for the general (not necesarily smooth) $L$-Lipschitz functions
$f:\SU\left(\mathcal{H}\right)\longmapsto\mathbb{R}$ \cite{Anderson2010}, that is functions satisfying
\begin{equation}
\left|f\left(U\right)-f\left(V\right)\right|\leq L\,\mathrm{d}\left(V,\,W\right)\,.\label{eq:general lipsch}
\end{equation}
where $\mathrm{d}\left(V, W\right)$ is the geodesic distance between unitaries $U$ and $V$ given by
\begin{equation}
\mathrm{d}\left(U, V\right) \coloneqq \inf_{\gamma: \gamma(0)=U, \gamma(1)=V} \mathrm{D}_{\gamma}\,,\label{eq:geodesic distance}
\end{equation}
with $\mathrm{D}_{\gamma} \coloneqq  \int_{\left[0, 1\right]}\sqrt{g_{\mathrm{HS}}\left(\frac{d\gamma}{dt}, \frac{d\gamma}{dt}\right)}$,
and the infimum is over the (piecewise smooth) curves $\gamma$ that start at $U$ and end at $V$.
\end{fact}

\begin{rem}\label{Lipremark}
Due to the definition of the gradient $\nabla f$ \eqref{gradientdef} and the structure of the tangent space $T_{U}\SU\left(\mathcal{H}\right)$ for $U\in\SU\left(\mathcal{H}\right)$ (see \eqnref{eq:identification special unitary}) we have
\begin{equation}\label{gradientfun}
\tr\left(\left.\nabla f\right|_{U} X\right)=\left.\frac{d}{d\phi}\right|_{\phi=0}f\left(\exp\left(-\ii \phi X\right)U\right)\,,
\end{equation}
where $X=X^{\dagger}$ and $\tr{X}=0$. Assume that for $C>0$ we can find the bound
\begin{equation}
\left| \left.\frac{d}{d\phi}\right|_{\phi=0}f\left(\exp\left(-\ii \phi X\right)U\right) \right| \leq C\cdot \|X \|_{\mathrm{HS}}\ ,
\end{equation}
which is valid for all $U\in\SU\left(\mathcal{H}\right)$. Then, from \eqnref{gradientfun} we can conclude that $C$ is an upper bound on the Lipschitz constant of $f$.
\end{rem}

\subsection{Lipschitz constants for the quantum Fisher information for the general Hamiltonian encoding}
Recall that for unitary encodings the quantum Fisher information for a state $\rho$ with spectral decomposition $\sum_i p_i \kb{e_i}{\e_i}$
\begin{equation}
 \QFI\left(\rho,H\right)=2\cdot\sum_{i,j: p_{i}+p_{j}\neq0}\frac{\left(p_{i}-p_{j}\right)^{2}}{p_{i}+p_{j}}\left|\braket{e_{i}}{H\left|e_{j}\right.}\right|^{2},\label{eq:QFI_formulaman}
\end{equation}
where $H$ is the Hamiltonian generating the unitary evolution of mixed states
\begin{equation}
 \mathbb{R}\ni\param\longmapsto\rho\left(\param\right) \coloneqq \exp\left(-i\param H\right)\,\rho\,\exp\left(i\param H\right)\in\mathcal{D}\left(\H\right).\label{eq:trajectory}
\end{equation}
The QFI depends on both the state $\rho\in\mathcal{D}\left(\H\right)$, and the Hamiltonian $H\in\mathrm{Herm}\left(\H\right)$
encoding the phase $\param$. In what follows, without any loss of
generality, we assume that $\tr\left(H\right)=0$. We are interested in the behavior of $F\left(\rho,H\right)$ when $H$
is fixed and $\rho$ varies over some ensemble of (generally mixed) states. As we want to use concentration inequalities (of the type of \eqnref{eq:su concentration}), our aim here is to give bounds on Lipschitz constant of QFI on relevant sets of density matrices.

We first study QFI on the set of isospectral density matrices
\begin{equation}
\Omega_{\left(p_{1},\ldots,p_{D}\right)} \coloneqq \left\{\rho\in\mathcal{D}\left(\mathcal{H} \right) |  \mathrm{sp}_{\uparrow}\left(\rho\right)=\{p_j\}_j \right\} ,
\end{equation}
where $\mathrm{sp}_{\uparrow}\left(\rho\right)$ denotes the vector on non increasingly ordered eigenvalues of $\rho$. In what follows, for the sake of simplicity we will use the shorthand notation $\Omega_{\left(p_{1},\ldots,p_{D}\right)}\coloneqq \Omega$. Let
\begin{equation}
F_{\Omega,H}:\SU\left(\H\right)\ni U\longmapsto F\left(U\rho_{0}U^{\dagger},H\right)\in\mathbb{R} \label{eq:auxiliary Fisher} \ ,
\end{equation}
where $\rho_{0}$ is the arbitrary chosen state belonging to $\Omega$. Then, we can prove the following lemma.

\begin{lem}
\label{important Lipschitz}
The Lipschitz constant (with respect to $g_{\mathrm{HS}}$) of the function $F_{\Omega,H}$ defined by \eqnref{eq:auxiliary Fisher}
is upper bounded by
\begin{equation}\label{Lipschitzcons}
L_\Omega\leq\mathrm{min}\left\{1,2\sqrt{2}\cdot\sqrt{\dB\left(\rho,\tfrac{\1}{D}\right)}\right\}\cdot32\cdot\left\| H\right\| ^{2}\ ,
\end{equation}
where $\1/D$ is the maximally mixed state on $\mathcal{H}$.
\end{lem}
\begin{rem}
The quantity $\dB\left(\rho,\1/D\right)$ depends only on the spectrum of $\rho$ and thus is constant on the set of isospectral density matrices $\Omega$.
\end{rem}

\begin{proof}
In a recent paper by \cite{Augusiak2016} the following inequality was proven:
\begin{equation}\label{continouityfish}
\left|F\left(\rho,H\right)-F\left(\sigma,H\right)\right|\leq 32\cdot \dB\left(\rho,\sigma\right)\left\| H\right\| ^{2}\,,
\end{equation}
where $\dB\left(\rho,\sigma\right)=\sqrt{2\left[1-\mathcal{F}\left(\rho,\sigma\right)\right]}$
is the Bures distance between density matrices with $\mathcal{F}(\rho,\sigma)=
\tr\sqrt{\sigma^{1/2}\rho\sigma^{1/2}}$ denoting the fidelity. Inserting
\begin{equation}
\rho=U\rho_{0}U^{\dagger}\quad \text{and} \quad \sigma=\exp\left(-\ii\phi X\right)U\rho_{0}U^{\dagger}\exp\left(\ii\phi X\right)\label{eq:auxiliary states}
\end{equation}
into \eqnref{continouityfish}, one obtains
\begin{equation}
\left|F_{\Omega,H} \left(\exp\left(-\ii \phi X\right)U\right)-F_{\Omega,H} \left(U\right)\right|\leq32\cdot \dB\left(U\rho_{0}U^{\dagger},\exp\left(-\ii \phi X\right)U\rho_{0}U^{\dagger}\exp\left(\ii \phi X\right)\right)\left\| H\right\| ^{2}\,.
\end{equation}
Dividing the above by $|\phi|$ and taking the limit $\phi\rightarrow0$, one
arrives at
\begin{equation}\label{nexteslip}
\left|\left.\frac{d}{d\phi}\right|_{\phi=0}F_{\Omega,H}\left(\exp\left(-\ii \phi X\right)U\right)\right|\leq32\left\| H\right\| ^{2}\,\mathrm{lim}_{\phi\rightarrow0}\frac{1}{|\phi|}\dB\left(U\rho_{0}U^{\dagger},\exp\left(-\ii \phi X\right)U\rho_{0}U^{\dagger}\exp\left(\ii \phi X\right)\right)\,.
\end{equation}
Then, the \eqnref{eq:QFI_geom} implies that
\begin{equation}\label{infbures}
\mathrm{lim}_{\phi\rightarrow0}\frac{1}{|\phi|}\dB\left(\rho,\exp\left(-\ii \phi X\right)\rho\,\exp\left(\ii \phi X\right)\right)=\frac{1}{2}\sqrt{F\left(\rho,X\right)},
\end{equation}
and therefore
\begin{equation}\label{finalestim}
\left|\left.\frac{d}{d\phi}\right|_{\phi=0}F_{\Omega,H}\left(\exp\left(-\ii \phi X\right)U\right)\right|\leq 16 \cdot \left\|H\right\|^2\cdot \sqrt{F\left(\rho,X\right)}.
\end{equation}
In addition, we have two upper bounds on the square root of the quantum Fisher information:
\begin{equation}\label{es1}
\sqrt{F\left(\rho,H\right)}\leq 2\cdot\left\| H\right\| \leq 2\cdot\left\| H\right\| _{\mathrm{HS}} ,
\end{equation}
and
\begin{equation}\label{es2}
\sqrt{F\left(\rho,H\right)}\leq 4\sqrt{2}\cdot\sqrt{\dB\left(\rho,\tfrac{\1}{D}\right)} \left\| H\right\| \leq 4\sqrt{2}\cdot\sqrt{\dB\left(\rho,\tfrac{\1}{D}\right)} \left\| H\right\| _{\mathrm{HS}},
\end{equation}
where \eqnref{es1} follows from the fact that the maximal value of the QFI for the phase encoded via the Hamiltonian $H$ is bounded from above  by $4\| H\|^2$ \cite{Giovannetti2011}, and \eqnref{es2} follows from \eqnref{continouityfish} for $\sigma=\1/D$, for which the QFI trivially vanishes. Combining inequalities \eqnsref{es1}{es2} with \eqnref{finalestim} and using \remref{Lipremark}, one finally obtains \eqnref{Lipschitzcons}.
\end{proof}

\begin{rem}
  The upper bound on the Lipschitz constant of $F_{\Omega,H}$ given in \eqnref{Lipschitzcons} depends explicitly on the spectrum of the considered set of isospectral density matrices. Specifically, the right-hand side of \eqnref{Lipschitzcons} decreases as $\rho$ becomes more mixed. For special cases of Haar-random pure states and random depolarized states, see below, we can get better bounds on the Lipschitz constant of the QFI.
\end{rem}

\begin{lem}
\label{improved Lipschitz}
Consider the ensemble of Haar-random depolarized pure states,
\begin{equation}\label{depol2}
\rho = \left(1-p \right)\psi +p\frac{\1}{D},
\end{equation}
where $\psi$ stands for the projector onto a Haar-random pure state $\ket{\psi}$ and $p\in[0,1]$. For fixed $p$, the states in \eqnref{depol2} form an ensemble of isospectral density matrices since for any such $\rho$ we have
\begin{equation}\label{spectrumdepol}
\mathrm{sp}_\uparrow\left(\rho\right)=\left(1-p+ \tfrac{p}{D},\tfrac{p}{D},\ldots, \tfrac{p}{D}\right) \ .
\end{equation}
For this particular spectrum the Lipschitz constant $L_p$ (with respect to $g_{\mathrm{HS}}$) of the function $F_p\coloneqq F_{\Omega,H}$, defined by \eqnref{eq:auxiliary Fisher}, is upper bounded by
\begin{equation}\label{Lipschitzconsdep}
L_p\leq 16\cdot \frac{\left(1-p\right)^2}{1-p+\tfrac{2\,p}{D}}\left\| H\right\| ^{2}\ .
\end{equation}
\end{lem}

\begin{proof}
Let us first note that for $\rho$ given by \eqnref{depol2}, the QFI takes the form \cite{Toth2014}:
\begin{equation}\label{depolFisher}
F\left(\rho,H\right)=\frac{\left(1-p\right)^2}{1-p+\tfrac{2\,p}{D}}\cdot F\left(\psi,H\right),
\end{equation}
from which it directly follows that for all $U\in\SU\left(\mathcal{H}\right)$,
\begin{equation}
F_p\left(U\right)=\frac{\left(1-p\right)^2}{1-p+\tfrac{2\,p}{D}} F_p\left(U\right),
\end{equation}
and consequently,
\begin{equation}\label{rellip}
L_p = \frac{\left(1-p\right)^2}{1-p+\tfrac{2\,p}{D}}L_0.
\end{equation}
One can estimate $L_0$ by exploiting the fact that for pure states the QFI is
simply $F(\psi_0,H)=4\{\tr(\psi_0H^2)-[\tr(\psi_0H)]^2\}$, which allows one to express $F_0(U)$ as
\begin{equation}
F_0\left(U\right)=\tr\left[\left(U\otimes U \psi_0 \otimes \psi_0 U^\dagger \otimes U^\dagger\right)V\right],
\end{equation}
where $V=4\cdot(H^2 \otimes \1 - H\otimes H)$. By the virtue Lemma 6.1 of \cite{moPHD} (see also \cite{Oszmaniec2014}), the Lipschitz constant of $F_0$ is bounded by $2\cdot\left\|V\right\|\leq 16\cdot\left\|H\right\|^2$. Combining this with \eqnref{rellip} yields \eqnref{Lipschitzconsdep}.
\end{proof}
%

It is also possible to prove the Lipschitz continuity of the optimized version of  QFI on $\Omega$,
\begin{equation}
F^\mathrm{\mathcal{V}}_{\Omega,H}:\SU\left(\H\right)\ni U\longmapsto \mathrm{sup}_{V\in\mathcal{V}}F_{\Omega,H}\left(VU\right)\in\mathbb{R} \label{eq:auxiliary Fisher optim} \ ,
\end{equation}
where $\mathcal{V}\subset\SU\left(\mathcal{H}\right)$ is a compact class of unitary gates on $\mathcal{H}$.

\begin{lem}
\label{important Lipschitz opt}
The Lipschitz constant $L_{\Omega}^{\mathcal{V}}$ (with respect to the geodesic distance) of the function $F^\mathcal{V}_\Omega$ defined by \eqnref{eq:auxiliary Fisher optim}
is upper bounded by the Lipschitz constant of $F_\Omega$,
\begin{equation}\label{optimizedlip}
L_{\Omega}^{\mathcal{V}}\leq L_{\Omega} \ .
\end{equation}
\end{lem}
\begin{proof}
Let $U,U'\in\SU\left(\mathcal{H}\right)$. Without loss of generality we can assume $F^\mathcal{V}_{\Omega,H}\left(U\right)\geq F^\mathcal{V}_{\Omega,H}\left(U'\right)$. Let $V_0\in\mathcal{V}$ be the element such that
\begin{equation}
F_{\Omega,H}\left(V_0U\right)=F^\mathcal{V}_{\Omega,H}\left(U\right)=\mathrm{sup}_{V\in\mathcal{V}}F_{\Omega,H}\left(VU\right) \ .
\end{equation}
Consequently, we have the following inequalities
\begin{equation}
\left|F^{\mathcal{V}}_{\Omega,H}\left(U\right)-F^{\mathcal{V}}_{\Omega,H}\left(U'\right)\right|=F_{\Omega,H}\left(V_0U\right)-F^{\mathcal{V}}_{\Omega,H}\left(U'\right)\leq F_{\Omega,H}\left(V_0U\right) - F_{\Omega,H}\left(V_0U'\right) \leq L_\Omega \cdot \mathrm{d}\left(U,U'\right) \ ,
\end{equation}
where in the last inequality we used Lipschitz continuity of $F_{\Omega,H}$ which is guaranteed by Lemma \ref{important Lipschitz}.
\end{proof}
\begin{rem}
For us the case of the greatest interest is $\mathcal{H}=\mathcal{H}_N$ and $\mathcal{V}=\mathrm{LU}$ (local unitary group on $N$ distinguishable particles).
\end{rem}

\subsection{Lipschitz constants for the quantum Fisher information with particle losses}
We now give bounds on the Lipschitz constant of the QFI in the case of particle losses for bosonic states. Recall that in this setting the Hamiltonian acting on $N$ particles is given by $H_{N}=\sum_{i=1}^{N}h^{(i)}$ and that the Hilbert space of the system is the totally symmetric space of $N$ particles denoted by $\Hsym_N$. Let us define a function
\begin{equation}
F^{[k]}_{\Omega}:\SU\left(\Hsym_N\right)\ni U\longmapsto F\left(\tr_k\left(U\rho U^{\dagger}\right),H_{N-k}\right)\in\mathbb{R} \label{eq:auxiliary Fisher particleloss}\ .
\end{equation}
\begin{lem}
\label{important Lipschitz kloss}
The Lipschitz constant (with respect to $g_{\mathrm{HS}}$) of the function $F^{[k]}_\Omega$ defined by \eqnref{eq:auxiliary Fisher particleloss}
is upper bounded by
\begin{equation}\label{Lipschitzconsmap}
L^{[k]}_\Omega\leq\mathrm{min}\left\{1,2\sqrt{2}\cdot\sqrt{\dB\left(\rho,\tfrac{\mathbbm{P}_{\mathrm{sym}}^N}{|\Hsym_N|}\right)}\right\}\cdot32\cdot\left\| H_{N-k}\right\| ^{2}\ ,
\end{equation}
where $\mathbbm{P}_{\mathrm{sym}}^N/|\Hsym_N|$ is the maximally mixed state on $\Hsym_N$ and $\mathbbm{P}_{\mathrm{sym}}^N$ stands for the projector onto $\Hsym_N$.
\end{lem}
\begin{proof}
We prove \eqnref{Lipschitzconsmap} in an analogous way to \eqnref{Lipschitzcons}. Let $\rho'=\tr_{k}\left(\rho\right)$ and $\sigma'=\tr_{k}\left(\sigma\right)$ be two states on $\Hsym_{N-k}$ obtained by tracing out $k$ particles from $\rho$ and $\sigma$, respectively. Applying the inequality \eqnref{continouityfish} to $\rho'$ and $\sigma'$ and the Hamiltonian $H=H_{N-k}$, one obtains
\begin{equation}\label{continpartial}
\left|F\left(\rho',H\right)-F\left(\sigma',H\right)\right|\leq 32\cdot \dB\left(\rho',\sigma'\right)\left\| H_{N-k}\right\| ^{2}\leq 32\cdot \dB\left(\rho,\sigma\right)\left\| H_{N-k}\right\| ^{2}\ ,
\end{equation}
where the second inequality follows from the fact that the Bures distance does not increase under trace preserving completely positive maps \cite{Bengtsson2006} (for us the relevant TPCP map is the partial trace, $\rho\longmapsto\tr_{k}\left(\rho\right)$).
We now set
\begin{equation}\label{eq:auxiliarsym}
\rho=U\rho_{0}U^{\dagger},\qquad \sigma=\exp\left(-i\param X\right)U\rho_{0}U^{\dagger}\exp\left(i\param X\right) \ ,
\end{equation}
where $U\in\SU\left(\Hsym_N\right)$ and $X\in\mathfrak{su}\left(\Hsym_N\right)$ and the rest of the proof is exactly the same as that of \lemref{important Lipschitz}.
\end{proof}

\subsection{Lipschitz constant of the classical Fisher information}
We conclude this section by giving bounds of Lipschitz constant of the classical Fisher information for the case of isospectral mixed states, fixed Hamiltonian encoding and fixed measurements setting. Recall that for the unitary encoding \eqref{eq:trajectory} classical Fisher information is a function of the state $\rho\in\mathcal{D}\left(\H\right)$, Hamiltonian $H$, the phase $\param$, and the POVM $\left\{\Pi_n\right\}$ used in the phase estimation procedure. These three object define a family of probability distributions
\begin{equation}\label{probdist}
p_{n|\param}\left(\rho\left(\param\right)\right)=\tr\left(\Pi_n \rho\left(\param\right)\right)\ ,
\end{equation}
where $\rho\left(\param\right)=\exp\left(-\ii\param H\right) \rho \exp\left(\ii\param H\right)$. The classical Fisher information is then given by
\begin{equation}
\FI\left(\rho,H,\param,\left\{\Pi_n\right\} \right)\equiv\FI(\left\{p_{n|\param}\right\})=\sum_n
\frac{\tr\left(\ii\,[\Pi_n,H] \rho(\param)\right)^2}
{\tr\left(\Pi_n \rho(\param)\right)} \ ,
\label{eq:cFI_ful}
\end{equation}
where the summation is over the range of indices labeling the outputs of a POVM $\left\{\Pi_n\right\}$ (for simplicity we consider POVMs with finite number of outcomes). Let us fix the Hamiltonian $H$, the phase $\param$, and the POVM $\left\{\Pi_n\right\}$. Let us define a function
\begin{equation}
F_{\t{cl},\Omega,H}:\SU\left(\H\right)\ni U\longmapsto \FI\left(U\rho_0 U^\dagger,H,\param,\left\{\Pi_n\right\} \right)\in\mathbb{R} \ ,
\label{eq:classfishdef}
\end{equation}
for some fixed state $\rho_0\in\Omega$.

\begin{lem}
\label{classicalFishlip}
The Lipschitz constant (with respect to $g_{\mathrm{HS}}$) of the function $F_{\t{cl},\Omega,H}$ defined by \eqnref{eq:classfishdef}
is upper bounded by
\begin{equation}
L_{\t{cl},\Omega,H}\leq 24 \left\|H \right\|^2 \ .
\end{equation}
\end{lem}

\begin{proof}
The strategy of the proof is analogous to the one presented in the other lemmas in this section. The idea s to find a bound for
\begin{equation}
\left|\left.\frac{d}{d\phi}\right|_{\phi=0}F_{\t{cl},\Omega,H}\left(\exp\left(-\ii \phi X\right)U\right)\right| \ ,
\end{equation} in terms of the Hilbert-Schmidt norm of $X$.  Let us first assume that at $U\in\mathrm{SU}\left(\mathcal{H}\right)$ for all $n$
\begin{equation}\label{nonzeroden}
\tr\left(\Pi_n \rho(\param)\right)\neq0\ .
\end{equation}
Under the above condition we have
\begin{equation}\label{eq:uglyderr}
 \begin{split}
\left.\frac{d}{d\phi}\right|_{\phi=0}F_{\t{cl},\Omega,H}\left(\exp\left(-\ii \phi X\right)U\right) &=
\sum_n \frac{\tr\left\{\left[H,\Pi_n\right] \rho_U\left(\param\right) \right\} \tr\left\{\left[H,\Pi_n\right] \left[\ii X,\rho_U\left(\param\right)\right] \right\}}{\tr\left(\Pi_n \rho_U\left(\param\right)\right) } \\
&+ \sum_n \frac{\tr\left\{\left[H,\Pi_n\right] \rho_U\left(\param\right) \right\}^2 \tr\left\{\left[\ii H,\Pi_n\right] \rho_U\left(\param\right) \right\}}{\tr\left(\Pi_n \rho_U\left(\param\right)\right)^2} \ .
 \end{split}
\end{equation}
where $\rho_U\left(\param\right)=\exp\left(-\ii\param H\right) U \rho_0 U^\dagger \exp\left(\ii\param H\right)$. Let us introduce the auxiliary notation
\begin{align}\label{eq:split}
A&=\sum_n \left|\frac{\tr\left\{\left[H,\Pi_n\right] \rho_U\left(\param\right) \right\} \tr\left\{\left[H,\Pi_n\right] \left[\ii X,\rho_U\left(\param\right)\right] \right\}}{\tr\left(\Pi_n \rho_U\left(\param\right)\right) }\right| \ , \\
B&= \sum_n \left|\frac{\tr\left\{\left[H,\Pi_n\right] \rho_U\left(\param\right) \right\}^2 \tr\left\{\left[\ii H,\Pi_n\right] \rho_U\left(\param\right) \right\}}{\tr\left(\Pi_n \rho_U\left(\param\right)\right)^2}\right| \ .
\end{align}
Clearly, we have the inequality
\begin{equation}\label{simplesplit}
\left|\left.\frac{d}{d\phi}\right|_{\phi=0}F_{\t{cl},\Omega,H}\left(\exp\left(-\ii \phi X\right)U\right)\right|\leq A+B \ .
\end{equation}
In order to bound $A$ and $B$ (from above) we observe that for any state $\rho\in\mathcal{D}\left(\H\right)$ we have
\begin{align}
\left|\tr\left(\left[H,\Pi_n\right]\rho \right) \right|& \leq 2\tr\left(\Pi_n \rho\right) \left\|H\right\| \ ,\label{eq:posineq} \\
\left|\tr\left(\left[X,\Pi_n\right]\rho \right) \right|&\leq 2 \tr\left(\rho \Pi_n\right) \left\|X\right\| \ , \label{diffbound1}\\
\left|\tr\left(\{\left[H,\Pi_n\right] \left[\ii X,\rho\right]\right\}\right)\| &\leq 4 \tr\left(\rho \Pi_n\right) \left\|X\right\| \left\|H\right\| \ \label{diffbound} .
\end{align}
In order to prove \eqref{eq:posineq} we first upper bound $\left|\tr\left(H \Pi_n \rho \right) \right|$ ,
\begin{align}
\left|\tr\left(H \Pi_n \rho \right) \right|&=\left|\tr\left(H \sqrt{\Pi_n}\sqrt{\Pi_n} \sqrt{\rho}\sqrt{\rho} \right) \right| \\ \label{nonneg}
&\leq \sqrt{\tr\left(\rho H^2 \Pi_n\right)}\sqrt{\tr\left(\rho \Pi_n\right)} \\
&\leq \sqrt{\sqrt{\tr\left(\rho \Pi_{n}^2\right)}\sqrt{\tr\left(\rho H^4\right)}}\sqrt{\tr\left(\rho \Pi_n\right)} \label{eq:cs} \\
&\leq \tr\left(\rho \Pi_n\right) \left\|H \right\| \ . \label{last}
\end{align}
where in \eqref{nonneg} we have used the nonnegativity of operators $\Pi_n$ and $\rho$. In \eqref{eq:cs} we have repetitively used the Cauchy-Schwartz inequality, first for $P=\sqrt{\rho} H \sqrt{\Pi_n},Q=\sqrt{M}\sqrt{\rho}$ and then for $P=\sqrt{\rho}H^2,Q=M\sqrt{\rho}$. The final inequality \eqref{last} follows immediately form operator inequalities
\begin{equation}
H^4\leq \left\|H\right\|^4\1\ , \ \Pi_n\leq \Pi_{n}^2 \ .
\end{equation}
Using analogous reasoning it is possible to prove $\left|\tr\left(H \Pi_n \rho \right) \right|\leq tr\left(\rho \Pi_n\right) \left\|H \right\|$. This finishes the proof of \eqref{eq:posineq}. Using essentially the same methodology it is possible to prove the inequalities \eqref{diffbound} and \eqref{diffbound1}. By plugging  inequalities \eqref{eq:posineq}, \eqref{diffbound} and \eqref{diffbound1} into \eqnref{simplesplit} for $\rho=\rho_U \left(\param\right)$ and using the normalization condition
\begin{equation}
\sum_n \tr\left(\Pi_n \rho_U \left(\param\right)\right)=1 \ ,
\end{equation}
we obtain
\begin{align}\label{eq:laststep}
\left|\left.\frac{d}{d\phi}\right|_{\phi=0}F_{\t{cl},\Omega,H}\left(\exp\left(-\ii \phi X\right)U\right)\right| &\leq 24 \left\|H\right\|^2 \left\|X\right\| \\
&\leq 24 \left\|H\right\|^2 \left\|X\right\|_\mathrm{HS} \ .
\end{align}
By the virtue of \remref{Lipremark} we conclude that the Lipschitz constant of $F_{\t{cl},\Omega,H}\left(U\right)$ is upper bounded by $24 \left\|H\right\|$. The above derivation explicitly used the assumption \eqref{nonzeroden} which translates to assuming that denominators appearing in the definition of classical Ficher information do not vanish. However, with the help of inequalities \eqref{eq:posineq}, \eqref{diffbound} and \eqref{diffbound1}, one can easily prove that the possible singularities coming form zeros of some denominators are actually removable and that $F_{\t{cl},\Omega,H}\left(U\right)$ is actually a differentiable function of $U$. Consequently, inequality \eqref{eq:laststep} is acually satisfied for $U\in\mathrm{SU}\left(\H\right)$ for which conditions \eqref{nonzeroden} are not satisfied.
\end{proof}

\section{Lower bounds on the QFI} 
\label{sec:lowerbound}
%
\begin{lem}\label{lower bound fisch}
Let $\rho_\param$ be a one parameter family of states states on aa Hilbert space $\H$. Then, The following lower bound for QFI holds (see also \cite{Jarzyna2016})
\begin{equation}\label{eq:Mahou0}
F\left(\rho_\param,\dot{\rho}_\param\right)\geq \left\|\dot{\rho}\right\|_1^2.
\end{equation}
In particular for $\rho_\param=\exp\left(-iH\param\right)\rho\ \exp\left(iH\param\right)$ we have 
\begin{equation}\label{eq:Mahou}
F\left(\rho,H\right)\geq\|[H,\rho]\|^2_1.
\end{equation}
Recall that right hand side of \eqref{eq:Mahou}, $\|[H,\rho]\|^2_1$, equals the measure of asymmetry introduced in \cite{Marvian2014}.
\end{lem}

\begin{proof}
We give the proof only in the Hamiltonian case \eqref{eq:Mahou}. The proof of the general case is analogous. Recall that the quantum Fisher information is related to the Bures distance $\dB\left(\rho,\sigma\right)=\sqrt{2[1-\mathcal{F}\left(\rho,\sigma\right)]}$ through
\begin{equation}\label{ineq1}
\dB\left(\rho_{\param},\rho_{\param+\delta\param}\right)=\frac{1}{2}\sqrt{F\left(\rho_{\param},H\right)}|\delta\param|+O(\delta\param^{2}) .
\end{equation}
At the same time, from the
Fuchs-van der Graaf inequalities \cite{Fuchs1999} we know that
\begin{equation}\label{ineq2}
\left\| \rho-\sigma\right\| _{1} \leq 2\sqrt{1-\mathcal{F}\left(\rho,\sigma\right)^{2}} \leq 2\sqrt{2[1-\mathcal{F}\left(\rho,\sigma\right)]}=2d_B(\rho,\sigma),
\end{equation}
with the second inequality stemming from that fact that $\mathcal{F}\left(\rho,\sigma\right)\leq 1$. By combining \eqnsref{ineq1}{ineq2}
we obtain
\begin{equation}
 \left\| \rho_{\param}-\rho_{\param+\delta\param}\right\| _{1}\le\sqrt{F\left(\rho_{\param},H\right)}|\delta\param|+O(\delta\param^{2}).
\end{equation}
%
%
Diving it by $|\delta\param|$ and taking then the limit $\delta\param\to 0$, we get
\begin{equation}
\underset{\delta\param\to0}{\lim} \frac{1}{|\delta\param|}\left\| \rho_{\param}-\rho_{\param+\delta\param}\right\| _{1}\le\sqrt{F\left(\rho_{\param},H\right)},
\end{equation}
which, by virtue of the fact that
\begin{equation}
\underset{\delta\param\to0}{\lim} \frac{1}{|\delta\param|}\left\| \rho_{\param}-\rho_{\param+\delta\param}\right\| _{1}=\left\|\dot{\rho}\right\|_1\end{equation}
directly leads us to \eqnref{eq:Mahou}.
\end{proof}

\begin{rem}
Using the standard inequality between trace and Hilbert-Schmidt norms we obtain the weaker versions of inequalities \eqref{eq:Mahou0} and \eqref{eq:Mahou}
\begin{equation}\label{eq:Mahou2}
F\left(\rho_\param,\dot{\rho}_\param\right)\geq \left\|\dot{\rho}\right\|_{\mathrm{HS}}^2 ,
\end{equation}
\begin{equation}\label{eq:Mahou3}
F\left(\rho,H\right)\geq \left\|[H,\rho]\right\|_{\mathrm{HS}}^2.
\end{equation}
\end{rem}

\begin{rem}
For pure states $F\left(\rho,H\right)=\|[H,\rho]\|^2_1$.
\end{rem}

\begin{proof}
Let us prove that for pure states, the right-hand side of \eqnref{eq:Mahou} is simply the QFI of $\rho$.
To this end, let us denote $H\ket{\psi}=\ket{\widetilde{\param}}$ and then
\begin{equation}
\|[H,\rho]\|^2_1=\|\ket{\widetilde{\param}}\bra{\psi}-\ket{\psi}\bra{\widetilde{\param}}\|_1^2=
\langle\psi|H^2|\psi\rangle\|\ket{\param}\bra{\psi}-\ket{\psi}\bra{\param}\|_1^2\ ,
\end{equation}
where $\ket{\param}=\ket{\widetilde{\param}}/\sqrt{\langle\psi|H^2|\psi\rangle}$. The matrix
under the trace norm is manifestly anti-Hermitian and of rank two, and so it is straightforward to compute its norm. To do this let us write $\ket{\param}=\alpha\ket{\psi}+\beta\ket{\psi^{\perp}}$, where $|\alpha|^2+|\beta|^2=1$ and
$\ket{\psi^{\perp}}$ is some normalized vector orthogonal to $\ket{\psi}$.
It also follows that $\alpha\in\mathbbm{R}$ because $\alpha=\langle\psi|\param\rangle=\langle\psi|H|\psi\rangle/\sqrt{\langle\psi|H^2|\psi\rangle}$. All this implies that
\begin{equation}\label{Alhambra}
\ket{\param}\bra{\psi}-\ket{\psi}\bra{\param}=\beta\ket{\psi^{\perp}}\bra{\psi}-\beta^{*}\ket{\psi}\bra{\psi^{\perp}},
\end{equation}
and consequently, the eigenvalues of the above matrix are $\pm\ii|\beta|$. Thus, its trace norm amounts to $2|\beta|$, giving
\begin{align}\label{SanMiguel}
\|[H,\rho]\|^2_1&=4\langle\psi|H^2|\psi\rangle|\beta|^2=4\langle\psi|H^2|\psi\rangle(1-\alpha^2) \\
&=4\langle\psi|H^2|\psi\rangle\left(1-\frac{\langle\psi|H|\psi\rangle^2}{\langle\psi|H^2|\psi\rangle}\right)\\
&=4(\langle\psi|H^2|\psi\rangle-\langle\psi|H|\psi\rangle^2)=F\left(\rho,H\right).
\end{align}
\end{proof}

\section{Averages and bounds on averages of QFI and FI on relevant statistical ensambles}
\label{sec:averages}
In this section we compute and/or bound averages of the FI or QFI on the ensembles of mixed quantum states appearing in the main text.

\subsection{Averages of QFI on ensembles of isospectral density matrices}
We will extensively use the following result concerning the integration on the special unitary group.

\begin{fact}{(Integration of quadratic function on the unitary group)}
\label{integration}
Letting $V\in\mathrm{Herm}\left(\H\otimes\H\right)$, the following equality holds \cite{Popescu2006},

\begin{equation}
\int_{\SU\left(\H\right)}d\mu\left(U\right)\,U^{\otimes2}V\left(U^{\dagger}\right)^{\otimes2}=\alpha\mathbb{P}^{\mathrm{sym}}+\beta\mathbb{P}^{\mathrm{asym}}\ ,\label{eq:integration formula}
\end{equation}
where $\mathbb{P}^{\mathrm{sym}}$ and $\mathbb{P}^{\mathrm{asym}}$ are projectors onto the symmetric and antisymmetric subspaces of $\H\otimes\H$ and can be expressed as
\begin{equation}
\mathbb{P}^{\mathrm{sym}}=\frac{1}{2}(\1\otimes\1+\mathbb{S}),\qquad
\mathbb{P}^{\mathrm{asym}}=\frac{1}{2}(\1\otimes\1-\mathbb{S})
\end{equation}
with $\mathbb{S}$ being the co-called swap operator satisfying $\mathbb{S}\ket{x}\ket{y}=\ket{y}\ket{x}$
for any pair $\ket{x},\ket{y}\in\H$. Finally, the real coefficients $\alpha$ and $\beta$ are given by
\begin{equation}
\alpha=\frac{1}{D_+}\mathrm{tr}\left(\mathbb{P}^{\mathrm{sym}}V\right),\qquad\beta=\frac{1}{D_-}\mathrm{tr}\left(\mathbb{P}^{a\mathrm{sym}}V\right),
\end{equation}
where $D_{\pm}=D(D\pm1)/2$ with $D=|\H|$.
\end{fact}

We first consider the case in which both the Hilbert space $\H$ and the Hamiltonian $H$ are fully general.

\begin{lem}
\label{main result integration}
Let $F_{\Omega,H}$ be defined as in \eqnref{eq:auxiliary Fisher}. then, the following equality holds
\begin{equation}
\Expect_{U \sim \mu\left(\mathcal{\mathcal{H}}\right)}F_{\Omega,H}\left(U\right)=\frac{2\cdot\tr\left(H^{2}\right)}{D^{2}-1}\sum_{i,j:\,p_{i}+p_{j}\neq0}\frac{\left(p_{i}-p_{j}\right)^{2}}{p_{i}+p_{j}}\ .\label{eq:compact integration}
\end{equation}
\end{lem}
\begin{proof}
We have the following sequence of equalities
\begin{align}
\Expect_{U \sim \mu\left(\mathcal{\mathcal{H}}\right)}F_{\Omega,H}\left(U\right) & =2\cdot\sum_{i,j:\,p_{i}+p_{j}\neq0}\frac{\left(p_{i}-p_{j}\right)^{2}}{p_{i}+p_{j}}\left(\int_{\SU\left(\H\right)}\mathrm{d}\mu\left(U\right)\left|\braket{e_{i}}{U^{\dagger}HU\left|e_{j}\right.}\right|^{2}\right) \\
 & =2\cdot\sum_{i,j:\,p_{i}+p_{j}\neq0}\frac{\left(p_{i}-p_{j}\right)^{2}}{p_{i}+p_{j}}\left(\int_{\mathrm{SU}\left(\H\right)}\mathrm{d}\mu\left(U\right)\tr\left[U\otimes U\,H\otimes H\,U^{\dagger}\otimes U^{\dagger}\,\ketbra{e_{i}}{e_{j}}\otimes\ketbra{e_{j}}{e_{i}}\right]\right) \\
 & =\sum_{i,j:\,p_{i}+p_{j}\neq0}\frac{\left(p_{i}-p_{j}\right)^{2}}{p_{i}+p_{j}}\tr\left\{\left[(\alpha_H+\beta_H)\1\otimes\1+(\alpha_H-\beta_H)\mathbb{S}
 \right]\,\ketbra{e_{i}}{e_{j}}\otimes\ketbra{e_{j}}{e_{i}}\right\},\label{eq:semi final}
\end{align}
where the third equality follows from \factref{integration} for $V=H\otimes H$, and
the real numbers $\alpha_H$ and $\beta_H$ are given by
\begin{equation}
\alpha_{H}=\frac{1}{D_+}\tr\left(H\otimes H\,\mathbb{P}^{\mathrm{sym}}\right)=\frac{1}{2D_+}\tr\left(H^{2}\right),\qquad\beta_{H}=\frac{1}{D_-}\tr\left(H\otimes H\,\mathbb{P}^{a\mathrm{sym}}\right)=-\frac{1}{2D_-}\tr\left(H^{2}\right).\label{eq:constants}
\end{equation}
To obtain \eqref{eq:constants} we also used the fact that $\tr\left(H\right)=0$.
Inserting then \eqnref{eq:constants} to \eqnref{eq:semi final} and using
the identities
\begin{equation}
\mathrm{tr}\left(\ketbra{e_{i}}{e_{j}}\otimes\ketbra{e_{j}}{e_{i}}\right)=0,\qquad\mathrm{tr}\left(\ketbra{e_{i}}{e_{j}}\otimes\ketbra{e_{j}}{e_{i}}\mathbb{S}\right)=1,
\end{equation}
one arrives at
\begin{equation}
\Expect_{U \sim \mu\left(\mathcal{\mathcal{H}}\right)}F_{\Omega,H}\left(U\right)=\frac{\tr\left(H^{2}\right)}{2}\left(\frac{1}{D_{+}}+\frac{1}{D_{-}}\right)\sum_{i,j: p_{i}+p_{j}\neq0}\frac{\left(p_{i}-p_{j}\right)^{2}}{p_{i}+p_{j}},
\end{equation}
which, by virtue of the definitions of $D_{\pm}$, leads us to \eqnref{eq:compact integration}.
\end{proof}
The formula \eqref{eq:compact integration} simplifies significantly for the case of pure states.
\begin{rem}
Let $\Omega_0$ consist of pure states on $\H$. In this case
we it fairly easy to see that
\begin{equation}
\sum_{i,j: p_{i}+p_{j}\neq0}\frac{\left(p_{i}-p_{j}\right)^{2}}{p_{i}+p_{j}}=2(D-1),
\end{equation}
and consequently,
\begin{equation}
\Expect_{U \sim \mu\left(\mathcal{\mathcal{H}}\right)}F_{\Omega_0,H}\left(U\right)=\frac{4\cdot\mathrm{tr}\left(H^{2}\right)}{D+1}\,.\label{eq: Pure states average}
\end{equation}

\end{rem}
By comparing \eqnsref{eq: Pure states average}{eq:compact integration} one finds that
the average QFI over any ensemble $\Omega$ of isospectral states can be easily related to
the average QFI over pure states. Specifically, one has
\begin{equation}
\Expect_{U \sim \mu\left(\mathcal{\mathcal{H}}\right)}F_{\Omega,H}\left(U\right)=
\frac{4\cdot\mathrm{tr}\left(H^{2}\right)}{D+1} \cdot \Lambda(\{p_j\}_j)=
\Expect_{U \sim \mu\left(\mathcal{\mathcal{H}}\right)}F_{\Omega_0,H}\left(U\right)\Lambda_{\Omega},
\end{equation}
where
\begin{equation}\label{lambda}
\Lambda(\{p_j\}_j)=\frac{1}{2\left(D-1\right)}\sum_{i,j: p_{i}+p_{j}\neq0}\frac{\left(p_{i}-p_{j}\right)^{2}}{p_{i}+p_{j}}.
\end{equation}
Note that  $\Lambda(\{p_j\}_j)=1$ for pure states and $\Lambda(\{p_j\}_j)=0$ for the maximally mixed state. Since in general the dependance on the spectrum in the above formula is quite complicated it is desirable to have simple bounds on $\sum_{i,j: p_{i}+p_{j}\neq0}\frac{\left(p_{i}-p_{j}\right)^{2}}{p_{i}+p_{j}}$. The following fact provides one such bound:
\begin{fact}
\label{christian bound}
Let the numbers $p_{1},\ldots,p_{D}$ satisfy
$p_{i}\geq0$ and $\sum_{i=1}^{D}p_{i}=1$.
Then, the following inequality holds
\begin{equation}
\sum_{i,j: p_{i}+p_{j}\neq0}\frac{\left(p_{i}-p_{j}\right)^{2}}{p_{i}+p_{j}}\geq2\left[D-\left(\sum_{i=1}^{D}\sqrt{p_{i}}\right)^{2}\right]
=2D\left[1-\mathcal{F}^2(\rho,\tfrac{\mathbbm{1}}{D})\right].
\label{eq:important bound}
\end{equation}
\end{fact}

\begin{proof}
First, by using the identity $(p_i-p_j)^2=(p_i+p_j)^2-4p_ip_j$ the left-hand side of the inequality \ref{eq:important bound} can be rewritten as
\begin{equation}
\sum_{i,j: p_{i}+p_{j}\neq0}\frac{\left(p_{i}-p_{j}\right)^{2}}{p_{i}+p_{j}} =\sum_{i,j: p_{i}+p_{j}\neq0}\left[\frac{\left(p_{i}+p_{j}\right)^{2}}{p_{i}+p_{j}}-\frac{4p_{i}p_{j}}{p_{i}+p_{j}}\right],
\end{equation}
which, noting that the first sum in the above amounts to $2(D-1)$, can be rewritten as
\begin{equation}\label{eq:second auxiliary}
\sum_{i,j: p_{i}+p_{j}\neq0}\frac{\left(p_{i}-p_{j}\right)^{2}}{p_{i}+p_{j}}
=2\left(\left(D-1\right)-\sum_{i,j: p_{i}+p_{j}\neq0}\frac{2p_{i}p_{j}}{p_{i}+p_{j}}+1\right).
\end{equation}
To obtain the inequality in \eqnref{eq:important bound} we apply
the following well-known relation between the harmonic and geometric means
\begin{equation}
\frac{2}{\frac{1}{p_{i}}+\frac{1}{p_{j}}}\leq\sqrt{p_{i}p_{j}}
\end{equation}
to \eqnref{eq:second auxiliary}. Then, to obtain the equality in
\eqref{eq:important bound} and complete the proof it suffices to
notice that
\begin{equation}
\mathcal{F}(\rho,\tfrac{\mathbbm{1}}{D})=\tr\sqrt{\sqrt{\tfrac{\mathbbm{1}}{D}}\rho\sqrt{\tfrac{\mathbbm{1}}{D}}}=\frac{1}{\sqrt{D}}\tr\sqrt{\rho}=\frac{1}{\sqrt{D}}
\sum_{i=1}^{D}\sqrt{p_i}.
\end{equation}
\end{proof}
\begin{rem}
Note that the bound \eqref{eq:important bound} is tight. To be more
precise, it is saturated for the maximally mixed state $\rho=\mathbb{I}/D$, for which both sides
of the inequality \eqref{eq:important bound} simply vanish, and for pure states for which
they amount to $2(D-1)$.
\end{rem}

%
%

\subsection{Averages of QFI for $N$ particles}
We now discuss the average behaviour of the QFI for ensembles consisting of states of distinguishable or bosonic particles (in the case when all particles evolve in the same manner under local Hamiltonian). For the case of $N$ distinguishable particles we have
\begin{equation}
\mathcal{H}=\mathcal{H}_N=\left(\mathbb{C}^{d}\right)^{\otimes N} \label{hilbertdist} \ ,
\end{equation}
where $\mathbb{C}^{d}$ is a Hilbert space of a single particle and
$N$ is a number of particles. Clearly, we have $D=|\mathcal{H}_N|=d^{N}$. The Hilbert space of $N$ bosons in $d$ modes is the completely symmetric subspace of $\mathcal{H}_N$,
\begin{equation}
\mathcal{H}=\Hsym_N=\mathrm{span}_\mathbb{C}\left\{\ket{\phi}^{\otimes N} | \ket{\phi}\in\mathbb{C}^d \right\} \label{hilbertsym} \ ,
\end{equation}
of dimension $D= |\Hsym_N|=\binom{N+d-1}{N}$. It will be convenient for us to use the orthonormal basis of $\Hsym_N$ consisting of generalized Dicke states \cite{Migdal2014} (we wil use them extensively also in the part of the Appendix, where we estimate the impact of particle losses on typical properties of QFI). Within the second quantization picture $\Hsym_N$ can be treated as a subspace of $d$ mode bosonic Fock space and the generalized Dicke states are of the form
\begin{equation}
\ket{\vec{k},N}=\frac{\prod_{i=1}^{d} \left(a_{i}^\dagger\right)^{k_i}}{\sqrt{\prod_{i=1}^{d} k_i! }} \ket{\Omega},
\end{equation}
where $\ket{\Omega}$ is the Fock vacuum, $a_{i}^\dagger$ are the standard creation operators and the vector $\vec{k}=\left(k_1,k_2,\ldots,k_d\right)$ consists of non-negative integers counting how many particles occupy each mode. Due to the fact that the number of particles is $N$, the vector $\vec{k}$ satisfies the normalization condition $|\vec{k}|\coloneqq\sum_{i=1}^{d}k_i=N$. Let us also notice that in the particle picture the Dicke states are given by
\begin{equation}\label{eq:particledicke}
\ket{\vec{k},N}=\mathcal{N}(\vec{k},N)\mathbb{P}_{\mathrm{sym}}^N\ket{\vec{k}},
\end{equation}
where $\ket{\vec{k}}$ is a vector from $(\mathbb{C}^d)^{\otimes N}$ given by
$\ket{\vec{k}}=\ket{1}^{\otimes k_1}\otimes \ket{2}^{\otimes k_2} \otimes \ldots \otimes \ket{d}^{\otimes k_d}$,
%
the constant $\mathcal{N}(\vec{k},N)$ is given by
\begin{equation}
\mathcal{N}(\vec{k},N)=\sqrt{\binom{N}{\vec{k}}}
\end{equation}
with
\begin{equation}
\binom{N}{\vec{k}}=\frac{N!}{\prod_{i=1}^{d} k_i!},
\end{equation}
and by $\mathbb{P}_{\mathrm{sym}}^N$ we denote
the orthonormal projector onto $\Hsym_N \subset \mathcal{H}_N$.

The Hamiltonian used in the phase estimation is local and symmetric under exchange of particles,
\begin{equation}
H=H_{N}=h\otimes\mathbb{I}\otimes\ldots\otimes\mathbb{I}+\mathbb{I}\otimes h\otimes\mathbb{I}\otimes\ldots\otimes\mathbb{I}+\ldots+\mathbb{I}\otimes\ldots\otimes\mathbb{I}\otimes h\,,\label{eq:local hamiltonian}
\end{equation}
where $h$ stands for the single-particle local Hamiltonian. In what
follows we assume for simplicity that $\mathrm{tr}\left(h\right)=0$. Note that the Hamiltonian $H_{N}$ preserves the subspace $\Hsym_N$.

Now, it follows from \eqnsref{eq:compact integration}{eq: Pure states average}
that the average behaviour of the QFI on states supported on the subspace $\mathcal{W}\subset\H$
is dictated by the value of $\mathrm{tr}_{\mathcal{W}}(H_{N}^{2})$. In the following lemma
we compute the latter in the cases $\mathcal{W}=\mathcal{H}_N$ and $\mathcal{W}=\Hsym_N$.
\begin{lem}\label{averagesrem}
Let Hamiltonian $H$ be given by \eqref{eq:local hamiltonian}. Then, the following relations
\begin{equation}\label{eq:distingav}
\Expect_{U \sim \mu\left(\mathcal{H}_N\right)}F_{\Omega,H}=\frac{4 N\cdot \mathrm{tr}\left(h^2\right)}{d}\cdot\frac{|\mathcal{H}_N|}{|\mathcal{H}_N|+1}\cdot\Lambda(\{p_j\}_j)
\end{equation}
and
\begin{equation}\label{eq:symav}
\Expect_{U \sim \mu\left(\Hsym_N\right)}F_{\Omega,H}= \frac{4N\left(N+d\right)\mathrm{tr}\left(h^2\right)}{d\left(d+1\right)}
\cdot \frac{|\Hsym_N|}{|\Hsym_N|+1} \cdot\Lambda(\{p_j\}_j).
\end{equation}
are true. For pure qubits \eqref{eq:symav} simplifies to
\begin{equation}
\Expect_{U \sim \mu\left(\Hsym_N\right)}F_{(1,0,\ldots,0),H}=\frac{2}{3}N(N+1)\tr\left(h^2\right)\ .
\end{equation}

\end{lem}

\begin{rem}
The qualitative meaning of the above lemma is twofold. First, it shows that for uniformly distributed isospectral states from $\mathcal{H}_N$ the scaling of the QFI on average is at most linear in the number of particles $N$, and,
%
%
secondly, it proves that for random pure symmetric states
the average QFI scales quadratically with $N$, both for fixed local Hamiltonian $h$ and local dimension $d$. Thus, for symmetric states the average QFI attains the Heisenberg limit. This behaviour still holds for random isospectral density matrices, provided their spectrum is sufficiently pure with the ``degree of purity'' quantified by $\Lambda(\{p_j\}_j)$ defined by \eqref{lambda}.
\end{rem}

\begin{proof}
We start from the proof of \eqref{eq:distingav}. Using the fact that the local Hamiltonian $h$ is traceless, one obtains
\begin{equation}
\mathrm{tr}_{\mathcal{H}_N}\left(H^2\right)=\sum_{i=1}^{N}\mathrm{tr}_{\mathcal{H}_N}\left[\left(h^{(i)}\right)^2 \right]=\sum_{i=1}^{N}\mathrm{tr}\left(h^2 \right)d^{N-1}=N\cdot\mathrm{tr}\left(h^2 \right) \frac{|\mathcal{H}_N|}{d}.
\end{equation}
Inserting the above to \eqref{eq:compact integration} (note that here $\mathcal{H}=\mathcal{H}_N$), we arrive at \eqref{eq:distingav}.

The proof of \eqref{eq:symav} is more involving as it requires the computation of $\mathrm{tr}_{\Hsym_N}\left(H^2\right)$. The final result reads
\begin{equation}\label{symmetric trace}
\mathrm{tr}_{\Hsym_N}\left(H^2\right)=\frac{N\left(N+d\right)\mathrm{tr}\left(h^2\right)}{d\left(d+1\right)}\cdot|\Hsym_N|
\end{equation}
which when plugged into \eqref{eq:compact integration} yields \eqref{eq:symav}.

To determine explicitly \eqref{symmetric trace} let us chose the basis $\left\{\ket{i}\right\}_{i=1}^{d}$ of the single particle space as the eigenbasis of the local Hamiltonian $h$. Thus we have $h\ket{i}=\lambda_i\ket{i}$ for $i=1,\ldots,d$. The corresponding generalized Dicke states (see \eqref{eq:particledicke}) satisfy
\begin{equation}\label{gen Dicke identity}
H_{N}\ket{\vec{k},N}=\left(\vec{k}\cdot\vec{\lambda}\right) \ket{\vec{k},N},
\end{equation}
where $\vec{\lambda}=\left(\lambda_1,\ldots,\lambda_d\right)$ is the vector of eigenvalues of $h$ and $\cdot$ is the standard inner product in $\mathbb{R}^d$. Now, Eq.~\eqref{gen Dicke identity} together with the fact that the generalized Dicke states form a basis of $\Hsym_N$ allow us to write
\begin{align}\label{eq:expansion1}
\mathrm{tr}_{\Hsym_N}\left(H^2\right)&=\sum_{\vec{k}:|\vec{k}|=N}\left(\vec{k}\cdot \vec{\lambda}\right)^2 \\
&=\sum_{\vec{k}:|\vec{k}|=N}\sum_{i=1}^d\lambda_{i}^2 k_{i}^2+\sum_{\vec{k}:|\vec{k}|=N}\sum_{\substack{i,j=1\\i\neq j}}^d \left(\lambda_i \lambda _j\right)\left(k_i k_j \right),
\end{align}
where to obtain the second equality we explicitly squared all scalar products appearing under the sum.
From the symmetry we have
\begin{equation}\label{eq:expansion2}
\sum_{\vec{k}:|\vec{k}|=N} k_{i}^2 = \sum_{\vec{k}:|\vec{k}|=N} k_{i'}^2, \qquad \sum_{\vec{k}:|\vec{k}|=N}k_i k_j = \sum_{\vec{k}:|\vec{k}|=N}k_{i'} k_{j'},
\end{equation}
for all $i,i'$ and for all pairs of different indices $\left(i,j\right)$ and $\left(i',j'\right)$. As a result Eq.~\eqref{eq:expansion1} simplifies to
\begin{equation}\label{eq:expansion1s}
\mathrm{tr}_{\Hsym_N}\left(H^2\right)=\sum_{\vec{k}:|\vec{k}|=N}\left(\left[\sum_{i=1}^d\lambda_{i}^2\right]k_1 + \left[\sum_{\substack{i,j=1\\i\neq j}}^d \lambda_i \lambda _j\right]k_1 k_2  \right).
\end{equation}
The fact that $h$ is traceless yields
\begin{equation}\label{eq:expansion3}
\mathrm{tr}\left(h^2\right)=\sum_{i=1}^d \lambda_{i}^2=-\sum_{\substack{i,j=1\\i\neq j}}^d \lambda_i \lambda_j.
\end{equation}
Moreover, due to the condition $k_1+\ldots+k_d=N$ we have
\begin{equation}
\sum_{\vec{k}:|\vec{k}|=N} \left(k_1+\ldots+ k_d \right)^2 = |\Hsym_N|N^2.
\end{equation}
By exploiting the identities \eqref{eq:expansion2} the left-hand side of the above equation can be rewritten as
\begin{equation}\label{eq:expansion4}
\sum_{\vec{k}:|\vec{k}|=N} \left(k_1+\ldots+ k_d \right)^2 = d\cdot \sum_{\vec{k}:|\vec{k}|=N} k_{1}^2 + d\left(d-1\right)\cdot\sum_{\vec{k}:|\vec{k}|=N} k_{1} k_2.
\end{equation}
As a result, one obtains
\begin{equation}\label{eq:expansion5}
d\left(d-1\right)\cdot\sum_{\vec{k}:|\vec{k}|=N}k_1 k_2 = |\Hsym_N|N^2 - d\cdot \sum_{\vec{k}:|\vec{k}|=N} k_{1}^2
\end{equation}
Using Eqs. \eqref{eq:expansion1s}, \eqref{eq:expansion3} and \eqref{eq:expansion5} we finally arrive at
\begin{equation}\label{almpostfinal}
\mathrm{tr}_{\Hsym_N}\left(H^2\right)=\mathrm{tr}\left(h^2\right)\cdot\left[\left(d+1\right)\left(\sum_{\vec{k}:|\vec{k}|=N} k_{1}^2\right) - N^2 \cdot |\Hsym_N| \right].
\end{equation}
We compute the sum $\sum_{\vec{k}:|\vec{k}|=N} k_{1}^2$ by noting that
\begin{equation}\label{numeration}
\#\left(\left\{\vec{k} \left|\ |\vec{k}|=N ,k_1 =i \right\}\right.\right)= \binom{N-i+d-2}{N-i},
\end{equation}
where $\#\left(\cdot\right)$ denotes the number of elements of a discrete set. The above equation follows from the fact that the number of elements of the set $\{\vec{k} |\ |\vec{k}|=N ,k_1 =i \}$ is the same as the dimension of the Hilbert space of $N-i$ bosons in $d-1$ modes. Consequently, we get
\begin{equation}\label{finalugly}
\sum_{\vec{k}:|\vec{k}|=N} k_{1}^2 = \sum_{i=0}^{N} i^2\cdot \binom{N-i+d-2}{N-i}=\frac{N\left(2N+d-1\right)}{d\left(d+1\right)}|\Hsym_N|.
\end{equation}
Inserting the above expression to \eqref{almpostfinal} yields \eqref{symmetric trace}. The equlity \eqref{finalugly} can be proven using standard combinatorial identities. Below we sketch its proof for completeness. First, by the virtue of the diagonal sum property of binomial coefficients \cite{Bronstejn2013} we have that
\begin{equation}\label{diagsum}
\binom{N-i+a}{N-i}= \sum_{k=0}^{N-i} \binom{a+k-1}{k},
\end{equation}
where $a$ an arbitrary integer. Inserting \eqref{diagsum} (with $a=d=2$) to the left hand side of \eqref{finalugly} we get
\begin{equation}
\sum_{i=0}^{N} i^2\cdot \binom{N-i+d-2}{N-i}=\sum_{i=0}^{N} \sum_{k=0}^{N-i} i^2\cdot \binom{d-2+k-1}{k}=\sum_{k=0}^{N}\left(\sum_{i=0}^{N-k} i^2\right)\cdot \binom{d-2+k-1}{k}.
\end{equation}
The sum $\sum_{i=0}^{N-k} i^2$ is a polynomial of degree $3$ in k and can be easily computed. Therefore, in order to finish the computation it suffices to know the moments
\begin{equation}
\sum_{k=0}^{N} k^j \cdot \binom{a-1+k}{k},
\end{equation}
for the powers $j=1,2,3$. These can be found for instance on page 5 of \cite{Gould2010}.
\end{proof}
\begin{rem}
The most demanding part in the proof of \lemref{averagesrem} was the computation of $\mathrm{tr}_{\Hsym_N}\!\left(H^2\right)$ that can be simplified greatly by the use of group theoretic methods. This should allow one to perform analogous analysis for other irreducible representations of the group $\mathrm{SU}\left(d\right)$, for instance,
for the fermionic subspace of $\mathcal{H}_N$%
.
\end{rem}

\subsection{Average QFI for bosons with particle losses}
\label{sec:average_qfi_for_bosons_with_particle_losses}
Let $\rho'=\mathrm{tr}_k \left(\rho\right)$ be a mixed symmetric state on $N-k$ particles arising from tracing out $k$ particles of some $N$-partite state $\rho\in\mathcal{D}\left(\Hsym_N\right)$. Our aim in this section is to bound the average of the QFI over mixed states created in the above way, where $\rho$ is a random isospectral state acting on $\Hsym_N$. Recall that that we are interested in the standard context of quantum metrology, i.e., the Hamiltonian $H$ encoding the phase $\param$ is given by \eqnref{eq:local hamiltonian}.


\begin{lem}\label{boundavloss}
Let $\rho\in\mathcal{D}\left(\mathcal{S_N}\right)$ be a state of $N$ bosons with single particle $d$-dimensional Hilbert space $\mathcal{H}_l$ and the spectrum $\left\{p_1,\ldots,p_{|\Hsym_N|}\right\}$. Let us fix the local Hamiltonian $h$ and a non-negative integer $k$. Then, the following inequality holds
\begin{equation}
\Expect_{U \sim \mu\left(\mathcal{\Hsym_N}\right)} F\left(\mathrm{tr}_k\left(U\rho U^\dagger \right),H^{[N-k]}\right)\geq 2\frac{(N-k)(N+d)}{(d+1)(d+k)}\cdot\frac{|\Hsym_N|(|\Hsym_N|\tr \rho_N^2-1)\tr h^2}{|\Hsym_k|(|\Hsym_N|^2-1)}.
%
\end{equation}
\end{lem}
\begin{proof}
Denoting $\sigma_{U}=\tr_k(U\rho U^{\dagger})$ we notice that the inequality \eqref{eq:Mahou3}  allows one to lower-bound the QFI as
\begin{align}\label{Remik1}
F(\sigma_U,H_{N-k})&\geq \left\|\left[\sigma_U,H_{N-k}\right]\right\|_{\mathrm{HS}}^2 \\
&=2\left(\tr_{\Hsym_{N-k}}\!\left\{\sigma_U^2H_{N-k}^2\right\}-\tr_{\Hsym_{N-k}}\!\left\{(\sigma_UH_{N-k})^2\right\}\right),
\end{align}
where due to the fact that $\sigma_U$ is symmetric, the trace is taken over
the symmetric subspace $\Hsym_{N-k}$. For the same reason we can cut the Hamiltonian to
the symmetric subspace on which it acts as
%
%
\begin{equation}\label{Remik2}
H_{N-k}\Big\vert_{\Hsym_{N-k}}\coloneqq\mathbb{P}_{\mathrm{sym}}^{N-k} H_{N-k} \mathbb{P}_{\mathrm{sym}}^{N-k}=\sum_{\vec{n}}\lambda_{\vec{n}}^{(N-k)}\proj{\vec{n},N-k},
\end{equation}
where, as before, $\ket{\vec{n},N-k}$ are $(N-k)$-partite generalized Dicke states and
$\vec{n}=(n_0,\ldots,n_{d-1})$ is a vector of non-negative integers such that $n_0+\ldots+n_{d-1}=N-k$, and,
$\lambda_{\vec{n}}^{(N-k)}$ are the eigenvalues of $H_{N-k}$.
By abuse of notation, in what follows we denote both the Hamiltonian and its symmetric part
\eqref{Remik2} by $H_{N-k}$.

Using the swap operator introduced in \factref{integration} for $\H=\Hsym_{N-k}$
and the fact that $\tr(\mathbb{S}A\ot B)=\tr(AB)$ holds for any pair of operators acting on $\Hsym_{N-k}$, we can rewrite \eqnref{Remik1} as
\begin{align}\label{JanekBezFiranek2}
F(\sigma_U,H_{N-k})&\geq 2\,\left\{\tr[(H_{N-k}\sigma_U\ot\sigma_U H_{N-k})\mathbb{S}_{N-k}]-\tr\left[(\sigma_U H_{N-k}\ot\sigma_U H_{N-k}\right)\mathbb{S}_{N-k}]\right\} \\
&= 2\,\left\{\tr\left[(\sigma_U\ot\sigma_U)\left(\mathbb{P}_{\mathrm{sym}}^{N-k}\ot H_{N-k}^2-H_{N-k}\ot H_{N-k}\right)\mathbb{S}_{N-k}\right]\right\},
\end{align}
where to obtain the second line we used the fact that $\sigma_U$ acts on $\Hsym_{N-k}$ and that $\mathbb{S}^2_{N-k}=\mathbb{P}_{\mathrm{sym}}^{N-k}\ot \mathbb{P}_{\mathrm{sym}}^{N-k}$, and, for simplicity, we dropped the subscript $\Hsym_{N-k}\ot\Hsym_{N-k}$ in the trace.

Exploting the fact that the symmetric projector $\mathbb{P}_{\mathrm{sym}}^{N}$
is diagonal in the Dicke basis, that is,
\begin{equation}\label{Leeds}
\mathbb{P}_{\mathrm{sym}}^N=\sum_{\vec{p}}\proj{\vec{p},N},
\end{equation}
the representation of the Hamiltonian in \eqnref{Remik1} and the definition of the swap operator,
one arrives at the following formula
\begin{equation}
(\mathbb{P}_{\mathrm{sym}}^{N-k}\ot H^2_{N-k}-H_{N-k}\ot H_{N-k})\mathbb{S}_{N-k}=\sum_{\vec{n},\vec{m}}(\lambda_{\vec{n}}^2-\lambda_{\vec{n}}\lambda_{\vec{m}})
\ket{\vec{m},N-k}\bra{\vec{n},N-k}\otimes \ket{\vec{n},N-k}\bra{\vec{m},N-k},
\end{equation}
which when plugged into \eqnref{JanekBezFiranek2} gives
\begin{equation}\label{Remik3}
  F(\sigma_U,H_{N-k})\geq 2
  \sum_{\vec{n},\vec{m}}(\lambda_{\vec{n}}^2-\lambda_{\vec{n}}\lambda_{\vec{m}})
  \tr[(\sigma_U\ot\sigma_U)
  \ket{\vec{m},N-k}\bra{\vec{n},N-k}\otimes \ket{\vec{n},N-k}\bra{\vec{m},N-k}].
\end{equation}

We are now ready to lower bound the average $\Expect_{U \sim \mu\left(\mathcal{\Hsym_N}\right)} F\left(\sigma_U,H_{N-k}\right)$. Using the fact that $\sigma_U=\tr_k(U\rho U^{\dagger})$ and that $U\rho U^{\dagger}$ is symmetric, we obtain from inequality \eqref{Remik3} that
\begin{align}\label{Av1}
  &\Expect_{U \sim \mu\left(\Hsym_N\right)} F\left(\sigma_U,H_{N-k}\right)
  \geq 2
  \sum_{\vec{n},\vec{m}}(\lambda_{\vec{n}}^2-\lambda_{\vec{n}}\lambda_{\vec{m}})\\
  &\times\int_{\mathrm{SU}(\Hsym_N)}\mathrm{d}\mu(U)\,\tr[(U\rho U^{\dag}\ot U\rho U^{\dag})
  \ket{\vec{m},N-k}\bra{\vec{n},N-k}\ot\mathbb{P}_{\mathrm{sym}}^k\otimes \ket{\vec{n},N-k}\bra{\vec{m},N-k}\ot \mathbb{P}_{\mathrm{sym}}^k], \nonumber
\end{align}
where now the trace is performed over $\Hsym_N\ot\Hsym_N$.
Let us focus for a moment on the state
\begin{equation}
\int_{\mathrm{SU}(\Hsym_N)}\mathrm{d}\mu(U)\,(U\rho U^{\dag}\ot U\rho U^{\dag}).
\end{equation}
It follows from \factref{integration} (for $\mathcal{H}=\Hsym_N$)
that after performing the integration the above state assumes the following form
\begin{equation}\label{InvState}
\int_{\mathrm{SU}(\Hsym_N)}\mathrm{d}\mu(U)\,(U\rho U^{\dag}\ot U\rho U^{\dag})=\alpha\mathbb{P}_{\mathrm{sym}\wedge\mathrm{sym}}+\beta\mathbb{P}_{\mathrm{as}\wedge\mathrm{as}}.
\end{equation}
For completeness let us recall that $\mathbb{P}_{\mathrm{sym}\wedge\mathrm{sym}}$ and
$\mathbb{P}_{\mathrm{as}\wedge\mathrm{as}}$ are the projectors onto the
symmetric and antisymmetric subspaces of $\Hsym_N\ot\Hsym_N$, respectively, and are given by
\begin{equation}\label{SymAsym}
\mathbb{P}_{\mathrm{sym}\wedge\mathrm{sym}}=\frac{1}{2}\left(\mathbb{P}_{\mathrm{sym}}^N\ot \mathbb{P}_{\mathrm{sym}}^N+\mathbb{S}_N\right),\qquad
\mathbb{P}_{\mathrm{as}\wedge\mathrm{as}}=\frac{1}{2}\left(\mathbb{P}_{\mathrm{sym}}^N\ot \mathbb{P}_{\mathrm{sym}}^N-\mathbb{S}_N\right).
\end{equation}
Moreover, the real coefficients $\alpha$ and $\beta$ are explicitly given by
\begin{equation}
\alpha=\frac{1}{2D_+(\Hsym_N)}(1+\tr\rho^2),\qquad
\beta=\frac{1}{2D_-(\Hsym_N)}(1-\tr\rho^2),\qquad
\end{equation}
where $D_{\pm}(\Hsym_N)=|\Hsym_N|(|\Hsym_N|\pm1)/2$.

Plugging \eqnref{InvState} into \eqnref{Av1} and using \eqnref{SymAsym}
one arrives at
\begin{align}\label{Av2}
&\Expect_{U \sim \mu\left(\Hsym_N\right)} F\left(\sigma_U,H_{N-k}\right)
\geq
\sum_{\vec{n},\vec{m}}(\lambda_{\vec{n}}^2-\lambda_{\vec{n}}\lambda_{\vec{m}})
\left\{(\alpha+\beta)\left|\tr\left(\mathbb{P}_{\mathrm{sym}}^N\ket{\vec{m},N-k}\bra{\vec{n},N-k}\ot\mathbb{P}_{\mathrm{sym}}^k\right)\right|^2\right.\\
&\left.+(\alpha-\beta)\tr\left[\mathbb{S}_N
\left(\ket{\vec{m},N-k}\bra{\vec{n},N-k}\ot\mathbb{P}_{\mathrm{sym}}^k\otimes \ket{\vec{n},N-k}\bra{\vec{m},N-k}\ot \mathbb{P}_{\mathrm{sym}}^k\right)\right]\right\}. \nonumber
\end{align}
The right-hand side of this inequality can significantly be simplified if one observes that
the first trace under the curly brackets is nonzero only if $\vec{m}=\vec{n}$, giving
\begin{align}\label{Av3}
\Expect_{U \sim \mu\left(\Hsym_N\right)} F\left(\sigma_U,H_{N-k}\right)
&\geq (\alpha-\beta)
\sum_{\vec{n},\vec{m}}(\lambda_{\vec{n}}^2-\lambda_{\vec{n}}\lambda_{\vec{m}}) \\
&\left.\times\tr\left[\mathbb{S}_N
\left(\ket{\vec{m},N-k}\bra{\vec{n},N-k}\ot\mathbb{P}_{\mathrm{sym}}^k\otimes \ket{\vec{n},N-k}\bra{\vec{m},N-k}\ot \mathbb{P}_{\mathrm{sym}}^k\right)\right]\right\}. \nonumber
\end{align}
Our aim now is to compute the remaining trace, which for further purposes we
denote $T_{\vec{m},\vec{n}}$. We use the fact that the projector
$\mathbb{P}_{\mathrm{sym}}^k$ can be written as in \eqnref{Leeds},
which together with the following identity
\begin{equation}\label{dupa1}
\mathbb{P}_{\mathrm{sym}}^N\ket{\vec{n},N-k}\ket{\vec{p},k}=\frac{\sqrt{{N-k}\choose{\vec{n}}}\sqrt{{k}\choose{\vec{p}}}}{\sqrt{{N}\choose{\vec{n}+\vec{p}}}}\ket{\vec{n}+\vec{p},N},
\end{equation}
allows us to express $T_{\vec{m},\vec{n}}$ as
\begin{equation}
T_{\vec{m},\vec{n}}=\sum_{\vec{o}}\frac{{{N-k}\choose{\vec{m}}}{{k}\choose{\vec{o}}}}{{{N}\choose{\vec{m}+\vec{o}}}}
\frac{{{N-k}\choose{\vec{n}}}{{k}\choose{\vec{o}}}}{{{N}\choose{\vec{n}+\vec{o}}}}.
\end{equation}
This gives
\begin{align}\label{Av4}
\Expect_{U \sim \mu\left(\Hsym_N\right)} F\left(\sigma_U,H_{N-k}\right)
&\geq (\alpha-\beta)
\sum_{\vec{n},\vec{m}}(\lambda_{\vec{n}}^2-\lambda_{\vec{n}}\lambda_{\vec{m}})\sum_{\vec{o}}\frac{{{N-k}\choose{\vec{m}}}{{k}\choose{\vec{o}}}}{{{N}\choose{\vec{m}+\vec{o}}}}
\frac{{{N-k}\choose{\vec{n}}}{{k}\choose{\vec{o}}}}{{{N}\choose{\vec{n}+\vec{o}}}}\\
&=2\frac{|\Hsym_N|\tr\rho^2-1}{|\Hsym_N|(|\Hsym_N|^2-1)}(L_{N,k}-L'_{N,k}),
\end{align}
where we used the explicit expressions for $\alpha$ and $\beta$ and denoted
\begin{equation}
L_{N,k}=\sum_{\vec{o}}\left[\sum_{\vec{m}}\frac{{{N-k}\choose{\vec{m}}}{{k}\choose{\vec{o}}}}{{{N}\choose{\vec{m}+\vec{o}}}}\right]\left[\sum_{\vec{n}}\lambda_{\vec{n}}^2\frac{{{N-k}\choose{\vec{n}}}{{k}\choose{\vec{o}}}}{{{N}\choose{\vec{n}+\vec{o}}}}\right],\quad
L_{N,k}'=\sum_{\vec{o}}\left[\sum_{\vec{m}}\lambda_{\vec{m}}\frac{{{N-k}\choose{\vec{m}}}{{k}\choose{\vec{o}}}}{{{N}\choose{\vec{m}+\vec{o}}}}\right]^2.
\end{equation}
We now compute each sum separately.
To this end, let us first notice that it follows from
\eqnref{dupa1} that
\begin{align}\label{Lipton}
\sum_{\vec{m}}\frac{{{N-k}\choose{\vec{m}}}{{k}\choose{\vec{o}}}}{{{N}\choose{\vec{m}+\vec{o}}}}&=\sum_{\vec{m}}\tr\left[\mathbb{P}_{\mathrm{sym}}^N\proj{\vec{m},N-k}\ot\proj{\vec{o},k}\right]\\
&=\tr\left[\mathbb{P}_{\mathrm{sym}}^N\left(\mathbb{P}_{\mathrm{sym}}^{N-k}\ot \proj{\vec{o},k} \right)\right]\\
&=\frac{|\Hsym_N|}{|\Hsym_k|},
\end{align}
where to get the second equality we used \eqnref{Leeds}, while to obtain the third one we used
the fact that the partial trace of $\mathbb{P}_{\mathrm{sym}}^N$ over $N-k$ subsystems is given by
\begin{equation}\label{Friday}
\tr_{N-k}\left(\mathbb{P}_{\mathrm{sym}}^N\right)=\frac{|\Hsym_N|}{|\Hsym_k|}\mathbb{P}_{\mathrm{sym}}^k.
\end{equation}
With the aid of formula \eqnref{Lipton} we can write $L_{N,k}$
as
\begin{equation}
L_{N,k}=\frac{|\Hsym_N|}{|\Hsym_k|}\sum_{\vec{n}}\lambda^2_{\vec{n}}\sum_{\vec{o}}
\frac{{{N-k}\choose{\vec{n}}}{{k}\choose{\vec{o}}}}{{{N}\choose{\vec{n}+\vec{o}}}}.
\end{equation}
Then, exploiting formulas \eqnref{Leeds} and \eqnref{dupa1} and the form of the Hamiltonian $H_{N-k}$
this further rewrites as
\begin{align}\label{el}
L_{N,k}&=\frac{|\Hsym_N|}{|\Hsym_k|}\tr\left[\mathbb{P}_{\mathrm{sym}}^N\left(H^2_{N-k}\ot\mathbb{P}_{\mathrm{sym}}^k\right)\right]
\\
&= \frac{|\Hsym_N|^2}{|\Hsym_k||\Hsym_{N-k}|}\tr_{\Hsym_{N-k}}\left(H^2_{N-k}\right).
\end{align}
where the second equality stems from \eqnref{Friday}.

To compute $L'_{N,k}$ we follow more or less the same strategy. First, using
\eqnsref{dupa1}{Remik2} we can rewrite it as
\begin{equation}\label{Vader}
L'_{N,k}=\sum_{\vec{o}}\left\{\tr\left[\mathbb{P}_{\mathrm{sym}}^N(H_{N-k}\ot \proj{\vec{o},k})\right]\right\}^2.
\end{equation}
Then, we use the fact in that the full Hilbert space $(\mathbb{C}^d)^{\ot (N-k)}$, $H_{N-k}$ assumes the form given in \eqnref{eq:local hamiltonian}, which gives
\begin{align}\label{Priorat}
\tr\left[\mathbb{P}_{\mathrm{sym}}^N(H_{N-k}\ot \proj{\vec{o},k})\right]&=(N-k)\tr\left[\mathbb{P}_{\mathrm{sym}}^N(h\ot \mathbb{P}_{\mathrm{sym}}^{ N-k-1}\ot \proj{\vec{o},k})\right]\\
&=(N-k)\frac{|\Hsym_N|}{|\Hsym_{k+1}|}\tr\left[\mathbb{P}_{\mathrm{sym}}^{k+1}(h\ot \proj{\vec{o},k})\right],
\end{align}
where the second line follows from \eqnref{Friday}. To compute the remaining trace we expand $h$ in its eigenbasis
as $h=\sum_{n=0}^{d-1}\xi_n\proj{n}$ (where $\xi_i$ are the eigenvalues of $h$), which can also be written using the ``mode representation'' as
\begin{equation}
h=\sum_{\vec{n}}\xi_{\vec{n}}\proj{\vec{n}},
\end{equation}
where $\vec{n}=(i_0,\ldots,i_{d-1})$ is now a $d$-dimensional vector whose
components are such that $n_i=0,1$ and $n_0+\ldots+n_{d-1}=1$. In this representation a number $n=i\in\{0,\ldots,d-1\}$ is represented by a vector $\vec{n}$ whose $i$th component $n_i=1$ and the remaining ones are zero. Using \eqnref{dupa1} one obtains
\begin{equation}
\tr\left[\mathbb{P}_{\mathrm{sym}}^{k+1}(h\ot \proj{\vec{o},k})\right]=\sum_{\vec{n}}\xi_{\vec{n}}
\frac{{{1}\choose{\vec{n}}}{{k}\choose{\vec{o}}}}{{{k+1}\choose{\vec{n}+\vec{o}}}}=\sum_{\vec{n}}\xi_{\vec{n}}
\frac{{{k}\choose{\vec{o}}}}{{{k+1}\choose{\vec{n}+\vec{o}}}},
\end{equation}
where the summation is taken over vectors $\vec{n}$ specified above (there is $d$ such vectors). The second equality straightforwardly stems from the fact that ${{1}\choose{\vec{n}}}=1$. We then exploit the fact that
${{k+1}\choose{\vec{n}+\vec{o}}}=\tfrac{k+1}{o_n+1}{{k}\choose{\vec{o}}}$ and the assumption that $\tr h=0$ to get
\begin{equation}
\tr\left[\mathbb{P}_{\mathrm{sym}}^{k+1}(h\ot \proj{\vec{o},k})\right]=\frac{1}{k+1}\sum_{n=0}^{d-1}\xi_n o_n
=\frac{1}{k+1}\lambda_{\vec{o}}^{(k)},
\end{equation}
where, to recall, $\lambda_{\vec{o}}^{(k)}$ is the eigenvalue of
the $k$-partite Hamiltonian $H_k$ (compare \eqnref{Remik2}).
Combining the above identity with \eqnsref{Priorat}{Vader}, one finds that
\begin{equation}\label{elPrime}
L'_{N,k}=\left(\frac{N-k}{k+1}\frac{|\Hsym_N|}{|\Hsym_{k+1}|}\right)^2
\tr_{\Hsym_k}\left( H_k^2\right).
%
\end{equation}
Plugging \eqnsref{el}{elPrime} into \eqnref{Av1}, one eventually finds that the average QFI is lower-bounded as
\begin{equation}
\Expect_{U \sim \mu\left(\Hsym_N\right)} F\left(\sigma_U,H_{N-k}\right)\geq
2\frac{|\Hsym_N|}{|\Hsym_k|}\frac{|\Hsym_N|\tr \rho_N^2-1}{|\Hsym_N|^2-1}\frac{(N-k)(N+d)}{(d+1)(d+k)}\tr h^2.
\end{equation}
\end{proof}

\begin{rem}
It is worth mentioning that using similar techniques, one can also
provide an upper bound on the average QFI for bosons in the case of particle losses.
To be more precise, in what follows we will derive such a bound for multi-qubit states.
As the QFI is upper bounded by the variance, one has
\begin{equation}
F_Q(\sigma_U,H_{N-k})\leq 4\Delta_{\sigma_{U}}^2H_{N-k}=4\{\tr(\sigma_U H^2_{N-k})-[\tr(\sigma_U H_{N-k})]^2\}\leq 4\tr(\sigma_U H_{N-k}^2).
\end{equation}
Using then the fact that the right-hand side can be rewritten as
%
$\tr(\sigma_U H^2_{N-k})=\tr[\rho (H^2_{N-k}\ot\mathbb{P}^{k}_{\mathrm{sym}})]$
%
and that
\begin{equation}
\int_{\mathrm{SU}(\Hsym_N)}\mathrm{d}\mu(U)U\rho U^\dag =\frac{\mathbb{P}_{\mathrm{sym}}^N}{N+1},
\end{equation}
one obtains
\begin{equation}
\Expect_{U \sim \mu\left(\Hsym_N\right)} F\left(\sigma_U,H_{N-k}\right)\leq \frac{4}{N+1}\tr[\mathbb{P}^N_{\mathrm{sym}}(H^2_{N-k}\ot \mathbb{P}_{\mathrm{sym}}^{ k})].
\end{equation}
With the aid of \eqnsref{Friday}{symmetric trace} we eventually get
\begin{equation}\label{UpperDupper}
\Expect_{U \sim \mu\left(\Hsym_N\right)} F\left(\sigma_U,H_{N-k}\right)\leq \frac{1}{3}(N-k)(N-k+2).
\end{equation}
Notice that for $k=0$ this bound gives $N(N+2)/3$ which differs from the exact value for qubtis by a factor linear in $N$. In general, however, this bound is not very informative because even for significant particle losses as e.g.
$k=\eta N$ with $0<\eta<1$, the right-hand side of \eqnref{UpperDupper} scales quadratically with $N$.
\end{rem}

\subsection{Average FI of random two-mode bosonic states in the interferometric setup}
\label{sec:averagefifortwomodes}
In this part we study the interferometric setup introduced in \secref{sec:phys} and depicted in \figref{fig:interferometer}. 
Recall that the \emph{classical Fisher information} (FI) associated with such a measurement scheme is given by:
\begin{equation}
\FI(\left\{p_{n|\param}(\psi)\right\})
=
\sum_{n=0}^N
\frac{\tr\left(\ii\,[\Pi^{N}_n,\hat J_z] \psi(\param)\right)^2}
{\tr\left(\Pi^{N}_n\psi(\param)\right)} \ ,
\label{eq:FIapp}
\end{equation}
where
by
$\hat J_\alpha\coloneqq\frac{1}{2}\sum_{i=0}^N \sigma_\alpha^{(i)}$
($\alpha=\{x,y,z\}$) we denote the angular momentum operators, 
$\Pi^{N}_n=\hat{B}D_{n}^N \hat{B}^\dagger$, with 
$\hat{B}\coloneqq\exp(-\ii \pi{\hat J}_x/2)$, $D_n^N=\kb{D_n^N}{D_n^N}$, and $\ket{D_n^N}\coloneqq\ket{n,N-n}$, are the projections 
onto the Dicke states propagated through a balanced beam-splitter, and 
$\psi(\varphi)\coloneqq\exp(-\ii \hat J_z\param)\psi\exp(\ii \hat J_z\param)$ 
with $\psi$ some pure state in $\Hsym_N$ with $d=2$ modes. 

Similarly as in \thmref{thm:inter_meas_fixed_phi} of \secref{sec:phys}, after 
fixing $\psi_N$ to be a particular pure state on $\Hsym_N$, we may then define
\begin{equation}\label{FIclas2}
\FI(U,\param) \coloneqq \FI(\left\{p_{n|\param}(U\psi_N U^\dagger)\right\} \ ,
\end{equation}
where $U\in\mathrm{SU}\left(\Hsym_N\right)$ and $\param\in[0,2\pi]$.

\begin{lem}\label{avclassical}
Let $\FI(U,\param)$ be defined as above. Then, the following inequalities hold
\begin{equation}
c_- N^2\leq \Expect_{U \sim \mu\left(\mathcal{\Hsym}_N\right)} \!\!\FI(U,\param) \;\leq\ c_+ N^2+N, 
\end{equation}
where 
\begin{equation}
c_-=\frac{1}{36}-\frac{4}{3\e^5} \approx 0.0244  \ ,\ c_+=-\frac{5}{6}+\frac{3}{\e} \approx 0.270 \ .
\end{equation} 
\end{lem}
\begin{proof}
The main difficulty in the proof comes from the fact that $\FI(U,\param)$ is a complicated, non-linear function of $U$.
Let us first note that by using the relation  $\hat B\,\e^{-\ii {\hat J}_z \param}\hat B^\dagger =\e^{\ii {\hat J}_y \param}$ it is possible to rewrite the FI in \eqnref{eq:FIapp} as
\begin{equation}
\FI(\left\{p_{n|\param}(\psi)\right\})
=
\sum_{n=0}^N
\frac{\tr\!\left(\ii\,[D^{N}_n,\hat J_y] \tilde{\psi}(\param)\right)^2}
{\tr\!\left(D^{N}_n\tilde{\psi}(\param)\right)} \ ,
\end{equation}
where $\tilde{\psi}(\param)= \exp\left(\ii\param\hat{J}_y\right)\psi\exp\left(-\ii\param\hat{J}_y\right)$. Let us introduce the auxiliary notation
\begin{align} 
f_n\left(U,\param\right) &= \left\{\tr\!\left(\ii\,[D^{N}_n,\hat J_y] \exp\left(\ii\param\hat{J}_y\right)U\psi U^\dagger \exp\left(-\ii\param\hat{J}_y\right) \right)\right\}^2 \ , \\
g_n\left(U,\param\right) &= \tr\!\left( D^{N}_n \exp\left(\ii\param\hat{J}_y\right)U\psi U^\dagger \exp\left(-\ii\param\hat{J}_y\right) \right) \ .
\end{align}
Using the above formulas we obtain the compact expression for $\FI(U,\param)$ ,
\begin{equation}\label{simpler}
\FI(U,\param)=\sum_{n=0}^N \frac{f_{n}(U,\param)}{g_{n}(U,\param)}\  .
\end{equation}
In what follows we will make use of the inequality
\begin{equation}\label{simplbound}
f_{n}(U,\param)\leq N^2 g_{n}(U,\param)^2 \ ,
\end{equation}
which follows directly from \eqref{eq:posineq} applied to the considered setting. In order to obtain bounds on the average $\Expect_{U \sim \mu\left(\mathcal{\Hsym}_N\right)}\FI(U,\param)$ we will use the use the following subsets of the $\mathrm{SU}\left(\Hsym_N\right)$,
\begin{align}
\mathcal{G}_{+,\alpha}^n&=\left\{U\in\mathrm{SU}\left(\left.\Hsym_n\right)\right|\ g_n\left(U,\param\right)\geq \alpha    \right\}\ , \\
\mathcal{G}_{-,\alpha}^n&=\left\{U\in\mathrm{SU}\left(\left.\Hsym_n\right)\right|\ g_n\left(U,\param\right)\leq \alpha    \right\}\ ,
\end{align}
where $n=0,\ldots,N$ and $\alpha\in\left[0,1\right]$. Because of the unitary invariance of the Haar measure and the fact that projectors $D^{N}_n$ have rank one the distribution of the random variable $g_n\left(U,\param\right)$ is identical with the distribution of the 
random variable $X\left(V\right)=\tr\left(\psi V \psi V^\dagger\right)$, where $V$ - is Haar distributed unitary on $\mathbb{C}^{N+1}$ and $\psi$ is a pure state on this Hilbert space. The distribution of $X\left(V\right)$ is known (see for instance equation (9) in \cite{Zyczkowski2000}) and is given by
\begin{equation}\label{distrib}
p(X)=N\left(1 - X\right)^{N - 1} \ , \ X\in[0,1] \ .
\end{equation}

\emph{Lower bound.} Let us first derive the lower bound for the average of FI. Consider first the average of a single term in a sum \eqref{simpler}. For $\alpha>0$ we have the following chain of (in)equalities 
\begin{align}
\Expect_{U \sim \mu\left(\mathcal{\Hsym}_N\right)} \frac{f_n(U,\param)}{g_n(U,\param)}&\geq \int_{U\in\mathcal{G}_{-,\alpha}^n}\mathrm{d}\mu(U) \frac{f_n(U,\param)}{\alpha} +  \int_{U\in\mathcal{G}_{+,\alpha}^n}\mathrm{d}\mu(U) \frac{f_n(U,\param)}{g_n(U,\param)}\ \label{f1} \\
&= \frac{1}{\alpha} \int_{U\in\mathrm{SU}\left(\Hsym_N\right)}\mathrm{d}\mu(U) f_n(U,\param) -  \int_{U\in\mathcal{G}_{+,\alpha}^n}\mathrm{d}\mu(U) f_n(U,\param)\frac{(g_n(U,\param)-\alpha)}{g_n(U,\param)\alpha}\ \label{f2} \\ 
&\geq \frac{1}{\alpha} \int_{U\in\mathrm{SU}\left(\Hsym_N\right)}\mathrm{d}\mu(U) f_n(U,\param) - \frac{N^2}{\alpha} \int_{U\in\mathcal{G}_{+,\alpha}^n}\mathrm{d}\mu(U) g_n(U,\param)(g_n(U,\param)-\alpha)\ \label{f3} \\ 
&=  \frac{1}{\alpha}\int_{U\in\mathrm{SU}\left(\Hsym_N\right)}\mathrm{d}\mu(U) f_n(U,\param) - \frac{N^2}{\alpha} \int_{\alpha}^1 dX p(X)  X(X-\alpha) \  \label{f4} \\
&=  \frac{1}{\alpha}\int_{U\in\mathrm{SU}\left(\Hsym_N\right)}\mathrm{d}\mu(U) f_n(U,\param) - \frac{N^2(1-\alpha)^{N+1}(2+\alpha N)}{\alpha(1+N)(2+N)} \ \label{f5}\ .
\end{align}
In the above sequence of (in)equalities  \eqref{f1} follows form the definitions of sets $\mathcal{G}_{\pm,\alpha}^n$,  \eqref{f3} follows from the nonnegativity of $g_n(U,\param)-\alpha$ on $\mathcal{G}_{+,\alpha}^n$ and from \eqref{simplbound}. Equation \eqref{f4} follows form the definition of the random variable $X$ presented in the discussion above \eqref{distrib}. Finally equation \eqref{f5} follows directly form \eqref{distrib}. Summing up over $n$ we obtain the inequality 
\begin{equation} \label{almostfinlow}
\Expect_{U \sim \mu\left(\mathcal{\Hsym}_N\right)} \FI(U,\param) \geq \frac{1}{\alpha}\sum_{n=0}^N \left(\int_{U\in\mathrm{SU}\left(\Hsym_N\right)}\mathrm{d}\mu(U) f_n(U,\param)\right) - \frac{N^2(1-\alpha)^{N+1}(2+\alpha N)}{\alpha(2+N)} \ .
\end{equation}
Using the integration techniques analogous to the ones used in preceding sections it is possible to show that
\begin{equation}
\int_{U\in\mathrm{SU}\left(\Hsym_N\right)}\mathrm{d}\mu(U) f_n(U,\param)=\frac{\tr\left(-\left[D^{N}_n, \hat{J}_y \right]^2\right)}{(N+1)(N+2)} \ .
\end{equation}
Making use of the fact that $\tr\left(\hat{J}_y D^{N}_n\right)=0$ we obtain  
\begin{equation}\label{auxaverage}
\sum_{n=0}^N \left(\int_{U\in\mathrm{SU}\left(\Hsym_N\right)}\mathrm{d}\mu(U) f_n(U,\param)\right)= \sum_{n=0}^N \frac{2\tr\left(D^{N}_n, \hat{J}_{y}^2 \right)}{(N+1)(N+2)} = \frac{2\tr\left(\hat{J}_{y}^2\right)}{(N+1)(N+2)}=\frac{N}{6} \ ,
\end{equation}
In the last equality of \eqref{auxaverage} we have used \eqref{symmetric trace} and the fact that 
$\hat{J}_y$ 
originates in a single particle Hamiltonian satisfying $\tr\left(h^2\right)=\frac{1}{2}$. Plugging \eqref{auxaverage} to \eqref{almostfinlow} we obtain that for all $\alpha>0$ we obtain
\begin{equation}
\Expect_{U \sim \mu\left(\mathcal{\Hsym}_N\right)} \FI(U,\param) \geq \frac{N}{6\alpha} - \frac{N^2(1-\alpha)^{N+1}(2+\alpha N)}{\alpha(2+N)} \ .
\end{equation}
By setting $\alpha=\frac{\Delta}{N}$, where $\Delta$ is a fixed positive parameter, and by using the inequality $(1-\frac{\Delta}{N})^{N+1}\leq\exp\left(-\Delta\right)$  we obtain 
\begin{equation}\label{almostfinn}
\Expect_{U \sim \mu\left(\mathcal{\Hsym}_N\right)} \FI(U,\param) \geq \frac{N^2}{6\Delta} - \frac{N^2 \exp\left(-\Delta\right)\left(2+\Delta\right)}{\Delta} \ .
\end{equation}
Finding the maximal value of right hand side of \eqref{almostfinn} (treated as a function of $\Delta$) is difficult. Numerical investigation shows that the maximal value is obtained very close to $\Delta=6$ which finally gives
\begin{equation}
\Expect_{U \sim \mu\left(\mathcal{\Hsym}_N\right)} \FI(U,\param) \geq c_- N^2 \ ,
\end{equation}
where $c_-=\frac{1}{36}-\frac{4}{3e^5} \approx 0.0244$. 

\emph{Upper bound.} The proof of the upper bound of the average Fisher information is analogous. For $\alpha>0$ we have the following chain of (in)equalities 
\begin{align}
\Expect_{U \sim \mu\left(\mathcal{\Hsym}_N\right)} \frac{f_n(U,\param)}{g_n(U,\param)}&\leq \int_{U\in\mathcal{G}_{+,\alpha}^n}\mathrm{d}\mu(U) \frac{f_n(U,\param)}{\alpha} +  \int_{U\in\mathcal{G}_{-,\alpha}^n}\mathrm{d}\mu(U) \frac{f_n(U,\param)}{g_n(U,\param)}\  \label{g1} \\
&= \frac{1}{\alpha} \int_{U\in\mathrm{SU}\left(\Hsym_N\right)}\mathrm{d}\mu(U) f_n(U,\param) +  \int_{U\in\mathcal{G}_{-,\alpha}^n}\mathrm{d}\mu(U) f_n(U,\param)\frac{(\alpha-g_n(U,\param))}{g_n(U,\param)\alpha}\  \label{g2} \\ 
&\leq \frac{1}{\alpha} \int_{U\in\mathrm{SU}\left(\Hsym_N\right)} f_n(U,\param) + \frac{N^2}{\alpha} \int_{U\in\mathcal{G}_{-,\alpha}^n} g_n(U,\param)(\alpha-g_n(U,\param))\ \label{g3} \\ 
&=  \frac{1}{\alpha}\int_{U\in\mathrm{SU}\left(\Hsym_N\right)}\mathrm{d}\mu(U) f_n(U,\param) + \frac{N^2}{\alpha} \int_{0}^\alpha dX p(X)  X(\alpha-X) \label{g4} \\
&=  \frac{1}{\alpha}\int_{U\in\mathrm{SU}\left(\Hsym_N\right)}\mathrm{d}\mu(U) f_n(U,\param) + \frac{N^2\left(\alpha(2+N)+(1-\alpha)^{N+1}(2+\alpha N)-2\right)}{\alpha(1+N)(2+N)} \ \label{g5}\ .
\end{align}
In the above sequence of (in)equalities  \eqref{g1} follows form the definitions of sets $\mathcal{G}_{\pm,\alpha}^n$,  \eqref{g3} follows from the nonnegativity of $\alpha-g_n(U,\param)$ on $\mathcal{G}_{-,\alpha}^n$ and from \eqref{simplbound}. Equation \eqref{f4} follows form the definition of the random variable $X$ presented in the discussion above \eqref{distrib}. Finally equation \eqref{g5} follows directly form \eqref{distrib}. Summing up over $n$ we obtain the inequality 
\begin{equation} \label{almostfinup}
\Expect_{U \sim \mu\left(\mathcal{\Hsym}_N\right)} \FI(U,\param) \leq \frac{1}{\alpha}\sum_{n=0}^N \left(\int_{U\in\mathrm{SU}\left(\Hsym_N\right)}\mathrm{d}\mu(U) f_n(U,\param)\right) + \frac{N^2\left(\alpha(2+N)+(1-\alpha)^{N+1}(2+\alpha N)-2\right)}{\alpha(2+N)} \ .
\end{equation}
 Let  $\Delta$ be a fixed positive number. By setting $\alpha=\frac{\Delta}{N}$, and by using \eqref{auxaverage} we obtain the upper bound
\begin{equation} \label{almostfinup1}
\Expect_{U \sim \mu\left(\mathcal{\Hsym}_N\right)} \FI(U,\param) \leq \frac{N^2}{6\Delta} + \frac{N^2}{\Delta}\left(\Delta -2 +\left(\Delta+2\right)\exp\left(-\Delta\right) \right)+N\ .
\end{equation}
Finding the minimal value of right hand side of \eqref{almostfinup1} (treated as a function of $\Delta$) is difficult. Numerical investigation shows that the minimal value is obtained very close to $\Delta=1$. Inserting this to  \eqref{almostfinup1} gives
\begin{equation}
\Expect_{U \sim \mu\left(\mathcal{\Hsym}_N\right)} \FI(U,\param) \leq c_+ N^2 +N \ ,
\end{equation}
where $c_+=-\frac{5}{6}+\frac{3}{e} \approx 0.270$. 
\end{proof}

\section{Proofs of main theorems} \label{sec:finproofs}
In this section we use the technical results developed in the preceding parts of the Appendix to prove main theorems form the main manuscript. In the main text we have used, for the sake of simplicity, the $\Theta$ notation that allowed us to hide the presence of complicated constants in the concentration inequalities. In what follows we will present technical versions of these theorems giving explicitly all the relevant constants. Proofs of Theorems \ref{thm:moststatesareuseless},\ref{thm:symmetricstatesareuseful},\ref{thm:symmetriclosses}, and Example \ref{thm:depolensamble} are analogous in a sense that they all relay concentration inequalities \eqref{eq:su concentration} and on
\begin{itemize}
\item Upper bounds on the Lipschitz constants of the relevant functions on $\mathrm{SU}\left(\H\right)$;
\item Bounds or explicit values on the average of these functions on $\mathrm{SU}\left(\H\right)$.
\end{itemize}  
The proof of Theorem \ref{thm:inter_meas_fixed_phi} is slightly more complicated  and relies on the regularity of $\FI(U,\param)$ viewed as a function of the parameter $\param$. 

Let us start with a immediate corollary of \factref{concform} describing the concentration of measure on $\mathrm{SU}\left(\H\right)$. 
\begin{corr}\label{adjustedconcform}
Let $f:\SU\left(\H\right)\longmapsto\mathbb{R}$ be a function on $\SU\left(\H\right)$. Let $D=|\H|$ be the dimension of $\H$. Assume that the function $f$  with the Lipschitz constant $L$ satisfying $L\leq\tilde{L}$ for some nonnegative scalar $\tilde{L}$. 
Assume that the expectation value of $f$ is upper bounded as $\Expect_{U \sim \mu\left(\mathcal{\mathcal{H}}\right)}f\leq F_+$. Then, for every $\epsilon\geq0$ the following large deviation bound holds,
\begin{equation}\label{eq:upconc1}
\Pr_{U\sim \mu\left(\H\right)}\left(f\left(U\right)\geq F_++\epsilon\right)\leq\exp\left(-\frac{D\epsilon^{2}}{4\tilde{L}^{2}}\right) \ .
\end{equation}
Assume that the expectation value of $f$ is lower bounded as $\Expect_{U \sim \mu\left(\mathcal{\mathcal{H}}\right)}f\geq F_-$. Then, for every $\epsilon\geq0$ the following large deviation bound holds,
\begin{equation}\label{eq:lowconc2}
\Pr_{U\sim \mu\left(\H\right)}\left(f\left(U\right)\leq F_- -\epsilon\right)\leq\exp\left(-\frac{D\epsilon^{2}}{4\tilde{L}^{2}}\right) \ .
\end{equation}
\end{corr}

We use \corrref{adjustedconcform} to prove technical versions of Theorems \ref{thm:moststatesareuseless},\ref{thm:symmetricstatesareuseful},\ref{thm:symmetriclosses} and Example \ref{thm:depolensamble} from the main text. 

\begin{thm}[Technical version of Theorem \ref{thm:moststatesareuseless} from the main text] 
\label{thm:moststatesareuselesstech}
Fix a single-particle Hamiltonian $h$, local dimension $d$ and a pure state $\psi_N$ on $\H_N$. 
Let $\QFI^{\mathrm{LU}}(U)\coloneqq\QFI^{\mathrm{LU}}(U\psi_N\,U^\dagger, H)$, then for every $\epsilon\geq0$
  \begin{equation}\label{eq:distinup}
    \Pr_{U \sim \mu\left(\H_N\right)}\left( \QFI^{\mathrm{LU}}(U) \geq 4N\|h \|^2 \left(1+\frac{(N-1)d^2}{\sqrt{d^N}}\right)+\epsilon \right) \leq \exp\left(-\frac{\epsilon^2 d^N}{4096 \|h \|^4 N^4} \right) \ ,
  \end{equation}
	  \begin{equation}\label{eq:distindown}
    \Pr_{U \sim \mu\left(\H_N\right)}\left( \QFI^{\mathrm{LU}}(U) \leq \frac{4N\tr(h^2)d^N}{d(d^N+1)}-\epsilon \right) \leq \exp\left(-\frac{\epsilon^2 d^N}{4096 \|h \|^4 N^4} \right) \ .
  \end{equation}
Setting $\epsilon=2N\|h \|^2 \left(1+\frac{(N-1)d^2}{\sqrt{d^N}}\right)$ and $\epsilon=\frac{2N\tr(h^2)d^N}{d(d^N+1)}$ in \eqref{eq:distinup} and \eqref{eq:distindown} respectively yields Theorem \ref{thm:moststatesareuseless}.
\end{thm}
\begin{proof}
The proof of \thmref{thm:moststatesareuselesstech} follows directly from \corrref{adjustedconcform} and results proved previously. From \lemref{important Lipschitz} and \ref{important Lipschitz opt} one can infer that the Lipschitz constant of $F^{\mathrm{LU}}$ is upper bounded by $\tilde{L}=32\|H\|^2=32N^2 \|h\|^2$. From \eqref{eq:ineq1} we have the upper bound on $\Expect_{U \sim \mu\left(\mathcal{\mathcal{H}}\right)}\QFI^{\mathrm{LU}}$. Using this bound in \eqref{eq:upconc1} gives \eqref{eq:distinup}. The lower bound $\Expect_{U \sim \mu\left(\mathcal{\mathcal{H}}\right)}\QFI^{\mathrm{LU}}$ can be obtained by noting that the unoptimized QFI is a lower bound to its optimized version. Therefore
\begin{equation}\label{eqaux1}
\Expect_{U \sim \mu\left(\mathcal{\mathcal{H}}\right)}\QFI^{\mathrm{LU}}\geq \Expect_{U \sim \mu\left(\mathcal{\mathcal{H}}\right)} \QFI(U\psi_N\,U^\dagger, H)= \frac{4N\tr(h^2)d^N}{d(d^N+1)} \ ,
\end{equation}
where in the last equality we used \eqref{eq:distingav}. Plugging \eqref{eqaux1} in \eqref{eq:lowconc2} yields \eqref{eq:distindown}.
\end{proof}

\begin{thm}[Technical version of Theorem \ref{thm:symmetricstatesareuseful} from the main text] 
\label{thm:symmetricstatesareusefultech}
Fix a single-particle Hamiltonian $h$, local dimension $d$ and a state $\sigma_N$ from the symmetric subspace $\Hsym_N$  with eigenvalues $\{p_j\}_j$. Let $\sigma_{\mathrm{mix}}$ be the maximally mixed state on $\Hsym_N$. Let $\QFI(U) \!\coloneqq\! \QFI(U\,\sigma_N\,U^\dagger, H)$, then for every $\epsilon\geq0$
  \begin{equation}\label{isospectralboundex}
      \Pr_{U \sim \mu\left(\Hsym_N\right)} \left( \QFI\left(U\right) \leq \dB\left(\sigma_N,\sigma_{\mathrm{mix}}\right)^2 \frac{2N(N+d)\tr(h^2)}{d(d+1)}\frac{|\Hsym_N|^2}{|\Hsym_N|^2-1} -\epsilon \right)\leq \exp\left( - \frac{\epsilon^2 |\Hsym_N|}{4096 C\|h \|^4 N^4}) \right) \ ,
  \end{equation}
where $|\Hsym_N|=\binom{N+d-1}{N}$ and $C=\min\left\{1,8\dB\left(\sigma_N,\sigma_{\mathrm{mix}}\right) \right\}$.
Setting $\epsilon=\dB\left(\sigma_N,\sigma_{\mathrm{mix}}\right)^2 \frac{N(N+d)\tr(h^2)}{d(d+1)}\frac{|\Hsym_N|^2}{|\Hsym_N|^2-1}$ in \eqref{isospectralboundex} yields Theorem \ref{thm:symmetricstatesareuseful}.
\end{thm}
\begin{proof}
The proof is analogous to the proof of \thmref{thm:moststatesareuselesstech}.  From \lemref{important Lipschitz} we infer that the Lipschitz constant of $\QFI(U)$ is upper bounded by $\tilde{L}=32\|H\|^2\min\left\{1,8\dB\left(\sigma_N,\sigma_{\mathrm{mix}}\right) \right\}=32N^2 \|h\|^2\min\left\{1,8\dB\left(\sigma_N,\sigma_{\mathrm{mix}}\right) \right\}$. From equation \eqref{eq:symav} in \lemref{averagesrem} we get
\begin{equation}\label{eq:distingav1}
\Expect_{U \sim \mu\left(\Hsym_N\right)}\QFI(U)=\frac{4N(N+d)\tr(h^2)}{d(d+1)}\frac{|\Hsym_N|}{|\Hsym_N|+1}\Lambda(\{p_j\}_j) \ .
\end{equation} 
Using the inequality \eqref{eq:important bound} and the Fuch-van de Graaf inequality \cite{Fuchs1999},
$1-\mathcal{F}^2(\sigma_N,\sigma_{\mathrm{mix}})\leq\frac{1}{2}\dB\left(\sigma_N,\sigma_{\mathrm{mix}}\right)$, we obtain 
  \begin{equation}\label{eq:funnyineq1}
    \Lambda(\{p_j\}_j) \geq \frac{|\Hsym_N|}{2(|\Hsym_N|-1)} \, \dB\left(\sigma_N,\sigma_{\mathrm{mix}}\right)^2  \ .
  \end{equation} 
Inserting this inequality into \eqref{eq:distingav1} gives
\begin{equation}\label{eq:funnyineq2}
\Expect_{U \sim \mu\left(\mathcal{H}_N\right)}\QFI(U)\geq\dB\left(\sigma_N,\sigma_{\mathrm{mix}}\right)^2 \frac{2N(N+d)\tr(h^2)}{d(d+1)}\frac{|\Hsym_N|^2}{|\Hsym_N|^2-1} \ ,
\end{equation}
which together with the bound on the Lipschitz constant of $\QFI(U)$ and \corrref{adjustedconcform} allows us to conclude \eqref{isospectralboundex}.
\end{proof}

\begin{exa}[Technical version of Example \ref{thm:depolensamble} from the main manuscript]
\label{thm:depolensambletech} 
Fix a local dimension $d$, single particle Hamiltonian $h$, and $p\in\left[0,1\right]$.
Let $\psi_N$ be a pure state on $\Hsym_N$ and set
\begin{equation}\label{depol1}
 \sigma_N(p) = (1-p)\,\psi_N + p\, \sigma_{\mathrm{mix}} \ .
\end{equation}
Let $F_p\left(U\right) \coloneqq \QFI(U\,\sigma_N(p)\,U^\dagger, H)$, then for every $\epsilon>0$
\begin{equation}\label{depolarconc1}
 \Pr_{U \sim \mu\left(\Hsym_N\right)}\left( |\QFI_p(U) - \Expect_{U \sim \mu\left(\mathcal{\Hsym}_N\right)} F_p| \geq \epsilon \Expect_{U \sim \mu\left(\mathcal{\Hsym}_N\right)} F_p \right) \leq 2\exp\left(-\frac{\epsilon^2 \tr(h^2)^2 (N+d)^2 |\Hsym_N|^2}{64\|h\|^4\left(d(d+1)N(1+|\Hsym_N|)\right)^2}|\Hsym_N|\right) \ ,
\end{equation}
where $|\Hsym_N|=\binom{N+d-1}{N}$ and
\begin{equation}
\label{eq:lambda_dep1}
 \Expect_{U \sim \mu\left(\mathcal{\Hsym}_N\right)} F_p =\frac{4N(N+d)\tr(h^2)}{d(d+1)}\frac{|\Hsym_N|}{|\Hsym_N|+1} \frac{(1-p)^2}{(1-p+2\,p/|\Hsym_N|)} \ .
\end{equation}
Equation \eqref{eq:lambda_dep1} implies \exaref{thm:depolensamble} as for fixed local dimension $d$ we have $|\Hsym_N|\in\Theta\left(N^{d-1}\right)$.
\end{exa}
\begin{proof}[Sketch of the proof]
The proof of \exaref{thm:depolensambletech} parallels  proofs of \thmref{thm:moststatesareuselesstech} and \ref{thm:symmetricstatesareusefultech} and relies on  \factref{concform}. The bound of the Lipschitz constant of $\QFI_{p}(U)$ is provided by \lemref{important Lipschitz opt}. The expression the average of $\QFI_{p}(U)$ is given in \lemref{averagesrem}. The inequality \eqref{depolarconc1} follows directly from concentration inequalities from \factref{concform} by setting $\epsilon=\tilde{\epsilon} \Expect_{U \sim \mu\left(\mathcal{\Hsym}_N\right)} F_p$. 
\end{proof}

\begin{thm}[Technical version of Theorem 3 from the main manuscript] \label{thm:symmetriclossestech}
  Fix a single particle Hamiltonian $h$, local dimension $d$, nonngative integer $k$ and a state $\sigma_N$ on $\Hsym_N$ with eigenvalues $\{p_j\}_j$.  Let $\sigma_{\mathrm{mix}}$ be the maximally mixed state on $\Hsym_N$.
  Let $\QFI_k\left(U\right)\coloneqq\QFI(\tr_k\left(U\,\sigma_N\,U^\dagger\right), H_{N-k})$, then for every $\epsilon\geq0$
  \begin{equation}\label{losspar3}
   \Pr_{U \sim \mu\left(\Hsym_N\right)}\left( \QFI_k\left(U\right) \leq 2\frac{(N-k)(N+d)}{(d+1)(d+k)}\cdot\frac{|\Hsym_N|(|\Hsym_N|\tr \rho_N^2-1)\tr h^2}{|\Hsym_k|(|\Hsym_N|^2-1)}-\epsilon\right) \leq \exp\left( - \frac{\epsilon^2 |\Hsym_N|}{4096 C\|h \|^4 (N-k)^4}\right)\ \ ,
  \end{equation}
	where $|\Hsym_N|=\binom{N+d-1}{N}$ and $C=\min\left\{1,8\dB\left(\sigma_N,\sigma_{\mathrm{mix}}\right) \right\}$. Setting $\epsilon=\frac{(N-k)(N+d)}{(d+1)(d+k)}\cdot\frac{|\Hsym_N|(|\Hsym_N|\tr \rho_N^2-1)\tr h^2}{|\Hsym_k|(|\Hsym_N|^2-1)}$ in \eqref{losspar3} yields Theorem 3.
\end{thm}
\begin{proof}[Sketch of the proof]
The proof of \thmref{thm:symmetriclossestech} parallels  proofs of \thmref{thm:moststatesareuselesstech} and \ref{thm:symmetricstatesareusefultech} and relies on \corrref{adjustedconcform}. The bound of the Lipschitz constant of $\QFI_{k}(U)$ is provided by \lemref{important Lipschitz kloss}. The lower bound for the average of $\QFI_{k}(U)$ is given in \lemref{boundavloss}. 
\end{proof}

\begin{thm}[Technical version of Theorem \ref{thm:inter_meas_fixed_phi} from the main manuscript]
\label{thm:inter_meas_fixed_phi1}
Let $\psi_N$ be a fixed pure state on $\Hsym_N$ with $d=2$ bosonic modes. Let $p_{n|\param}(U\psi_NU^\dagger)$ the probability to obtain outcome $n$ in the interferometric scheme defined in Section~\ref{sec:phys}, given that the value of the unknown phase parameter is $\param$ and the input state was $U\psi_NU^\dagger$ (see also \eqref{eq:p_n}). Let $\FI(U,\param) \coloneqq \FI(\left\{p_{n|\param}(U\psi_N U^\dagger)\right\})$ be the corresponding FI according to \eqref{eq:FI_p_n} (or \eqref{eq:FIapp}). Then, for every $\epsilon\geq0$ and every $\param\in[0,2\pi]$ we have
\begin{align}
\Pr_{U \sim \mu\left(\Hsym_N\right)}\left(  \FI(U,\param) \leq \Expect_{U \sim \mu\left(\mathcal{\Hsym}_N\right)} \FI(U,\param)-\epsilon   \right) &\leq \mathrm{exp}\left(-\frac{\epsilon^2}{144N^4 } (N+1) \right)\ , \label{eq:clFishpup} \\ ,
\Pr_{U \sim \mu\left(\Hsym_N\right)}\left(  \FI(U,\param) \geq \Expect_{U \sim \mu\left(\mathcal{\Hsym}_N\right)} \FI(U,\param)+\epsilon   \right) &\leq \mathrm{exp}\left(-\frac{\epsilon^2}{144N^4 } (N+1) \right)\ \label{eq:clFishpdown} .  
\end{align}
In the equations above $\Expect_{U \sim \mu\left(\mathcal{\Hsym}_N\right)} \FI(U,\param)$ satisfies inequalities
\begin{equation}\label{inequal}
c_- N^2\leq \Expect_{U \sim \mu\left(\mathcal{\Hsym}_N\right)}\FI(U,\param) \leq\ c_+ N^2+N\ , 
\end{equation}
where
\begin{equation}
c_-=\frac{1}{36}-\frac{4}{3e^5} \approx 0.0244  \ ,\ c_+=-\frac{5}{6}+\frac{3}{e} \approx 0.270 \ .
\end{equation}
Moreover, we have the following inequality 
\begin{equation}\label{clmainres}
\Pr_{U \sim \mu\left(\Hsym_N\right)}\left( \exists \param\in[0,2\pi]\   \FI(U,\param) \leq \frac{c_-}{4} N^2   \right)\leq \left\lceil  \frac{12 \pi N}{c_-} \right\rceil\mathrm{exp}\left(-\frac{c_{-}^2}{576} (N+1) \right)\ , 
\end{equation}
where $\left\lceil x \right\rceil$ stands for the smallest integer not less than $x$.
Equation \eqref{clmainres} yields exactly \eqref{eq:classfishtyp} from Theorem \ref{thm:inter_meas_fixed_phi}.
\end{thm}
\begin{proof}
Equations \eqref{eq:clFishpup} and \eqref{eq:clFishpdown} follow directly from \factref{concform} and the bounds of the Lipschitz constant of $\FI(U,\param)$, treated as a function of $U$ (for fixed $\param$), given in \lemref{classicalFishlip}. From \lemref{classicalFishlip} it follows that the Lipschitz constant of $\FI(U,\param)$ is bounded above as
\begin{equation}
L\leq24\|H\|^2=24\|\hat{J}_z\|^2=6N^2 \ .
\end{equation}
Inequalities from \eqref{inequal} follow from \lemref{avclassical}. The nontrivial part of the proof is the justification of \eqref{clmainres}. Let us first introduce the discretization of the interval $[0,2\pi]$ by $M$ equally spaced points:
\begin{equation}\label{eq:disctrete}
\param_i = (i-1)\frac{2\pi}{M} \ ,\ i=1,\ldots,M \ .
\end{equation}
Moreover, let us notice that from \eqref{eq:laststep} it follows that $\FI(U,\param)$ is Lipschitz continuous for fixed $U$ and varying $\param$:
\begin{equation}\label{derivbound}
\left|\left.\frac{d}{d\param}\right.\FI(U,\param)\right| \leq 24 \|H\|^3 = 3 N^3 \ ,
\end{equation}
where in \eqref{eq:laststep} we set $X=H=\hat{J}_z$. From \eqref{derivbound} it follows that for fixed $U\in\mathrm{SU}\left(\Hsym_N\right)$ and for $\param,\tilde{\param}\in[0,2\pi]$ we have
\begin{equation}\label{lipphase}
\left|\FI(U,\param)-\FI(U,\tilde{\param})\right|\leq 3N^3 |\param-\tilde{\param}| \ . 
\end{equation}
When the points in the discretization \eqref{eq:disctrete} are separated by $\Delta =\frac{2\pi}{M}$, the distance on any $\param\in[0,2\pi]$ to closest $\param_i$ the does not exceed $\Delta'=\frac{\Delta}{2}=\frac{\pi}{M}$. Using the union bound, equation \eqref{eq:clFishpup} and the lower bound in equation \eqref{inequal} we obtain 
\begin{equation}\label{ineqdiscretization}
\Pr_{U \sim \mu\left(\Hsym_N\right)}\left(\exists \param_i\  \FI(U,\param_i) \leq c_- N^2 -\epsilon   \right) \leq M \mathrm{exp}\left(-\frac{\epsilon^2}{144N^4 } (N+1) \right)\ .
\end{equation}
Using \eqref{lipphase} and the discussion following it we obtain 
\begin{equation}\label{ineqdiscretization1}
\Pr_{U \sim \mu\left(\Hsym_N\right)}\left(\exists \param\  \FI(U,\param) \leq c_- N^2 - 3N^3 \Delta' -\epsilon   \right) \leq M \mathrm{exp}\left(-\frac{\epsilon^2}{144N^4 } (N+1) \right)\ .
\end{equation}
Now by setting in the above equation $M=\left\lceil  \frac{12 \pi N}{c_-} \right\rceil$ (this is the smallest integer $M$  such that $3N^3 \Delta'\leq \frac{c_i}{4} N^2$)  and $\epsilon=\frac{c_-}{2}N^2$ we obtain \eqref{clmainres}.

\end{proof}

\section{Partial-trace and beam-splitter models of particle losses} \label{sec:partloss}

In this section we prove the equivalence of the beam-splitter model of particle losses and the 
operation of taking partial trace over the constituent particles in the system of $N$ bosons in $d=2$ modes.
A general pure state $\psi_N$ of $N$ bosons in two modes $a$ and $b$ can be written as
\begin{equation}
\left|\psi_{N}\right\rangle =\sum_{n=0}^{N}  \alpha_n \left|n,N-n\right\rangle  =\sum_{n=0}^{N}\alpha_{n}\left|D_{n}^{N}\right\rangle
\label{eq:state_m}
\end{equation}
with the complex coefficients $\{\alpha_n\}_{n=0}^N$ satisfying $\sum_{n=0}^N |\alpha_n|^2=1$. 
Each Dicke state $\left|D_{n}^{N}\right\rangle$ can be written in the basis of \emph{particle basis} $\ket{\phantom{x}}_{\mathsf{p}}$ as 
\begin{equation}
\left|D_{n}^{N}\right\rangle  =\frac{1}{\sqrt{\binom{N}{n}}}\sum_{\mathbf{x} \in \{0,1\}^{N}}\delta_{x,n}\left|\mathbf{x}\right\rangle _{\mathsf{p}}\ ,
\label{eq:dicke}
\end{equation}
where $x \coloneqq |\mathbf{x}| \coloneqq \sum_{i}x_{i}$ denotes the Hamming weight of any binary string $\mathbf{x}=[x_{1},\dots,x_{N}]$,
whose consecutive entries specify the state of each qubit. 
As a result, we may write a general bosonic pure state \eref{eq:state_m} in the particle basis as 
\begin{equation}
|\psi_{N}\rangle =\sum_{\mathbf{x} \in \{0,1\}^{N}} 
c_{\mathbf{x}} \, |
\mathbf{x}\rangle_{\mathsf{p}}=\sum_{\mathbf{x} \in \{0,1\}^{N}}c_{\mathbf{x}} \, |x_{1}\rangle \left|x_{2}\right\rangle \dots|x_{N}\rangle ,
\label{eq:state_p}
\end{equation}
with the coefficients $c_{\mathbf{x}}$ then given by  $c_{\mathbf{x}} = \frac{1}{\sqrt{\binom{N}{n}}}\sum_{n=0}^{N}\alpha_{n}\delta_{x,n}$.

\subsection{Tracing-out $k$ particles}
Let us define notation in which we may split any binary string, $\mathbf{x}$
(describing $N$ qubits), into two strings, $\mathbf{x}'$ and $\mathbf{u}$
(describing first $N-k$ and last $k$ qubits respectively), so that
$\mathbf{x}=[\mathbf{x}',\mathbf{u}]=[x'_{1},\dots,x'_{N-k},u_{1},\dots,u_{k}]$.
Then, we may generally write the bosonic state \eref{eq:state_p} in the particle
basis after tracing-out the last $k$ qubits as
\begin{eqnarray}
\varrho_{N-k}^{\tr}
& \coloneqq & 
\tr_{k}\!\left\{ \psi_{N}\right\} \label{eq:state_tr_k}\\
 & = & \tr_{k}\!\left\{ \sum_{\mathbf{x},\mathbf{y}=0^{N}}^{1^{N}}c_{\mathbf{x}}c_{\mathbf{y}}^{\star}\left|\mathbf{x}\right\rangle _{\mathsf{p}}\left\langle \mathbf{y}\right|\right\} =\sum_{\mathbf{x},\mathbf{y}=0^{N}}^{1^{N}}c_{\mathbf{x}}c_{\mathbf{y}}^{\star}\;\tr_{k}\!\left\{ \left|\mathbf{x}\right\rangle _{\mathsf{p}}\!\left\langle \mathbf{y}\right|\right\} \\
 & = &  \sum_{\mathbf{x}',\mathbf{x}'=0^{N-k}}^{1^{N-k}}\left[\varrho_{N-k}^{\tr}\right]_{\mathbf{x}'\mathbf{y}'}\;\left|\mathbf{x}'\right\rangle _{\mathsf{p}}\!\left\langle \mathbf{y}'\right|,\label{eq:dupa}
\end{eqnarray}
where the above matrix entries of $\varrho_{N-k}^{\tr}$ are given by
\begin{eqnarray}
\left[\varrho_{N-k}^{\tr}\right]_{\mathbf{x}'\mathbf{y}'} & = & \sum_{\mathbf{u},\mathbf{w}=0^{k}}^{1^{k}}c_{[\mathbf{x}',\mathbf{u}]}c_{[\mathbf{y}',\mathbf{w}]}^{\star}\,\delta_{\mathbf{u}\mathbf{w}}=\sum_{\mathbf{u}=0^{k}}^{1^{k}}c_{[\mathbf{x}',\mathbf{u}]}c_{[\mathbf{y}',\mathbf{u}]}^{\star}\\
 & = & \sum_{\mathbf{u}=0^{k}}^{1^{k}}\left(\sum_{n=0}^{N}\frac{\alpha_{n}\delta_{x'+u,n}}{\sqrt{\binom{N}{n}}}\right)\left(\sum_{m=0}^{N}\frac{\alpha_{m}^{\star}\delta_{y'+u,m}}{\sqrt{\binom{N}{m}}}\right)=\sum_{u=0}^{k}\binom{k}{u}\frac{\alpha_{x'+u}\alpha_{y'+u}^{\star}}{\sqrt{\binom{N}{x'+u}}\sqrt{\binom{N}{y'+u}}}.\label{eq:tr_matrix_entry_p}
\end{eqnarray}
In the mode basis we may equivalently write
\begin{equation}
\varrho_{N-k}^{\tr}=\sum_{n,m=0}^{N-k}\left[\varrho_{N-k}^{\tr}\right]_{nm}\left|D_{n}^{N-k}\right\rangle \left\langle D_{m}^{N-k}\right|
\end{equation}
and with the help of \eqnsref{eq:dicke}{eq:dupa} explicitly evaluate the corresponding density-matrix entries:
\begin{eqnarray}
\left[\varrho_{N-k}^{\tr}\right]_{nm} & = & \sum_{\mathbf{x}',\mathbf{y}'=0^{N-k}}^{1^{N-k}}\left[\varrho_{N-k}^{\tr}\right]_{\mathbf{x}'\mathbf{y}'}\;\left\langle D_{m}^{N-k}\right|\left|\mathbf{x}'\right\rangle _{\mathsf{p}}\!\left\langle \mathbf{y}'\right|\left|D_{n}^{N-k}\right\rangle \\
 & = & \sum_{\mathbf{x}',\mathbf{y}'=0^{N-k}}^{1^{N-k}}\left[\varrho_{N-k}^{\tr}\right]_{\mathbf{x}'\mathbf{y}'}\;\frac{\delta_{x',m}}{\sqrt{\binom{N-k}{m}}}\frac{\delta_{y',n}}{\sqrt{\binom{N-k}{n}}}\\
 & = & \sum_{x',y'=0}^{N-k}\binom{N-k}{x'}\binom{N-k}{y'}\sum_{u=0}^{k}\binom{k}{u}\frac{\alpha_{x'+u}\alpha_{y'+u}^{\star}}{\sqrt{\binom{N}{x'+u}}\sqrt{\binom{N}{y'+u}}}\;\frac{\delta_{x',m}}{\sqrt{\binom{N-k}{m}}}\frac{\delta_{y',n}}{\sqrt{\binom{N-k}{n}}}\\
 & = & \sum_{u=0}^{k}\alpha_{m+u}\alpha_{n+u}^{\star}\;\binom{k}{u}\,\sqrt{\frac{\binom{N-k}{m}\binom{N-k}{n}}{\binom{N}{m+u}\binom{N}{n+u}}}.\label{eq:tr_matrix_entry_m}
\end{eqnarray}

\subsection{Beam-splitter model of mode-asymmetric particle losses}
In quantum optics, photonic losses are modelled by adding
fictitious beam-splitters (BSs) of fixed transmittance into the light
transmission modes \cite{Bachor2004}. In this way, by impinging a
vacuum state on the other input port of any such BS and tracing out
its unobserved output port, one obtains a model depicting loss of photon. In case of the two-mode
$N$-photon bosonic state \eref{eq:state_m}, after fixing the transmissivity 
of the fictitious BS introduced in mode $a$ ($b$) to $\eta_a$ ($\eta_b$), 
the density matrix describing then the observed modes generally reads \cite{Demkowicz2015}:
\begin{eqnarray}
\varrho^\t{BS}_N & \coloneqq & \Lambda^\t{BS}_{\eta_a,\eta_b}[\psi_N]\\
& = & \sum_{l_{a}=0}^{N}\sum_{l_{b}=0}^{N-l_{a}}p_{l_{a},l_{b}}\,\left|\xi_{l_{a},l_{b}}\right\rangle _{\mathsf{m}}\!\left\langle \xi_{l_{a},l_{b}}\right|,
\label{eq:rho_BSloss}
\end{eqnarray}
where $\Lambda^\t{BS}_{\eta_a,\eta_b}$ is the effective quantum channel
representing the action of fictitious BSs in the two modes, while indices 
$l_{a}$ and $l_{b}$ denote the number of photons lost in modes $a$ and $b$ respectively. The states 
\begin{equation}
\left|\xi_{l_{a},l_{b}}\right\rangle _{\mathsf{m}}\coloneqq\frac{1}{\sqrt{p_{l_{a},l_{b}}}}\sum_{n=l_{a}}^{N-l_{b}}\,\alpha_{n}\,\sqrt{b_{n}^{(l_{a},l_{b})}}\left|n-l_{a},N-n-l_{b}\right\rangle 
\end{equation}
are generally non-orthogonal and their coefficients contain generalised
binomial factors:
\begin{equation}
b_{n}^{(l_{a},l_{b})}\coloneqq\binom{n}{l_{a}}\,\eta_{a}^{n-l_{a}}\left(1-\eta_{a}\right)^{l_{a}}\;\binom{N-n}{l_{b}}\,\eta_{b}^{N-n-l_{b}}\left(1-\eta_{b}\right)^{l_{b}}.
\end{equation}
The probability of losing $l_{a}$ and $l_{b}$ photons in modes $a$ and $b$ respectively then reads: 
\begin{equation}
p_{l_{a},l_{b}}=\sum_{n=l_{a}}^{N-l_{b}}\left|\alpha_{n}\right|^{2}b_{n}^{(l_{a},l_{b})}.
\end{equation}
On ther hand, after reindexing \eqnref{eq:rho_BSloss} by $l$---the
total number of photons lost in both modes---the output two-mode mixed state 
may be equivalently rewritten as 
\begin{equation}
\varrho^\t{BS}_N =\bigoplus_{l=0}^{N}\;p_{l}\;\varrho_{N,l}^\t{BS},
\label{eq:rho_BSloss_directsum}
\end{equation}
where
\begin{eqnarray}
\varrho_{N,l}^\t{BS} & \coloneqq & \frac{1}{p_{l}}\sum_{l_{a}=0}^{l}p_{l_{a},l-l_{a}}\,\left|\xi_{l_{a},l-l_{a}}\right\rangle _{\mathsf{m}}\!\left\langle \xi_{l_{a},l-l_{a}}\right|\\
 & = & \frac{1}{p_{l}}\sum_{l_{a}=0}^{l}\sum_{n,m=l_{a}}^{N-l+l_{a}}\alpha_{n}\alpha_{m}^{\star}\,\sqrt{b_{n}^{(l_{a},l-l_{a})}b_{m}^{(l_{a},l-l_{a})}}\;\left|n-l_{a},N-n-l+l_{a}\right\rangle _{\mathsf{m}}\!\left\langle m-l_{a},N-m-l+l_{a}\right|\label{eq:state_N-l}
\end{eqnarray}
belong to orthogonal subspaces and represent the state after loss
of $l$ photons, what may occur with probability:
\begin{equation}
p_{l}=\sum_{l_{a}=0}^{l}p_{l_{a},l-l_{a}}=\sum_{l_{a}=0}^{l}\sum_{n=l_{a}}^{N-l+l_{a}}\left|\alpha_{n}\right|^{2}b_{n}^{(l_{a},l-l_{a})}=\sum_{n=0}^{N}\left|\alpha_{n}\right|^{2}\sum_{l_{a}=\max\!\left\{ 0,n-N+l\right\} }^{\min\!\left\{ l,n\right\} }b_{n}^{(l_{a},l-l_{a})}.\label{eq:p_l_gen}
\end{equation}

\subsection{Equivalence of the partial-trace and beam-splitter models in case of equal losses in the two modes}
\begin{lem}
\label{lem:bal-bs_tr_equiv}
For equal photonic losses in both modes, $\eta\coloneqq\eta_{a}\!=\!\eta_{b}$, the fictitious BS model is equivalent to tracing-out
$k$ particles with $k$ distributed according to a binomial distribution, i.e.,
\begin{equation}
\forall_{\psi_n\in\mathcal{S}_N}:
\quad
\Lambda^\t{BS}_{\eta,\eta}[\psi_N]
\;=\;
\bigoplus_{k=0}^{N}\;p_{k}\;\tr_{k}\!\left\{ \psi_{N}\right\} 
\qquad\t{with}\qquad
p_k=\binom{N}{k}\,\eta^{N-k}\left(1-\eta\right)^{k}.
\label{eq:milosc}
\end{equation}
\end{lem}
\begin{proof}
In case of mode-symmetric losses, $\eta\coloneqq\eta_{a}\!=\!\eta_{b}$,
the overall probability of losing $l$ photons becomes 
independent of the state $\psi_N$ (i.e., its coefficients $\alpha_{n}$ of \eqnref{eq:state_m}),
as \eqnref{eq:p_l_gen} then simplifies to 
\begin{equation}
p_{l}=\sum_{n=0}^{N}\left|\alpha_{n}\right|^{2}\!\!\!\sum_{l_{a}=\max\!\left\{ 0,n-N+l\right\} }^{\min\!\left\{ l,n\right\} }\!\binom{n}{l_{a}}\binom{N-n}{l-l_{a}}\;\eta^{N-l}\left(1-\eta\right)^{l}\;=\;\binom{N}{l}\,\eta^{N-l}\left(1-\eta\right)^{l}.\label{eq:p_l_eq}
\end{equation}
Furthermore, the state \eref{eq:state_N-l} in each orthogonal subspace indexed by $l$ takes then a simpler form
\begin{eqnarray}
\varrho_{N,l}^\t{BS} & = & \frac{1}{p_{l}}\sum_{l_{a}=0}^{l}\sum_{n,m=0}^{N-l}\alpha_{n+l_{a}}\alpha_{m+l_{a}}^{\star}\,\sqrt{b_{n+l_{a}}^{(l_{a},l-l_{a})}}\sqrt{b_{m+l_{a}}^{(l_{a},l-l_{a})}}\;\left|n,N-l-n\right\rangle _{\mathsf{m}}\!\left\langle m,N-l-m\right|\\
 & = & \sum_{n,m=0}^{N-l}\left[\varrho_{N,l}^\t{BS}\right]_{nm}\;\left|D_{n}^{N-l}\right\rangle \left\langle D_{m}^{N-l}\right|,
\end{eqnarray}
where we have shifted the indices $n\to n+l_{a}$ and $m\to m+l_{a}$
to explicitly rewrite the state in the Dicke basis, in which its matrix entries then read
\begin{eqnarray}
\left[\varrho_{N,l}^\t{BS}\right]_{nm} & = & \frac{1}{p_{l}}\sum_{l_{a}=0}^{l}\alpha_{n+l_{a}}\alpha_{m+l_{a}}^{\star}\,\sqrt{b_{n+l_{a}}^{(l_{a},l-l_{a})}}\sqrt{b_{m+l_{a}}^{(l_{a},l-l_{a})}}\\
 & = & \frac{1}{p_{l}}\sum_{l_{a}=0}^{l}\alpha_{n+l_{a}}\alpha_{m+l_{a}}^{\star}\,\eta^{N-l}\left(1-\eta\right)^{l}\sqrt{\binom{n+l_{a}}{l_{a}}\binom{N-n-l_{a}}{l-l_{a}}\binom{m+l_{a}}{l_{a}}\binom{N-m-l_{a}}{l-l_{a}}}\\
 & = & \sum_{l_{a}=0}^{l}\alpha_{n+l_{a}}\alpha_{m+l_{a}}^{\star}\sqrt{\frac{\binom{n+l_{a}}{l_{a}}\binom{N-n-l_{a}}{l-l_{a}}}{\binom{N}{l}}}\sqrt{\frac{\binom{m+l_{a}}{l_{a}}\binom{N-m-l_{a}}{l-l_{a}}}{\binom{N}{l}}}.
\end{eqnarray}
However, using the $\binom{n}{m}\binom{m}{k}=\binom{n}{k}\binom{n-k}{m-k}$
and $\binom{n}{k}=\binom{n}{n-k}$ we get
\begin{equation}
\sqrt{\frac{\binom{n+l_{a}}{l_{a}}\binom{N-n-l_{a}}{l-l_{a}}}{\binom{N}{l}}}=\sqrt{\binom{l}{l_{a}}}\sqrt{\frac{\binom{N-l}{n}}{\binom{N}{n+l_{a}}}} \ .
\end{equation}
This allows us to finally write the matrix entries specified in the Dicke-basis
as
\begin{equation}
\left[\varrho_{N,l}^\t{BS}\right]_{nm}=\sum_{l_{a}=0}^{l}\alpha_{n+l_{a}}\alpha_{m+l_{a}}^{\star}\binom{l}{l_{a}}\sqrt{\frac{\binom{N-l}{n}}{\binom{N}{n+l_{a}}}}\sqrt{\frac{\binom{N-l}{m}}{\binom{N}{m+l_{a}}}}.\label{eq:loss_matrix_entry_m}
\end{equation}
Comparing the above expression with \eqnref{eq:tr_matrix_entry_m}
and relabelling the indices $l\rightarrow k$ and $l_{a}\rightarrow u$, 
one observes that independently of $\psi_N$ indeed $\varrho_{N,l}^\t{BS} = \varrho_{N-l}^\t{tr} = \tr_{l}\!\left\{ \psi_{N}\right\}$.
Hence, \eqnref{eq:rho_BSloss_directsum} yields \eqnref{eq:milosc} with binomially distributed $p_l$ according to \eqnref{eq:p_l_eq}.
\end{proof}

\end{document}